\documentclass[11pt]{article}
\usepackage[utf8]{inputenc}
\usepackage{style}

\usepackage[dvipsnames]{xcolor}
\usepackage{todonotes}

\usepackage{libertine}

\newcounter{commentcount}

\title{Spectral Certificates and Non-commutative Sum-of-Squares\\ Lower Bounds for Hamiltonians}

\author{
\begin{tabular}[h!]{cc}
   \FormatAuthor{Nicholas Kocurek}{\footnotesize nichok6@cs.washington.edu}{University of Washington}
   \FormatAuthor{Chinmay Nirkhe}{\footnotesize nirkhe@cs.washington.edu}{University of Washington}
\end{tabular}
}

\date{}

\begin{document}
\maketitle
\allowdisplaybreaks

\begin{abstract}
    A central question in quantum many-body physics is estimating the ground energy of a $k$-local Hamiltonian system. In this work, we present a spectral technique for certifying a lower bound on the ground energy of a random $n$-qubit Hamiltonian system defined as the sum of signed $k$-local Pauli operators. In particular, we prove that for any constant $\ell$, there exists an efficiently computable length $n^{O(\ell)}$ certificate that is \emph{always} a lower bound on the ground energy with the promise that---with high probability over the random Hamiltonian distribution---the certificate value is an $\eps$-good approximation of the true ground energy when the number of terms is sufficiently large, specifically $\Omega(n) \cdot \left(\frac{n}{\ell}\right)^{k/2-1} \log (n) /\varepsilon^4$. Second, we show by construction that this technique, while successful on average over random Hamiltonian systems, can fail to produce good certificates on worst-case instances.

    Our spectral technique for producing these certificates comes from extending classical results on $k$-$\mathsf{XOR}$ refutations to $k$-local Pauli Hamiltonians by crafting a quantum variant of the Kikuchi matrix for CSP refutations. To show the limitations of this technique, we prove non-commutative Sum-of-Squares lower bounds for worst-case signed $k$-local Pauli operators. More generally, we explore how the non-commutative Sum-of-Squares relaxation can be understood as augmenting the standard Sum-of-Squares relaxation with the commutation relations between the Pauli operators. We instantiate the resulting framework with a modification to prior quantum code-based NLTS Hamiltonians that yields stronger complexity guarantees for the low-energy space; our Hamiltonian family satisfies simultaneously (1) $\Omega(\log n)$-circuit depth lower bounds for all low-energy states, (2) constant-gap $\mathsf{NP}$-hardness to approximate the ground energy, and (3) a constant-gap non-commutative Sum-of-Squares integrality gap up to $\Omega(n)$-levels.
\end{abstract}

\thispagestyle{empty}

\newpage

{
\hypersetup{hidelinks}
\tableofcontents
}
\thispagestyle{empty}
\newpage

\setcounter{page}{1}

\section{Introduction}
\label{sec:introduction}


The local Hamiltonian is a fundamental object of study in quantum many-body physics. From a theoretical computer science perspective, its significance is reflected in the fact that local Hamiltonian ground-energy problems can capture a broad class of quantum computations \cite{KitaevSV02, KempeKR04}. In fact, this correspondence is universal: ground-energy estimation is complete for $\QMA$, the class of problems admitting efficient quantum verification. A central goal is therefore to understand the boundary between local Hamiltonians that remain computationally intractable and those that can be efficiently understood, whether through classical algorithms or realizable quantum computation. Much of the literature focuses on a few canonical average-case models, which we hope either capture salient features of random quantum systems described by nature or, at the very least, provide simplified settings in which to understand the behavior of local Hamiltonians. One significant example is the Sachdev–Ye–Kitaev (SYK) model, an average-case $k$-local fermionic Hamiltonian model in which independent Gaussian weights are assigned to all possible $k$-local Majorana fermion operators.
From a computational standpoint, the natural question is to approximate the ground energy, $\lambdamin(\bsH)$, given $\bsH$. Its motivation is perhaps best understood by observing that the SYK model is the natural fermionic generalization of the classical Sherrington-Kirkpatrick (SK) model when $k = 2$ and, in general, the classical $p$-spin spin glass model when $p = k$. In this work, we consider a simpler and related model called the ``quantum $k$-spin Hamiltonian'',  which is the natural qubit Hamiltonian analogue of the SK and $p$-spin models. 

\begin{definition}[Quantum $k$-spin Hamiltonians \cite{ErdosS14, BaldwinS20, AnschuetzGK25}]
    \label{def:kspin}
    A quantum $k$-spin Hamiltonian is a semirandom model which can be described by $\mathbfcal{I} = (\calH, \{(P_C, \bb_C)\}_{C \in \calH})$ where $\calH$ is a $k$-uniform hypergraph on $[n]$, each $P_C$ is a succinctly described $n$-qubit Pauli operator acting non-trivially exactly on $C$, and each $\bb_C \in \R$ is drawn as an independent standard Gaussian random variable. The instance $\mathbfcal{I}$ defines a Hamiltonian on $n$ qubits given by
    \begin{equation*}
        \bsH_{\mathbfcal{I}} \coloneqq \frac{1}{\abs{\calH}} \sum_{C \in \calH} \frac{\Id + \bb_CP_C}{2}\mper
    \end{equation*}
\end{definition}

Despite the identified importance of these models, there is still much to be understood about their computational complexity, which mirrors the broader state of Hamiltonian complexity, for which many integral phenomena lack a comprehensive theory even when a theory exists describing the classical analogue. Perhaps the most famous example is the quantum PCP (probabilistically checkable proofs) conjecture~\cite{QuantumNPsurvey,qpcp-survey}, which remains conjectural despite the great body of techniques from the study of classical PCPs and hardness-of-approximation. Our primary goal in this work is to extend tools from the study of average-case constraint satisfaction problems into the Hamiltonian setting to better understand these average-case models and more broadly the complexity of local Hamiltonians.

Towards this goal, we introduce a direct quantum analogue of the classical $k$-$\XOR$ problem. 

\begin{definition}[Hamiltonian $k$-$\XOR$]
    \label{def:quantumxor}
    We define an instance $\calI = (\calH, \{(P_C, b_C)\}_{C \in \calH})$ of Hamiltonian $k$-$\XOR$ where $\calH$ is a $k$-uniform (multi)hypergraph on $[n]$, each $P_C$ is a succinctly described $n$-qubit Pauli operator acting non-trivially exactly on $C$, and each $b_C$ is a $\pm1$-interaction coefficient. The instance $\calI$ defines a Hamiltonian on $n$ qubits given by
    \begin{equation*}
        \bH_{\calI} \coloneqq \frac{1}{\abs{\calH}} \sum_{C \in \calH} \frac{\Id + b_CP_C}{2}\mper
    \end{equation*}
\end{definition}

The $k$-$\XOR$ problem, maximizing the fraction of satisfied $k$-sparse equations over $\F_2$, is integral to the study of Boolean constraint satisfaction problems, largely due to its relation to the parity functions which form the basis for the canonical Fourier decomposition. Our primary motivation for this model, and the origin of our naming for it, comes from the simple observation that this Hamiltonian is a sum of $\pm 1$-signed $k$-local Pauli projectors. Since the Pauli operators form an orthonormal basis for Hermitian operators, this is an exact quantum analogue of classical $k$-$\XOR$, which, when written as a polynomial optimization problem, consists of a sum of $\pm1$-signed degree-$k$ parity functions. 
In this sense, the $k$-$\XOR$ problem marries our interest in random models found in nature with the study of classical constraint satisfaction problems. In particular, quantum $k$-spin Hamiltonians are equivalent to this model after replacing the $\pm1$-interactions with independent Gaussians.

Since our primary goal in this paper is to study the average-case complexity of Hamiltonians, we introduce our average-case model for the Hamiltonian $k$-$\XOR$ model here. We say a Hamiltonian $k$-$\XOR$ instance $\mathbfcal{I}$ is semirandom if $\calH$ and $\{P_C\}_{C \in \calH}$ are chosen worst-case but each $\bb_C$ for $C \in \calH$ is a uniformly random $\pm 1$-sign. If additionally $\mathbfcal{H}$ is a uniformly random hypergraph of a fixed number of terms (i.e. each hyperedge is drawn independently and uniformly from $\binom{n}{k}$ with replacement) and each Pauli term is uniformly random among degree-$k$ Paulis with fixed support, we say that $\mathbfcal{I}$ is random. With this point of view, the Hamiltonian $k$-$\XOR$ model is closely related to the quantum $k$-spin Hamiltonian of \Cref{def:kspin}.

\subsection{Main results}

In this work, we study, for any fixed constant $k$, the random and semirandom complexity of Hamiltonian $k$-$\XOR$ and quantum $k$-spin Hamiltonians, primarily through the refutation problem. Given a Hamiltonian $\bH$, our quantity of interest is the maximum energy $\lambdamax(\bH)$ (whose computational complexity is equivalent to that of the ground energy $\lambdamin(-\bH)$).

\subsubsection*{Spectral refutations of semirandom Hamiltonians}

 One can show that sufficiently dense instances of semirandom Hamiltonians have maximum energy $\in [\frac{1}{2},\frac{1}{2}+o(1)]$ with high probability. Therefore, the problem of approximating the energy is trivial; outputting $\frac{1}{2}$ deterministically is a sufficient approximation with high probability over the instance. This task is then more aptly replaced with that of providing an energy refutation: given a particular instance $\bsH_{\mathbfcal{I}}$ drawn from said distribution, provide a value $\algval(\bsH_{\mathbfcal{I}})$ satisfying $\lambdamax(\bsH_{\mathbfcal{I}}) \leq \algval(\bsH_{\mathbfcal{I}})$ for all instances in the distribution, while giving as tight an upper bound as possible. Our first result gives a spectral algorithm for the refutation task in the semirandom setting whose runtime depends on the density (in number of terms) of the Hamiltonian. 

\begin{theorem}[Spectral refutations of semirandom Hamiltonians]
    \label{thm:refutation}
    Fix integers $n/2 \geq \ell \geq k$ with $k \geq 2$. There is a classical spectral algorithm taking as input a Hamiltonian $k$-$\XOR$ instance $\mathbfcal{I} = (\calH, \{(P_C, \bb_C)\}_{C \in \calH})$ describing an $n$-qubit Hamiltonian $\bsH_{\mathbfcal{I}}$ and outputs a number $\algval(\bsH_{\mathbfcal{I}}) \in \R_{\geq 0}$ in time $n^{O(\ell)}$ with the following two guarantees:
    \begin{enumerate}
        \item \label{item:refutation1} $\algval(\bsH_{\mathbfcal{I}}) \geq \lambda_{\mathrm{max}}(\bsH_{\mathbfcal{I}})$ for every instance;
        \item \label{item:refutation2} If $\abs{\calH} = \Omega(n) \cdot \left(\frac{n}{\ell}\right)^{k/2-1} \cdot \log(n) \cdot \varepsilon^{-4}$ for $\varepsilon \in (0, \frac{1}{2})$ and $\mathbfcal{I}$ is drawn from the semirandom Hamiltonian $k$-$\XOR$ or quantum $k$-spin model, then with high probability over the draw of the semirandom instance, i.e., the randomness of the interactions $\{\bb_C\}_{C \in \calH}$, it holds that $\algval(\bsH_{\mathbfcal{I}}) \leq \frac{1}{2} + \varepsilon$.
    \end{enumerate}
\end{theorem}
The distinction between upper and lower bounds for the maximum eigenvalue is important for understanding what our algorithm guarantees. For a maximum energy problem, lower bounds are simple: a witness for the maximum energy is \emph{any} high energy state or its description, as verifying the energy can be checked efficiently on a quantum computer. In general, one might expect high energy states to be highly entangled and thus not have short classical witnesses, but in the average-case setting studied in this work, there is almost always an admissible product state approximation and classical algorithm. To see this, we can apply a quick matrix concentration bound to show the maximum energy of any semirandom $k$-$\XOR$ Hamiltonians lies in $[\frac{1}{2}, \frac{1}{2} + \varepsilon]$ (see \Cref{lem:matrixchernoff}) with high probability. Then, since a random product state achieves value $\frac{1}{2}$, a simple product state witness serves as an $\varepsilon$-approximation to the maximum energy state. Together, this yields a succinct classical approximation to the maximum energy. In contrast, providing a certified upper bound for the maximum energy is a priori much more difficult. 
Our particular algorithm takes the form of a parameterized spectral refutation:
\begin{enumerate}
    \item Choose a desired size $\ell \geq k/2$.
    \item Given a Hamiltonian $k$-$\XOR$ instance $\calI$, construct a matrix $K^{(\ell)}_\calI \in \R^{n^{O(\ell)} \times n^{O(\ell)}}$.
    \item Compute and return $\norm{K^{(\ell)}_\calI}{2}$ in time $n^{O(\ell)}$.
\end{enumerate}
Constructing matrices $K_\calI^{(\ell)}$ that satisfy $\lambdamax(\bH_\calI) \leq \norm{K^{(\ell)}_\calI}{2}$ is the central challenge. Most of the technical work is in choosing the right family to make the proof of this inequality simple, but armed with this inequality, the classical algorithm (or more technically, its transcript) serves as a certificate of the given refutation. In our particular case, the family of matrices is an adaptation of the Kikuchi matrices of \cite{WeinAM19, GuruswamiKM22}, which are the state-of-the-art refutations for classical $k$-$\XOR$.

\subsubsection*{\texorpdfstring{$\NP$-hard NLTS Hamiltonians with non-commutative Sum-of-Squares lower bounds}{NP-hard NLTS Hamiltonians with non-commutative Sum-of-Squares lower bounds}}

A natural question raised by \Cref{thm:refutation} is the origin of the tradeoff between the number of terms and the running time or accuracy, and whether this tradeoff is tight. For an informal intuition on tradeoff optimality, recall that any classical $k$-$\XOR$ instance can be embedded directly into a local Hamiltonian diagonal in the $Z$-basis. Consequently, an improvement to \Cref{thm:refutation} would imply a corresponding improvement to classical $k$-$\XOR$ refutation algorithms. However, the Kikuchi matrix method is known to be essentially tight in the classical setting \cite{KothariMOW17, GuruswamiKM22}, reflecting the fact that the level-$\ell$ Kikuchi matrix is captured by the degree-$2\ell$ Sum-of-Squares relaxation. We establish an analogous relationship between our Kikuchi matrix algorithm and the \emph{non-commutative} Sum-of-Squares hierarchy, thereby showing that our analysis is tight for this construction.

To understand this relationship thoroughly, we formally introduce the non-commutative Sum-of-Squares method from first principles, starting with the setup of the standard Sum-of-Squares algorithm over the Boolean hypercube. Suppose, for input, we have a real polynomial $p$ for which we are interested in computing the maximum value. An a priori perhaps strange way to rewrite this optimization is as a maximization over distributions:
\begin{align*}
    \max_{\mu}\;\;\;\;&\E_{\mu}\sbra{p} = \sum_{x \in \{\pm1\}^n} \mu(x) \cdot p(x)\\
    \text{subject to}\;\;\;\;& \mu \text{ a distribution over $\{\pm1\}^n\mper$}
\end{align*}
The (degree-$d$) Sum-of-Squares algorithm takes this optimization over distributions and relaxes to optimizing over (degree-$d$) \emph{pseudo-distributions}, which are not true distributions but roughly have the property that they agree on degree-$d$ moments. See \Cref{sec:sos} for a complete definition.

Analogously, the computation of the maximum energy of a Hamiltonian $\bH$ can be naturally expressed as an optimization over pure states or as an optimization over mixed states:
\begin{equation*}
    \max_{\text{unit~}\ket{\psi}} \ev{\bH}{\psi} = \max_{\text{density matrix~}\rho} \tr(\bH \rho).
\end{equation*}
In the same vein as typical Sum-of-Squares, the non-commutative Sum-of-Squares algorithm relaxes the notion of density operator to a pseudo-density operator.
\begin{align*}
    \max_\rho\;\;\;&\tr(\bH\rho) = \sum_{i =1}^{2^n} \lambda_i(\rho) \cdot \tr(\bH\ket{\psi_i} \bra{\psi_i}) = \pE_{i \sim \lambda(\rho)} \sbra{\bra{\psi_i}\bH \ket{\psi_i}}\\
    \text{subject to}\;\;\;\;& \rho \text{ a degree-$d$ pseudo-density operator on $(\C^2)^{\otimes n}.$}
\end{align*}
We define the notion of a \emph{pseudo-density operator} in \Cref{sec:sos}. We call the maximum of the above program the (degree-$d$) non-commutative Sum-of-Squares relaxation value, $\ncSoS_d(\bH)$. Since every density matrix is a valid pseudo-density operator, we always have $\lambdamax(\bH) \leq \ncSoS_d(\bH)$. Moreover, as $d \to 2n$ from below, the non-commutative Sum-of-Squares relaxation value approaches $\ncSoS_d(\bH) \to \lambdamax(\bH)$ from above. One concrete goal is then to bound the integrality gap, $\ncSoS_d(\bH) - \lambdamax(\bH)$, as a function of $d$. Accordingly, the integrality gap is the measure of the approximation quality. Because $\ncSoS_d(\bH)$ can typically be calculated in time $n^{O(d)}$ using semidefinite programming, this generates a natural hierarchy of relaxations that trade off accuracy and efficiency.

Our key insight here---which has a remarkably short proof (see \Cref{sec:kikuchi})---is that our spectral algorithm is captured by a degree-$2\ell$ non-commutative Sum-of-Squares proof, meaning that for any Hamiltonian $k$-$\XOR$ instance $\calI$ we have 
\begin{align*}
    \ncSoS_{2\ell}(\bH) \leq \algval(\bH_\calI)\mper
\end{align*}
The utility of this inequality is two-fold: first, as a direct corollary of our Kikuchi matrix analysis in \Cref{thm:refutation}, we establish equivalent tightness of non-commutative Sum-of-Squares throughout our parameter regime. This serves as a refutation witness \emph{distinct} from the Kikuchi matrix refutation. Namely, the new refutation is a degree-$2\ell$ non-commutative Sum-of-Squares proof, which is again of size $n^{O(\ell)}$. Second, this allows us to prove lower bounds against our spectral algorithm by proving them within the non-commutative Sum-of-Squares framework.

Using this formal connection between Kikuchi matrices and non-commutative Sum-of-Squares, we formalize in \Cref{sec:ncsos} our earlier argument reducing classical $k$-$\XOR$ to Hamiltonian $k$-$\XOR$ diagonal in the $Z$-basis, in greater generality, by rigorously showing how one can lift Sum-of-Squares lower bounds to non-commutative Sum-of-Squares. While this answers the question of tightness of \Cref{thm:refutation}, it leaves much to be desired for the quantum complexity theorist. First, an unsatisfying feature of this lifting theorem is that the resulting object is a Hamiltonian in name only; since all terms end up diagonal in the standard basis, the entire problem reduces back to the classical setting. Counterintuitively, we use this fact to give ``hardness'' for the quantum model via non-commutative Sum-of-Squares lower bounds. One motivation for studying non-commutative Sum-of-Squares is to understand what manifestly quantum and non-trivial Hamiltonians look like, to both better inform us as to their structure towards algorithm design, but also to provide increased insight into the types of entanglement structure necessary for quantum PCP local Hamiltonians, similar to the initial insights provided by the NLTS theorem~\cite{AnshuBN23}.

The natural question is, therefore, to identify Hamiltonian $k$-$\XOR$ instances that are not trivially aligned to classical CSPs and for which only high-degree non-commutative Sum-of-Squares provide good refutations. To make progress on this question, we extend from one-basis to two-basis Hamiltonian $k$-$\XOR$ via the following model.

\begin{definition}[$XZ$-$k$-$\XOR$]
    \label{def:xzxor}
    Let $(\calI_X, \calI_Z)$ be a pair of classical $k$-$\XOR$ instances, $\calI_X = (\calH_X, \{b_C^{(X)}\}_{C \in \calH_X})$ and $\calI_Z = (\calH_Z, \{b_C^{(Z)}\}_{C \in \calH_Z})$. We define the Hamiltonian $\bH_{\calI_X}$ as the Hamiltonian $k$-$\XOR$ instance given by $(\calH_X, \{(X_C, b_C^{(X)})\}_{C \in \calH_X})$ and define $\bH_{\calI_Z}$ analogously. The $XZ$-$k$-$\XOR$ instance is the pair $\calI = (\calI_X, \calI_Z)$ and the $XZ$-$k$-$\XOR$ Hamiltonian is given by
    \begin{align*}
        \bH_\calI = \frac{\abs{\calH_X}}{\abs{\calH_X} + \abs{\calH_Z}}\bH_{\calI_X} + \frac{\abs{\calH_Z}}{\abs{\calH_X} + \abs{\calH_Z}}\bH_{\calI_Z}\mper
    \end{align*}
\end{definition}

Restricting from general Hamiltonian $k$-$\XOR$ to $XZ$-$k$-$\XOR$ is more of a conceptual switch than a computational one, but it allows us to understand how non-commutative Sum-of-Squares optimizes $k$-$\XOR$ instances by relating more closely to the classical setting, since this Hamiltonian is expressed as a linear combination of two classical CSPs, crucially in \emph{non-commuting} bases. With this shift, we can see clearly a very satisfying paradigm: non-commutative Sum-of-Squares optimization of Hamiltonian $k$-$\XOR$ proceeds by first searching for local refutations of satisfiability in an analogous way to classical $k$-$\XOR$, followed up by checking commutation relations. The primary difference between the standard and non-commutative versions is isolated completely to this second step. In the case of $XZ$-$k$-$\XOR$ (or even Hamiltonian $k$-$\XOR$), this turns out to be simple to characterize: the only non-commutative structure is found in the anticommutation relations between the Paulis, and moreover, the presence of any anticommuting pairs implies a frustrated Hamiltonian instance.

This observation allows us to build $XZ$-$k$-$\XOR$ Hamiltonians with non-commutative Sum-of-Squares lower bounds that are further removed from classical CSPs.

\begin{theorem}[$\NP$-hard NLTS Hamiltonians with non-commutative Sum-of-Squares lower bounds]
    \label{thm:ncsosforr}
    For sufficiently large fixed integer $k$, there exists an infinite promise family $\mathfrak{I} = \bigcup_n \mathfrak{I}_n$ of (inhomogeneous) $XZ$-$k$-$\XOR$ instances, where each $\calI \in \mathfrak{I}_n$ is on $n$-qubits and satisfies for some constant $\varepsilon > 0$:
    \begin{enumerate}
        \item \label{item:nlts} \textnormal{(NLTS)} For all states $\ket{\psi} \in (\C^2)^{\otimes n}$ with $\bra{\psi} \bH_{\calI} \ket{\psi} \geq \lambdamax(\bH_{\calI}) - \varepsilon$, the circuit depth complexity of $\ket{\psi}$ is $\Omega(\log n)$.
        \item \label{item:nphard} \textnormal{($\NP$-hardness-of-approximation)} There are constants $\frac{1}{2} \leq \alpha < \beta \leq 1$ with $\beta - \alpha \geq \varepsilon$ such that it is $\NP$-hard under randomized polynomial-time reductions to distinguish $\lambdamax(\bH_{\calI}) \leq \alpha$ from $\lambdamax(\bH_{\calI}) \geq \beta$.
        \item \label{item:ncsos} \textnormal{(Non-commutative Sum-of-Squares lower bounds)} The degree-$\Omega(n)$ non-commutative Sum-of-Squares relaxation value is $1$, i.e., for some degree function $d(n) = \Omega(n)$, $\ncSoS_{d(n)}(\bH_{\calI}) = 1$ for all $n$.
    \end{enumerate}
\end{theorem}

The reason for studying these models of Hamiltonians, and the lower-bound techniques that accompany them, is ultimately to understand how robustly their low-energy behavior evades classical description. The NLTS property addresses this question through state preparation: it rules out low-depth quantum circuits as succinct descriptions of low-energy states. In doing so, NLTS Hamiltonians satisfy a necessary condition for quantum PCPs---namely, an inability to provide low-depth circuit descriptions as $\NP$ witnesses for local Hamiltonian energy estimation problems. A natural next step is therefore to ask whether the same Hamiltonians remain hard under other, more permissive notions of classical description. This viewpoint has motivated several recent works\footnote{These works rule out stabilizer states in the low-energy space~\cite{CobleCoudronNelsonNezhadi23Stabilizer}, prove linear lower bounds on the number of $T$ gates required to prepare low-energy states~\cite{CobleCoudronNelsonNezhadi23TCount}, and study lower bounds against still richer circuit models through the magic hierarchy~\cite{Parham26Magic}.}.

\Cref{thm:ncsosforr} pursues this program from the non-commutative Sum-of-Squares lower bounds notion of classical description, an angle that (to the best of our knowledge) is unstudied. Rather than restricting how a state may be prepared, we ask whether its relevant local statistics can be captured by a classical optimization hierarchy. The non-commutative Sum-of-Squares hierarchy provides a particularly natural test: a degree-$d$ pseudo-expectation records a collection of low-degree moments that are locally consistent with a quantum state, while being permitted globally to describe no physical state at all. In this sense, non-commutative Sum-of-Squares gives a substantially more permissive surrogate for a low-energy witness than an explicit low-depth preparation. Nevertheless, \Cref{thm:ncsosforr} shows that degree-$\Omega(n)$ non-commutative Sum-of-Squares fails to approximate the optimum of these Hamiltonians within a constant.

We emphasize that the three conclusions of \Cref{thm:ncsosforr} do not arise by gluing together separate constructions for NLTS, $\NP$-hardness, and Sum-of-Squares hardness. Instead, our Hamiltonians are only a mild perturbation of the existing code-based NLTS constructions, and the main new ingredient is a stronger analysis showing that these same Hamiltonians also exhibit linear-degree non-commutative Sum-of-Squares integrality gaps. Thus, we do not view the construction itself as proposing a particularly distinguished new physical model. Rather, the theorem provides evidence that the complexity already present in the NLTS Hamiltonians is robust to another powerful form of classical description. We see this as complementary to the growing body of work seeking increasingly strong manifestations of complexity in the low-energy spaces of local Hamiltonians~\cite{CobleCoudronNelsonNezhadi23Stabilizer,CobleCoudronNelsonNezhadi23TCount,Parham26Magic}.


\begin{figure}[t]
    \centering
    \begin{tikzpicture}[
        x=1cm,
        y=1cm,
        every node/.style={font=\small},
        >={Latex[length=2.2mm]},
        case/.style={
            align=right,
            text width=3.0cm,
            font=\small
        },
        desc/.style={
            align=right,
            text width=3.0cm,
            font=\small
        },
        landscape/.style={
            draw=black!50,
            rounded corners=2pt,
            line width=0.6pt
        },
        threshold/.style={
            dotted,
            line width=1.1pt,
            black!65
        },
        spectrum/.style={
            line width=1.8pt,
            black
        },
        sos/.style={
            dashed,
            line width=1.8pt,
            RoyalPurple
        }
    ]

    \def\xCase{1.6}
    \def\xL{7.2}
    \def\xR{13.2}

    \def\ya{2.08}
    \def\yb{2.96}

    \def\yTopBase{5.6}
    \def\yTopMax{9.00}

    \node[case] (topcase) at (\xCase,\yTopBase+2.15)
        {\textbf{Case $\lambda_{\max}(\bH)\ge \beta$}};

    \node[
        desc,
        text=Cerulean!70!black,
        font=\small\bfseries
    ] (hardstates) at (\xCase,\yTopBase+0.95)
        {high-energy space};

    \draw[landscape] (\xL,\yTopBase) rectangle (\xR,\yTopBase+4.0);

    \fill[Cerulean!40,opacity=.55]
        (\xL,\yTopBase+\ya) rectangle (\xR,\yTopMax);

    \node[anchor=east] at (\xL-0.18,\yTopBase) {$0$};
    \node[anchor=east] at (\xL-0.18,\yTopBase+4.0) {$1$};

    \draw[threshold] (\xL,\yTopBase+\ya) -- (\xR,\yTopBase+\ya);
    \node[anchor=west] at (\xR+0.18,\yTopBase+\ya) {$\alpha$};

    \draw[threshold] (\xL,\yTopBase+\yb) -- (\xR,\yTopBase+\yb);
    \node[anchor=west] at (\xR+0.18,\yTopBase+\yb) {$\beta$};

    \draw[spectrum] (\xL,\yTopMax) -- (\xR,\yTopMax);
    \node[anchor=west] at (\xR+0.18,\yTopMax) {$\lambda_{\max}(\bH)$};

    \draw[->, line width=1pt, Cerulean!70!black]
        (hardstates.east)
        to[out=10,in=235,looseness=1.15]
        ($( \xL,\yTopBase+\ya)!0.38!(\xR,\yTopMax)$);

    \def\yBotBase{0.0}
    \def\yBotLam{1.55}
    \def\yBotSos{3.58}

    \node[case] (botcase) at (\xCase,\yBotBase+2.15)
        {\textbf{Case $\lambda_{\max}(\bH)\le \alpha$}};

    \node[
        desc,
        text=RoyalPurple!80!black,
        font=\small\bfseries
    ] (gaplabel) at (\xCase,\yBotBase+0.95)
        {$\ncSoS_{\Omega(n)}(\bH) \geq \beta$};

    \draw[landscape] (\xL,\yBotBase) rectangle (\xR,\yBotBase+4.0);

    \fill[Lavender!40,opacity=.7]
        (\xL,\yBotBase+\yb) rectangle (\xR,\yBotBase+\yBotSos);

    \node[anchor=east] at (\xL-0.18,\yBotBase) {$0$};
    \node[anchor=east] at (\xL-0.18,\yBotBase+4.0) {$1$};

    \draw[spectrum] (\xL,\yBotBase+\yBotLam) -- (\xR,\yBotBase+\yBotLam);
    \node[anchor=west] at (\xR+0.18,\yBotBase+\yBotLam) {$\lambda_{\max}(\bH)$};

    \draw[threshold] (\xL,\yBotBase+\ya) -- (\xR,\yBotBase+\ya);
    \node[anchor=west] at (\xR+0.18,\yBotBase+\ya) {$\alpha$};

    \draw[threshold] (\xL,\yBotBase+\yb) -- (\xR,\yBotBase+\yb);
    \node[anchor=west] at (\xR+0.18,\yBotBase+\yb) {$\beta$};

    \draw[sos] (\xL,\yBotBase+\yBotSos) -- (\xR,\yBotBase+\yBotSos);
    \node[anchor=west, text=RoyalPurple] at (\xR+0.18,\yBotBase+\yBotSos)
        {$\ncSoS_{\Omega(n)}(\bH)$};

    \draw[->, line width=1pt, RoyalPurple!80!black]
        (gaplabel.east)
        to[out=10,in=235,looseness=1.15]
        ($( \xL,\yBotBase+\yb)!0.38!(\xR,\yBotBase+\yBotSos)$);

    \end{tikzpicture}

    \captionsetup{
    font=small,
    width=0.85\linewidth
}

    \caption{
        Two cases of the $(\beta,\alpha)$-$\mathsf{Local}$-$\mathsf{Hamiltonian}$ maximum-energy
        problem. In the YES case (top panel), there exist high-energy states above
        $\alpha$ up to $\lambda_{\max}(\bH)$. If the $(\beta,\alpha)$-energy problem is $\QMA$-hard, then the states in the marked energy space do not have verifiable short classical descriptions and finding any such state is (morally) $\QMA$-hard. In the NO case (bottom panel),
        $\lambda_{\max}(\bH)\le \alpha$, $\QMA$-hardness of the $(\beta,\alpha)$-energy problem implies that the low-degree
        non-commutative Sum-of-Squares value remains above $\beta$, as the value is classically easy to compute. 
        \Cref{thm:ncsosforr} describes Hamiltonians for which both types of bounds hold.
    }
    \label{fig:lh-max-energy-gap}
\end{figure}

\subsubsection*{Refutations and lower bounds for random Hamiltonians}

Observe that the construction in~\Cref{thm:ncsosforr} is of a commuting Hamiltonian; computing the maximum energy of commuting Hamiltonians is \emph{believed} to be strictly easier than that of non-commuting Hamiltonians (and crucially not known to be $\QMA$-complete). This separation suggests that the lower bound results of~\Cref{thm:ncsosforr} are essentially optimal for constructions built from commuting code Hamiltonians (this statement is not meant to be rigorous). Morally speaking, this suggests that upgrading~\Cref{thm:ncsosforr} to $\QMA$-hardness from $\NP$-hardness in hopes of even stronger evidence towards the quantum PCP conjecture will require studying non-commuting relations between local terms. This push to studying non-commuting terms, therefore, requires understanding the success of the non-commutative Sum-of-Squares framework on such Hamiltonians.

The reason for this commutation assumption, in the context of \Cref{thm:ncsosforr}, is that deciding if a Hamiltonian $k$-$\XOR$ has maximum energy 1 has a classical polynomial-time algorithm\footnote{The $\P$-algorithm is Gaussian elimination, which is known to require high-degree Sum-of-Squares.} (i.e., the problem is in $\P$). If the Hamiltonian $k$-$\XOR$ instance has any anticommuting pair of Paulis, the maximum energy is necessarily not equal to 1, and non-commuting Sum-of-Squares refutations can detect this anticommutation in degree-$2k$. This explains why the commuting assumption is necessary in \Cref{thm:ncsosforr} for perfect completeness integrality gaps for Pauli Hamiltonians.

As a testbed for understanding non-commutation through the lens of non-commutative Sum-of-Squares, we explore \emph{fully} random Hamiltonian $k$-$\XOR$ in the refutation setting. Since fully random instances have many anticommuting pairs above a certain density threshold, one cannot hope for a perfect completeness integrality gap like in the semirandom case. In our second hardness result, \Cref{thm:randomxor}, we show that random, sparse $k$-$\XOR$ Hamiltonians also have an integrality gap, meaning that $\ncSoS_{\Theta(n)}(\bH_\calI) - \lambdamax(\bH_\calI)$ is at least a constant. This holds even though the non-commutative Sum-of-Squares value is strictly less than $1$, a deviation from the previous constructions where the maximum energy was small but the non-commutative Sum-of-Squares value was 1 for high-degree refutations.  

\begin{theorem}[Non-commutative Sum-of-Squares bounds for random $k$-$\XOR$ Hamiltonians]
    \label{thm:randomxor}
    Fix $k \geq 3$ and $n \geq \ell \geq \Omega(\log n)$. Let $\bsH_{\mathbfcal{I}}$ be an $n$-qubit $k$-$\XOR$ Hamiltonian described by the fully random instance $\mathbfcal{I} = (\mathbfcal{H}, \{(\bP_C, \bb_C)\}_{C \in \mathbfcal{H}})$ with $\abs{\mathbfcal{H}} = \Theta(n) \cdot \left(\frac{n}{\ell}\right)^{k/2-1}$. Then, with constant probability over the draw of instance $\mathbfcal{I}$, it holds that:
    \begin{enumerate}
        \item \label{item:random1} $\lambdamax(\bsH_{\mathbfcal{I}}) \leq \frac{1}{2} + O(\sqrt{\frac{n}{\abs{\calH}}})$;
        \item \label{item:random2} The degree-$2k$ non-commutative Sum-of-Squares relaxation for Hamiltonians believes $\lambdamax(\bsH_{\mathbfcal{I}})$ is at most $\frac{1}{2} + \frac{1}{2}\sqrt{\frac{1}{3} + \frac{n}{\abs{\calH}}} + o(1)$;
        \item \label{item:random3} The degree-$\wt{\Omega}(\ell)$ non-commutative Sum-of-Squares relaxation for Hamiltonians believes $\lambdamax(\bsH_{\mathbfcal{I}})$ is at least $\frac{1}{2}+ \frac{1}{2 \cdot 3^k} - o(1)$.
        \end{enumerate}
\end{theorem}

Beyond weak refutation, which requires only refuting perfect completeness, our algorithm refutes the value down towards $\frac{1}{2} + \frac{1}{2}\sqrt{\frac{1}{3}} \approx 0.79$ with constant probability and in extremely low-degree, breaking the refutation threshold behavior in the classical setting. While this alternative refutation method works for weak refutation, it is bottlenecked from achieving strong refutation: while the true maximum eigenvalue still tends to $\frac{1}{2}$ when $\abs{\calH} \gg n$, we show a constant lower bound $\frac{1}{2} + \frac{1}{2 \cdot 3^k}$ holds below the refutation threshold. As a result, even though anticommutation drives the value below $1$, there is still a constant-sized gap between the true maximum energy and the non-commutative Sum-of-Squares relaxation value.

This decomposition fits nicely into what we ultimately view as a broader paradigm for non-commutative Sum-of-Squares: its relaxation value can be understood as an interpolation of standard Sum-of-Squares values of classical CSPs via their commutation relations. In the case of $XZ$-$k$-$\XOR$, we see that assuming commutation removes the latter part, allowing us to lift Sum-of-Squares lower bounds from the underlying $k$-$\XOR$ instances. In the random setting, we observe that below the refutation threshold non-commutative Sum-of-Squares is able to do something standard Sum-of-Squares cannot, leveraging the anticommutation to derive a different type of refutation. We hope that this point of view can be utilized further in the future to understand the broader role of non-commutative Sum-of-Squares as an optimal algorithm for Hamiltonian optimization.

\subsection{Proof overview}

\subsubsection*{The Hamiltonian Kikuchi Matrix}

Our main algorithmic contribution in this work is an adaptation of the Sum-of-Squares viewpoint and state-of-the-art in classical CSP refutation to the quantum setting through a quantum adaptation of the Kikuchi matrix hierarchy of \cite{WeinAM19, AbascalGK21, GuruswamiKM22, KocurekM25}. Our variant of the Kikuchi matrix is built on the signed adjacency matrix of a graph on (degree-$\ell$) Pauli operators, which may be of independent interest. 

We build up to our Hamiltonian Kikuchi matrix construction by relating it to the typical trace moment method analysis for bounding the maximum energy of a Hamiltonian. Given an $n$-qubit Hamiltonian $\bH$, we write its Pauli decomposition as $\bH= \sum_{P \in \{\Id, X, Y, Z\}^{\otimes n}} \hat{\bH}(P) \cdot P$ and assume it has a sparse $\poly(n)$-sized Pauli support. We then apply the trace moment method, which allows us to bound $\lambdamax(\bH) \leq \tr(\bH^{2r})^{1/2r}$ for some choice $r \in \N$. Computing this trace out gives the following expression
\begin{align*}
    \tr(\bH^{2r}) &= \tr\left(\left(\sum_{P \in \{\Id, X, Y, Z\}^{\otimes n}} \hat{\bH}(P) \cdot P\right)^{2r}\right)\\
    &= \sum_{P_1, \dots, P_{2r} \in \{\Id, X, Y, Z\}^{\otimes n}} \prod_{i=1}^{2r} \hat{\bH}(P_i) \cdot \tr\left(\prod_{i=1}^{2r} P_i\right)\mper
\end{align*}
In general, we can compute $\tr\left(\prod_{i=1}^{2r} P_i\right)$ efficiently using the stabilizer formalism for Pauli operators, so there is a simple algorithm that simply enumerates all $\mathsf{poly}(n)^{2r}$ sequences and computes each term. The only reason this fails to give an efficient algorithm is that we must take $2r \gg n$ for a $2^n$-dimensional operator in order for the trace moment method to form a tight estimate, meaning an exponential number of sequences. To get around this issue, we relate this analysis to performing a random walk on the set of $n$-qubit Pauli operators, generalizing the heuristic level-$n$ Kikuchi matrix seen in \cite{WeinAM19}.

\begin{definition}[Level-$n$ Kikuchi matrix for $k$-local Hamiltonians]
    Let $P \in \{\Id, X, Y, Z\}^{\otimes n}$ be a Pauli operator. We define the following adjacency matrix on $\{\Id, X, Y, Z\}^{\otimes n}$.
    \begin{equation*}
        A_P(Q, R) = \begin{cases}
            1 & Q^\dagger R = P\\
            0 & \text{otherwise}
        \end{cases}
    \end{equation*}
    We further extend the Kikuchi matrix of any operator $\bH$ on $\C^{2^n}$ through linearity, yielding the level-$n$ Kikuchi matrix
    \begin{equation*}
        K_{\bH} = \sum_{P \in \{\Id, X, Y, Z\}^{\otimes n}} \hat{\bH}(P) \cdot A_P\mper
    \end{equation*}
\end{definition}

Our Kikuchi matrix $K_\bH$ is now $4^n$-dimensional, and we attempt to fix this problem by counterintuitively defining an \textit{even bigger} operator.

\begin{proposition}
    \label{prop:spectrum}
    For a given Hamiltonian $\bH \in \End((\C^2)^{\otimes n})$, consider the matrices $K_\bH$, $K_\bH \otimes \Id_{2^n}$, and $\hat{K}_\bH = \sum_{P \in \{\Id, X, Y, Z\}^{\otimes n}} \hat{\bH}(P) \cdot A_P \otimes P$. The maximum eigenvalue of all three are the same.
\end{proposition}

\begin{proof}
    $K_\bH \otimes \Id_{2^n}$ has the same eigenvalues as $K_\bH$ up to multiplicity. $\hat{K}_\bH$ is obtained from $K_\bH \otimes \Id_{2^n}$ by conjugating by the block-diagonal unitary $U \in \C^{ 8^{n} \times 8^{n}}$ given by $U_{P, P} = P$.
\end{proof}

Applying the trace moment method to $\hat{K}_\bH$, which now has dimension $8^n$, finally yields something that looks like a random walk on $\{\Id, X, Y, Z\}^{\otimes n}$:
\begin{align*}
    \tr\left(\hat{K}_{\bH}^{2r}\right) &= \sum_{P_1, \dots, P_{2r} \in \{\Id, X, Y, Z\}^{\otimes n}} \prod_{i=1}^{2r} \hat{\bH}(P_i) \cdot \tr\left(\prod_{i=1}^{2r} P_i\right) \cdot \tr\left(\prod_{i=1}^{2r} A_{P_i}\right)\mper
\end{align*}

Observe this last term counts the number of walks on the graph $\{\Id, X, Y, Z\}^{\otimes n}$ with a given Pauli sequence. This term turns out to be quite uniform across Pauli sequences in the semirandom setting, making the trace moments of $\hat{K}_\bH$ match up nicely (up to proportionality) with those of $\bH$ itself. Unfortunately, our operator is still huge. Our fix is two-fold: (1) \Cref{prop:spectrum} allows us to drop the tensored $\Id_{2^n}$ and just look at the operator norm of the Kikuchi graph adjacency matrix $K_\bH$ and (2) as the key power of the Kikuchi matrix, we may restrict our adjacency matrix from the Paulis $\{\Id, X, Y, Z\}^{\otimes n}$ to the Pauli slice, $\calP_\ell(n)$, all weight-$\ell$ $n$-qubit operators, yielding a hierarchy of matrices trading dimensionality for approximation quality. Our level-$\ell$ Kikuchi matrix (for $k$ even) is formally defined as follows.

\begin{definition}[Level-$\ell$ Kikuchi matrix for $k$-local Hamiltonians]
    \label{def:introkikuchi}
    Fix $k \in \N$ even and $k/2 \leq \ell \leq n/2$. Let $P \in \calP_{\leq k}(n)$. We define the following adjacency matrix on $\mathcal{P}_\ell(n)$:
    \begin{equation*}
        A_P(Q, R) = \begin{cases}
            1 & Q^\dagger R = P \text{ and } \abs{\supp(Q) \cap \supp(R)} = \ell-k/2\\
            0 & \text{otherwise}
        \end{cases}
    \end{equation*}
    The level-$\ell$ Kikuchi matrix of a homogeneous $k$-local Hamiltonian $\bH$ is then defined by
    \begin{equation*}
        K_\bH = \left(\sum_{P \in \calP_k(n)} \hat{\bH}(P) \cdot A_P\right) \otimes \Id_{2^n} \coloneqq A^*_\bH \otimes \Id_{2^n}\mper
    \end{equation*}
\end{definition}

What we are left with is an energy bound for $\bH$ depending entirely on the operator norm of $A_\bH^*$, an operator of potentially subexponential dimension, which is what allows us to compute the spectral norm without incurring an exponential cost.

Our algorithm is then very simple: for choice of level-$\ell$, compute the spectral norm of the Kikuchi matrix, which is then our guaranteed maximum energy certificate. The main technical contribution is then in the analysis of this certificate and its efficacy in the semirandom Hamiltonian $k$-$\XOR$ and quantum $k$-spin models.

\subsubsection*{Non-commutative Sum-of-Squares Lower Bounds}

Dual to our Kikuchi matrix spectral algorithm, we take the non-commutative Sum-of-Squares point of view to understand the limits of this approach in optimizing Hamiltonian $k$-$\XOR$ instances. To do so, it is useful to restrict the Hamiltonian $k$-$\XOR$ to the $XZ$-$k$-$\XOR$ model introduced in \Cref{def:xzxor}.

Colloquially, an $XZ$-$k$-$\XOR$ Hamiltonian is a Hamiltonian $k$-$\XOR$ instance restricted to just two sets of Pauli operators: the $X$ and $Z$-basis operators. In other words, it is a simple pair of classical $k$-$\XOR$ instances that are embedded in the $X$ and $Z$-bases respectively. This notion of an $XZ$-Hamiltonian was originally introduced in \cite{MaN25}, which studied $XZ$-Quantum-$6$-$\SAT$ and ultimately proved this task $\QMA_1$-complete. Comparing to the general Hamiltonian $k$-$\XOR$ or $k$-local Hamiltonian definition, the simplification is that rather than having to account for operators decomposing into arbitrary Paulis aligning with different product bases, the $XZ$-model isolates everything to just two simple bases.

Our primary approach to non-commutative Sum-of-Squares takes advantage of this two-basis simplification to isolate a singular sense in which non-commutative Sum-of-Squares' behavior departs from that of standard Sum-of-Squares. In the context of $XZ$-$k$-$\XOR$, and even general Hamiltonian $k$-$\XOR$, this is best understood via the $\XOR$-resolution proof system and its extension into the stabilizer formalism. 

Consider first the case of classical $3$-$\XOR$, and suppose one has the following $\F_2$-equations within the instance:
\begin{align*}
    &x_1 \oplus x_2 \oplus x_3 = 0\\
    &x_1 \oplus x_2 \oplus x_4 = 0\\
    &x_3 \oplus x_5 \oplus x_6 = 0\\
    &x_4 \oplus x_5 \oplus x_6 = 1
\end{align*}
A very simple way to see why this particular system is unsatisfiable is to add all the equations together. Linear constraints should propagate through addition, but notice if we add all of the left-hand sides together we arrive at $0$, since each variable occurs with multiplicity $2$. On the other hand, the right sides of all equations add to $1$. Hence, we have derived $0 = 1$ a contradiction, and therefore not all equations can be simultaneously satisfied. This is exactly the kind of proof that Gaussian elimination searches for to solve satisfiability in polynomial-time. More explicitly, and using $\pm 1$-notation instead of $\F_2$, we introduce the following notion of a refutation.

\begin{definition}[$\XOR$-refutations]
    \label{def:xorref}
    Let $\calI = (\calH, \{b_C\}_{C \in \calH})$ be a classical $k$-$\XOR$ instance. A \emph{refutation} $\calR \subseteq \calH$ in $\calI$ is a set of equations satisfying the properties
    \begin{align*}
        \bigoplus_{C \in \calR} C = \varnothing\mcom \qquad \prod_{C \in \calR} b_C = -1\mper
    \end{align*}
\end{definition}

\begin{fact}
    Given a classical $k$-$\XOR$ instance $\calI$, $\calI$ is unsatisfiable if and only if there exists a refutation in $\calI$. Moreover, Gaussian elimination automates this procedure by finding a refutation (or a proof one does not exist) in polynomial-time.
\end{fact}

This framework extends to Hamiltonian $k$-$\XOR$ in a straightforward way.

\begin{definition}[Hamiltonian $\XOR$-refutations]
    \label{def:hamxorref}
    Let $\calI = (\calH, \{(P_C, b_C)\}_{C \in \calH})$ be a Hamiltonian $k$-$\XOR$ instance. A \emph{refutation} $\calR \subseteq \calH$ in $\calI$ is a set of terms satisfying the property that for
    \begin{align*}
        \calS = \{b_CP_C \mid C \in \calR\}\mcom \qquad -\Id \in \langle \calS \rangle\mper
    \end{align*}
    In English, the subgroup generated by the signed Pauli operators corresponding to $\calR$ contains negative identity.
\end{definition}

\begin{fact}
    Given a Hamiltonian $k$-$\XOR$ instance $\calI$, $\lambdamax(\bH_{\calI}) < 1$ if and only if there exists a refutation in $\calI$. Moreover, Gaussian elimination within the stabilizer formalism automates this procedure by finding a refutation (or a proof one does not exist) in polynomial-time.
\end{fact}

To best illustrate why Hamiltonian $\XOR$-refutations are the correct generalization of $\XOR$-refutations, we consider the case of an $XZ$-$k$-$\XOR$ instance $\calI = (\calI_X, \calI_Z)$.

We show there are in essence only two types of refutations. The immediate one is a classical $\XOR$-refutation embedded in just a single basis. If the refutation $\calR \subseteq \calH_{\calI_X}$, then it is simple to check the generation of $-\Id$ is exactly the existence of an $\XOR$-refutation within the classical instance $\calI_X$. The more interesting type of refutation comes from the commutation relations between the terms $\bH_{\calI_X}$ and $\bH_{\calI_Z}$. In the case of Pauli operators, this commutation is easy: local terms either commute or anticommute, and even the existence of a single anticommuting pair is enough to yield a refutation. Let $X_C$ and $Z_{C'}$ be two anticommuting operators. It is simple to observe
\begin{align*}
    X_C Z_{C'} X_C Z_{C'} = -\Id\mcom \qquad X_C Z_{C'} X_C Z_{C'} \in \langle X_C, Z_{C'} \rangle\mper
\end{align*}
This just captures the notion that anticommuting Paulis are not simultaneously stabilizable. Note, since each term appears twice in this case, the introduction of signs does not affect the validity of this refutation. It is a simple exercise to show that if no anticommuting pair exists, any refutation must come from a refutation in either the $X$ or $Z$-basis (or both).

The upshot of this is the following nice characterization: the existence of refutations in $XZ$-$k$-$\XOR$ boils down to checking both classical instances for a refutation and then checking for a ``commutation'' refutation, which in this case is simply the existence of any anticommuting pair. Our main contribution is then in showing that non-commutative Sum-of-Squares obeys this framework: there is a short non-commutative Sum-of-Squares proof that anticommuting pairs are not simultaneously stabilizable and moreover for any commuting instance satisfying that each classical instance has standard Sum-of-Squares value $1$, the non-commutative Sum-of-Squares value is also $1$. This latter observation immediately yields lifting lower bounds from Sum-of-Squares to non-commutative Sum-of-Squares by just embedding a classical $k$-$\XOR$ instance into a standard-basis Hamiltonian $k$-$\XOR$, but we also use it to show the two-basis result in \Cref{thm:ncsosforr} by instantiating it with a modification to the NLTS Hamiltonian. Furthermore, we show in \Cref{thm:randomxor} a simple way to take advantage of many anticommuting pairs appearing in random Hamiltonian $k$-$\XOR$ to prove a weak refutation result, and also show the limitations of such a method for proving strong refutations in the form of an imperfect completeness integrality gap.

\subsection{Related work}

 \parhead{Classical CSP refutation.} Classical $k$-$\XOR$ refutation has a storied history with deep ties to many fascinating complexity results. In the early 2000s, Håstad's celebrated result \cite{Hastad01} showed the problem of distinguishing a $3$-$\XOR$ instance with value $1-\varepsilon$ from one with value $\frac{1}{2} + \varepsilon$ is $\NP$-hard when the instance has $O(n)$ equations. Worse yet, assuming the Exponential Time Hypothesis \cite{ImpagliazzoP01}, later stronger PCP constructions \cite{MoshkovitzR08, BafnaM0Y25} rule out even $2^{n^{1-\delta}}$-time approximation algorithms for any fixed $\delta > 0$ on sparse instances. Conversely, allowing instances to be dense, i.e., $\Omega(n^k)$ equations, there is a PTAS for this problem \cite{AroraKK95}. This result can be further extended to get a smooth time vs. number of equations tradeoff that yields subexponential approximation schemes for instances with $\omega(n^{k-1})$ constraints \cite{FotakisLP16}, giving a near-complete picture for the worst-case setting. In light of this grim outlook of worst-case complexity, the study of random CSPs was initiated to potentially subvert these hardness results. Analogous to \cite{AroraKK95}, \cite{AllenOW15} showed that \textit{random} $k$-$\XOR$ instances admit a PTAS when the number of equations exceeds $\wt{O}(n^{k/2})$. This result was extended to a smooth tradeoff algorithm taking $n^{O(\ell)}$-time for an instance with $O(n) \cdot \left(\frac{n}{\ell}\right)^{k/2-1} \cdot \textsf{polylog}(n, 1/\varepsilon)$ equations \cite{RaghavendraRS17}, a tradeoff conjectured to be optimal via a line of work establishing near-matching Sum-of-Squares lower bounds \cite{Grigoriev01, Schoenebeck08, BenabbasGMT12, BarakCK15, MoriW16, KothariMOW17}. A newer line of work \cite{AbascalGK21, GuruswamiKM22, HsiehKM23} established algorithms exhibiting the same tradeoff in the stricter semirandom model of $k$-$\XOR$, where only the right-hand sides of each equation are random, that is, they show that the semirandom model is no harder than the random one. The state-of-the-art algorithms provide spectral refutations, certificates in the form of a spectral norm of a matrix, particularly through the novel Kikuchi matrix method which has recently seen a flurry of applications \cite{WeinAM19, GuruswamiKM22, HsiehKM23, AlrabiahGKM23, KothariM23, KocurekM25}. While the algorithm is different, it is known that the Kikuchi matrix method is captured by Sum-of-Squares, and thus serves only as a procedural and analytical simplification rather than a significant theoretical speedup.

The thesis of this storied line of work is perhaps the power of the Sum-of-Squares algorithm for setting ``computational speed limits'', and moreover how studying these thresholds reveals a great deal about the complexity of the underlying problem. This observation is not unique to the $k$-$\XOR$ problem or the random setting either, with the cornerstone result \cite{Raghavendra08} showing that the basic SDP, a relaxation captured by constant rounds of Sum-of-Squares, is universally optimal for CSPs under the Unique Games conjecture \cite{Khot02}. As such, a primary goal in this work is to extend the utility of Sum-of-Squares based optimization and its known computational thresholds for average-case $k$-$\XOR$ to the quantum setting, in hopes of understanding what these thresholds tell us about quantum complexity.

\parhead{Similar models to Hamiltonian $k$-$\XOR$.} To our knowledge the Hamiltonians of \Cref{def:quantumxor} are not widely labeled as an analogue of $\XOR$ and have never been studied in the semirandom setting, but similar or captured models have appeared in various contexts before. The most prominent local Hamiltonian model in computer science is perhaps the quantum $\MAXCUT$ or EPR Hamiltonian problem studied in \cite{HallgrenLP20, Lee22, HothemP023, HwangNP0W23, LeeP24} which is closely related to the $k = 2$ restriction of \Cref{def:quantumxor}. The primary algorithm used for this and related models is the non-commutative Sum-of-Squares algorithm (often under the name Quantum Lasserre Hierarchy) which we consider here. Another important class of Hamiltonians that fits directly into the Hamiltonian $k$-$\XOR$ model is the celebrated NLTS Hamiltonian construction of \cite{AnshuBN23}, or more generally the natural frustration-free Hamiltonian associated with a stabilizer code. The NLTS Theorem tells us that there is an instantiation of Hamiltonian $k$-$\XOR$ whose maximum energy states (equiv.~ground states) have non-trivial quantum circuit complexity, establishing evidence of quantum non-triviality of the worst-case version of this model. This particular example also elucidates further our identification with classical $k$-$\XOR$: the code Hamiltonian is simply the set of $X$-checks and $Z$-checks from the corresponding quantum code, which are ultimately just two classical $k$-$\XOR$ instances encoded in the $X$ and $Z$-bases respectively.

\parhead{Average-case Hamiltonians.} 
In the realm of average-case Hamiltonian complexity, we highlight in particular \cite{AnschuetzGK25, Al-GattasG26}, which establish evidence of quantum non-triviality in preparing the ground states of the mean-field quantum $k$-spin Hamiltonian, where the interaction hypergraph is complete. More explicitly, \cite{AnschuetzGK25} show a gap between the product state and general optimizer for this model, while \cite{Al-GattasG26} shows $\Omega(\log n)$-depth quantum circuit lower bounds for ground state preparation. In particular, the latter works in a thin energy window $[\alpha, \alpha + \Theta(n^{-(k-1)/2})]$ for some maximum energy threshold $\alpha \approx \frac{1}{2} + o(1)$. On the other hand, we give a refutation for this model certifying maximum energy $\frac{1}{2} + \varepsilon$ for any constant $\varepsilon > 0$ along the refutation threshold. Plugging the level of granularity required in \cite{Al-GattasG26} to our \Cref{thm:refutation}, we find matching their proposed approximation via our guarantee would require the density $\abs{\calH} \gtrsim n^{2k-1}(n/\ell)^{k/2-1}$, well above the refutation threshold. In other words, our algorithm works on a window of any constant-size, whereas the proposed quantum non-triviality of \cite{Al-GattasG26} is found in a window of inverse-polynomial size, making the results complementary. In a related vein, \cite{ChenDBBT23} studies randomly signed sparse Hamiltonians, with their main model essentially corresponding to \Cref{def:quantumxor} with $k = n$ for both Rademacher and Gaussian series, giving evidence that this regime is classically hard but quantumly easy. 

More broadly, in statistical physics, typically concrete dense or geometric models resembling Hamiltonian $k$-$\XOR$ are often studied in the spin glass literature to understand energy configurations of (potentially quantum) interacting particles. In the standard $p$-spin Ising model, where one has $p = k$, the interaction coefficients $\bb_C$ typically represent the ferromagnetism or antiferromagnetism of the interaction in the Rademacher case, and it is not uncommon to study semirandom models where these interactions are random Rademacher or Gaussian and the interacting particles are often chosen particularly on a lattice to model physically relevant systems.  The goal in these settings is usually to nail down exact behavior (down to constants) of these models, which Sum-of-Squares is usually not sufficient for. For example, for the SK model, $n^{\delta}$ levels of Sum-of-Squares are known to fail at certifying the energy of this model \cite{GhoshJJPR20} precisely. The SYK model has received similar attention; \cite{HastingsO22, AnschuetzCKK24} show a similar failure of classical ansatz to capture the exact maximum energy of the associated Hamiltonian and the former even studies non-commutative Sum-of-Squares lower bounds for the model. Our approach, which is captured by non-commutative Sum-of-Squares, is similarly not amenable to exact results of this form. This is to say, while our ansatz provides intuition that semirandom interactions make the problem easy within constant-sized energy gaps, our results do not suffice for quantitative sharpness in subconstant regimes, which remain an interesting avenue of study in understanding the behavior of ground and thermal states of quantum many-body systems and potential quantum supremacy. 

We can additionally compare our results to the works \cite{BrandaoH13a, AlevJT19}, which show Hamiltonians whose interaction hypergraph is sufficiently dense or a spectral expander can have their maximum energy state approximated by a product state, making them classically approximable. While these results do apply in our dense regimes, they ultimately get the weaker guarantee of just approximation rather than refutation, which is automatic in the semirandom setting.

\parhead{Fourier analytic view of Hamiltonians and the two-basis paradigm.} A recent work \cite{MaN25} shows that a problem called $XZ$-Quantum-$6$-$\SAT$ is $\QMA_1$-hard. $XZ$-Quantum-$6$-$\SAT$ is constructed by taking two instances of a projection-like $6$-CSP and prescribing them (Fourier analytically) in only the standard $Z$-basis and Hadamard $X$-basis. In the same vein, this work studies Hamiltonian $k$-$\XOR$ through the lens of a Fourier analytic generalization of classical $k$-$\XOR$ and we explicitly look at $XZ$-$k$-$\XOR$, which mirrors their $XZ$-Quantum-$6$-$\SAT$. Our notion of Hamiltonian $k$-$\XOR$ can be viewed as a prescription of classical $k$-$\XOR$ onto arbitrary Pauli operators, rather than just two bases, and contains the special two-basis case $XZ$-$k$-$\XOR$. Their work nonetheless gives evidence that the interactions between Pauli bases, even just two of them, create $\QMA_1$-hardness in the maximum energy space, which may hint that general approximation in Hamiltonian $k$-$\XOR$ is $\QMA_1$-hard. Nonetheless, our refutation algorithm shows semirandomness and density bring the problem down to $\PTIME$.

\parhead{Non-commutative CSPs and quantum $k$-$\XOR$.} Apart from Hamiltonians, there exists a second natural notion of ``quantum generalization of CSPs'' called non-commutative CSPs that is distinct from our work, and we refer the reader to \cite{HastingsO22, CulfMS24, MousaviS25, TallerV25} for some recent refutation/approximation results in this direction. In particular, we draw attention to \cite{HastingsO22}, which similarly considers non-commutative Sum-of-Squares for the problem of refuting the fermionic SYK model. While the non-commutative Sum-of-Squares hierarchy is also used here, the technique ends up looking much different due to differences between fermionic and standard Hamiltonians.

\parhead{Kikuchi matrices in quantum computation.} A recent line of work \cite{SchmidhuberOKB25, GuptaHOS25} uses Kikuchi matrices for classical and quantum algorithms for the planted version of the classical $k$-$\XOR$ problem, attempting to categorize potential quantum speedups for this problem. Rather than devising a quantum algorithm using the classical Kikuchi matrix, we instead develop a \textit{classical} algorithm for a quantum, or perhaps more aptly, Hamiltonian Kikuchi matrix, in order to solve the related problem of Hamiltonian refutation.

\subsection{Structure of the paper}

The rest of the paper is organized as follows. In \Cref{sec:prelims}, we introduce notation and some standard facts we use. In \Cref{sec:refutation}, we prove \Cref{thm:refutation}, our Hamiltonian refutation algorithm, for even $k$, which also serves as a self-contained overview of our main proof ideas. In \Cref{sec:oddcase}, we prove \Cref{thm:refutation} for odd $k$. In \Cref{sec:ncsos}, we prove the corresponding non-commutative Sum-of-Squares bounds of \Cref{thm:ncsosforr} and \Cref{thm:randomxor}.

\subsection{Acknowledgements}

We thank Peter Manohar, Ryan O'Donnell, and David Gosset for helpful discussions. NK was supported by an NSF Graduate Research Fellowship under grant number DGE-2140004, and both NK and CN were supported by NSF CAREER award IIS2541127.

\section{Preliminaries}
\label{sec:prelims}

We let $[n]$ denote the set $\{1, \dots, n\}$, $\binom{[n]}{\ell}$ denote the set of subsets of $[n]$ of size $\ell$, and $\oplus$ be the symmetric difference operator. For a rectangular matrix $A \in \C^{m \times n}$, we let $A^\dagger$ denote its conjugate transpose. We let $\norm{A}{2} \coloneqq \max_{\substack{x \in \C^m, y \in \C^n\\ \norm{x}{2} = \norm{y}{2} = 1}} \abs{x^{\dagger} A y}$ denote the spectral norm of $A$. For $A \in \C^{n \times n}$ we let $\tr(A)$ be the trace of $A$, i.e., $\sum_{i = 1}^n A_{i,i}$. 

\subsection{Quantum computation}

We represent pure states on $n$ qubits as unit vectors $\ket{\psi} \in (\C^2)^{\otimes n}$ and assume standard bra-ket notation. The rank-$1$ density operator for a pure state $\ket{\psi}$ we write as $\psi = \ketbra{\psi}{\psi}$. More generally, we have:

\begin{definition}[Density operator]
    We call $\rho \in \End((\C^2)^{\otimes n})$ a density matrix on $n$ qubits if it satisfies: (1) $\tr(\rho) = 1$ and (2) $\rho \succeq 0$. As a consequence, we may write $\rho$ as
    \begin{equation*}
        \rho = \sum_{i=1}^{2^n} \lambda_i \cdot \ketbra{\psi_i}{\psi_i}\mcom
    \end{equation*}
    where $\ket{\psi_1}, \dots, \ket{\psi_{2^n}}$ form an orthonormal basis of pure states in $(\C^2)^{\otimes n}$ and $\lambda_1, \dots, \lambda_{2^n}$ form a probability distribution in that $\sum_{i=1}^{2^n} \lambda_i = 1$ and $\lambda_i \geq 0$ for all $i = 1,\dots, 2^n$.
\end{definition}

 \begin{definition}[Local Hamiltonians]
     We call $\bH \in \End((\C^2)^{\otimes n})$ a $k$-local Hamiltonian on $n$ qubits if it can be written as
     \begin{equation*}
         \bH = \frac{1}{m}\sum_{i=1}^m \bH_i\mcom
     \end{equation*}
     where each $\bH_i$ is a Hermitian term acting non-trivially on at most $k$ of the $n$ qubits.
 \end{definition}
The Hamiltonian operators in this work will be defined in terms of Pauli matrices, and the following facts about Paulis are used.
\begin{definition}[Pauli operators]
    We define $X, Y, Z \in \C^{2 \times 2}$ as
    \begin{equation*}
        X = \begin{bmatrix}
            0 & 1\\
            1 & 0
        \end{bmatrix},\;\;\;\;
        Y = \begin{bmatrix}
            0 & -i\\
            i & 0
        \end{bmatrix},\;\;\;\;
        Z = \begin{bmatrix}
            1 & 0\\
            0 & -1
        \end{bmatrix}.
    \end{equation*}
    An $n$-qubit Pauli operator is a $2^n$-dimensional operator from $\{\Id, X, Y, Z\}^{\otimes n}$. 
\end{definition}

\begin{fact}
    The single-qubit Pauli operators are involutory, $X^2 = Y^2 = Z^2 = \Id$, Hermitian, i.e. $X^\dagger = X$, and pairwise anticommute, i.e. $XZ = -ZX$. As a result, all Pauli operators either commute or anticommute.
\end{fact}

The most important fact about the Pauli operators is that they form a basis for Hermitian operators, which allows us to decompose any Hamiltonian as a sum of Paulis.

\begin{definition}[Pauli basis]
    \label{def:paulibasis}
    The Pauli operators $\{\Id, X, Y, Z\}^{\otimes n}$ under scalar field $\C$ form an orthogonal basis for operators in $\End((\C^2)^{\otimes n})$ under the inner product $\langle P, Q \rangle = \frac{1}{2^n} \tr(P^\dagger Q)$. As a consequence, we can write any $\bH \in \End((\C^2)^{\otimes n})$ as
    \begin{equation*}
        \bH = \sum_{P \in \{\Id, X, Y, Z\}^{\otimes n}} \langle P, \bH \rangle \cdot P\mper
    \end{equation*}
    The coefficients $\langle P, \bH \rangle$ are called the Pauli coefficients, and we occasionally shorthand them as $\hat{\bH}(P)$. If all Pauli coefficients are real, then $\bH$ is Hermitian, and more generally the real span of $\{\Id, X, Y, Z\}^{\otimes n}$ is exactly the set of Hermitian operators over $(\C^2)^{\otimes n}$.
\end{definition}

\begin{definition}[Pauli degree]
    For a Pauli operator $P \in \{\Id, X, Y, Z\}^{\otimes n}$, we define $\deg(P) = k$ if $P$ acts with non-identity operators on \emph{exactly} $k$ qubits. For a general matrix $M \in \End((\C^2)^{\otimes n})$ we define
    \begin{align*}
        \deg(M) = \max\{\deg(P) \mid \hat{M}(P) \neq 0\}\mcom
    \end{align*}
    with the convention $\deg(M) = 0$ for $0$. We denote by $\mathcal{P}_k(n)$ the set of degree-$k$ Pauli operators on $n$ qubits and more generally $\calP(n) = \{\Id, X, Y, Z\}^{\otimes n}$. Note that a $k$-local Hamiltonian decomposes into degree-at-most-$k$ Pauli operators as above. Given two operators $P, Q \in \calP(n)$, we say $P \sqsubseteq Q$ if $P_i \neq \Id \implies Q_i = P_i$ for all indices $i \in [n]$---i.e., whenever $P$ is a non-identity operator on qubit $i$, $Q$ acts identically on that qubit.
\end{definition}

\subsection{CSPs and Hamiltonians}

In this section we define, in greater detail, the Hamiltonian $k$-$\XOR$ problem and prove a few auxiliary facts about the resulting family of Hamiltonians. We start by reviewing the definition of a classical Boolean $k$-CSP.

\begin{definition}[CSP]
    An $n$-variable Boolean CSP \textit{instance} $\mathcal{I}$ is a pair $(\mathcal{H}, \{\varphi_C\}_{C \in \mathcal{H}})$ where $\mathcal{H}$ is a hypergraph over $[n]$ and each $\varphi_C$ is a function $\{\pm 1\}^C \to \{0,1\}$. We call the elements $C \in \calH$ constraints. An \textit{assignment} is a vector $x \in \{\pm 1\}^{n}$ and we say an assignment $x$ satisfies a constraint $C \in \calH$ if $\varphi_C(x_C) = 1$. The classical \textit{value} of the instance under an assignment $x$ is simply the fraction of constraints it satisfies
    \begin{equation*}
        \val(\mathcal{I}, x) = \frac{1}{\abs{\mathcal{H}}} \sum_{C \in \mathcal{H}} \varphi_C(x_C)\mper
    \end{equation*}
    The value of the instance overall is simply $\val(\mathcal{I}) = \max_{x \in \{\pm 1\}^n} \val(\mathcal{I}, x)$.
\end{definition}

It is useful to view the value of a CSP Fourier analytically as a function $\{\pm1\}^n \to [0,1]$.

\begin{definition}[Fourier decomposition]
    Given a function $f : \{\pm 1\}^n \to \R$, the canonical Fourier decomposition of $f$ is
    \begin{equation*}
        f = \sum_{S \subseteq [n]} \hat{f}(S) \cdot \chi_S \mcom
    \end{equation*}
    where $\chi_S(x) = \prod_{i \in S} x_i$. The Fourier decomposition is unique (see \cite{ODonnell14}).
\end{definition}

It is well-known that any CSP can be written as a Hamiltonian with matching value through simulation in a fixed basis, and we prove this fact for clarity in \Cref{sec:ncsos}.

\begin{definition}[Standard basis Hamiltonian embeddings of a CSP]
    \label{def:hamiltoniansimulation}
    Given a CSP instance $\calI$, consider the Fourier decomposition of $\val(\calI, x)$ as a function $\val_\calI(x)$ of $x \in \{\pm1\}^n$,
    \begin{equation*}
        \val_\calI(x) = \sum_{S \subseteq [n]} \widehat{\val}_\calI(S) \cdot \chi_S(x)\mper
    \end{equation*}
    The $Z$-basis Hamiltonian simulation $\bH_\calI$ of $\calI$ is
    \begin{equation*}
        \bH_\calI = \sum_{S \subseteq [n]} \widehat{\val}_\calI(S) \cdot Z_S\mcom
    \end{equation*}
    where $Z_S = \prod_{i \in S} Z_i$. An analogous definition holds for the $X$-basis.
\end{definition}

The natural extension of this definition from one-basis to two-bases yields the following definition of $XZ$-Hamiltonians.

\begin{definition}[$XZ$-Hamiltonians for CSPs]
    \label{def:xzhamiltonian}
    Given a pair of CSP instances $\calI = (\calI_X, \calI_Z)$, consider the Fourier decompositions of $\val(\calI_X, x)$ and $\val(\calI_Z, x)$ as functions in $x \in \{\pm1\}^n$,
    \begin{align*}
        & \val_{\calI_X}(x) = \sum_{S \subseteq [n]} \widehat{\val}_{\calI_X}(S) \cdot \chi_S(x)\\
        & \val_{\calI_Z}(x) = \sum_{S \subseteq [n]} \widehat{\val}_{\calI_Z}(S) \cdot \chi_S(x)\mper
    \end{align*}
    The $XZ$-Hamiltonian $\bH_\calI$ associated with $\calI$ is
    \begin{align*}
        \bH_\calI = \frac{\abs{\calH_X}}{\abs{\calH_X} + \abs{\calH_Z}} \bH_{\calI_X} + \frac{\abs{\calH_Z}}{\abs{\calH_X}+\abs{\calH_Z}} \bH_{\calI_Z}\mper
    \end{align*}
\end{definition}

We now establish a couple basic facts regarding the Hamiltonian $k$-$\XOR$ model defined in \Cref{def:quantumxor}.

\begin{fact}
    Let $\calI = (\calH, \{(P_C, b_C)\}_{C \in \calH})$ be a Hamiltonian $k$-$\XOR$ instance on $n$ qubits defining $\bH_\calI$. Then $\lambda_{\mathrm{max}}(\bH_\calI) \geq \frac{1}{2}$. Moreover, there is a product state achieving energy $\ge \frac{1}{2}$.
\end{fact}
\begin{proof}
It's easy to verify that the maximally mixed state has energy exactly $1/2$, which is a mixture over all standard-basis (product) states. The second claim follows by the probabilistic method.
\end{proof}

\begin{fact}
    \label{fact:quantumxorconcentration}
    Let $\mathbfcal{I} = (\calH, \{(P_C, \bb_C)\}_{C \in \calH})$ be a semirandom Hamiltonian $k$-$\XOR$ instance on $n$ qubits defining $\bsH_{\mathbfcal{I}}$. Then if $\abs{\calH} \geq ((n+1)\log 2+\log(1/\delta)) /2\varepsilon^2$ with probability at least $1-\delta$, $\lambda_{\mathrm{max}}(\bsH_{\mathbfcal{I}}) \leq \frac{1}{2} + \varepsilon$.
\end{fact}

\begin{proof}
The result requires the (standard) matrix concentration bound listed in Lemma~\ref{lem:matrixchernoff}.
    Considering the sequence $\{P_C\}_{C \in \calH}$ and $\bsH = \sum_{C \in \calH} \bb_C P_C$, \Cref{lem:matrixchernoff} yields
    \begin{equation*}
        \Pr\sbra{\lambda_{\mathrm{max}}(\bsH) \geq 2\varepsilon \abs{\calH}} \leq 2^{n+1} \exp\left(-2\varepsilon^2 \abs{\calH}\right)\mper
    \end{equation*}
    We use here that the Pauli operators are involutory to get $\Var(\bsH) = \norm{\sum_{C \in \calH} \Id_{2^n}}{2} = \abs{\calH}$. By letting $\abs{\calH} \geq O\left(\frac{n + \log(1/\delta)}{\varepsilon^2}\right)$ for any $\delta > 0$ of our choosing, this probability can be bounded by at most $\delta$.
\end{proof}

\begin{lemma}[Matrix Rademacher/Gaussian series \cite{Tropp15}]
        \label{lem:matrixchernoff}
        Let $\{A_i\}_{i \in [m]}$ be a sequence of Hermitian complex matrices in $\C^{d \times d}$ and
        \begin{equation*}
            \boldsymbol{A} = \sum_{i=1}^m \bb_i A_i\mcom
        \end{equation*}
        where $\bb_i$ are i.i.d. from $\{\pm 1\}$ or $\calN(0, 1)$ for $i = 1,..., m$. Define the matrix variance $\Var(\boldsymbol{A}) \coloneqq \norm{\sum_{i=1}^m A_i^2}{2}$. Then for any $t > 0$
        \begin{equation*}
            \Pr\sbra{\lambda_{\mathrm{max}}(\boldsymbol{A}) \geq t} \leq 2d \exp\left(\frac{-t^2}{2\Var(\boldsymbol{A})}\right)\mper
        \end{equation*}
    \end{lemma}

\subsection{Sum-of-Squares hierarchy}
\label{sec:sos}

A distribution $\mu$ over the boolean hypercube is a function $\mu: \{\pm 1\}^n \to \R$ such that $\mu \ge 0$ and $\sum_{x \in \{\pm 1\}^n} \mu(x) = 1$. Every distribution $\mu$ defines a natural map from the space of polynomials over $n$-variables to the reals, $\E: \R[x_1, \ldots, x_n] \to \R$ defined by $\E_\mu[f] = \sum_x \mu(x) \cdot f(x)$.

\begin{definition}
\label{def:pseudo-expectation}
A pseudo-distribution is a generalization of the notion of a distribution; however, unlike a distribution, it need not be non-negative everywhere. Similarly, every pseudo-distribution $\mu$ defines a natural map from the space of polynomials over $n$-variables to the reals, $\pE: \R[x_1, \ldots, x_n] \to \R$ defined by $\pE_\mu[f] = \sum_x \mu(x) \cdot f(x)$. A degree-$d$ pseudo-expectation $\pE_\mu$ is one that satisfies the following three properties:
\begin{enumerate}
    \item (Normalization) $\pE_\mu[1] = 1$.
	\item (Booleanity) For any $x_i$ and any polynomial $p$ of degree $\leq d-2$, $\pE_\mu[p (x_i^2 - 1)] = 0$. 
	\item (Positivity) For any polynomial $p$ of degree at most $d/2$, $\pE_\mu[p^2] \geq 0$. 
\end{enumerate} 
For any polynomial $p$ of degree at most $d$, the degree-$d$ Sum-of-Squares relaxation value $\SoS_d(p)$ is defined as
\begin{align*}
    \max_\mu\;\;\;\;&\pE_{\mu}\sbra{p} = \sum_{x \in \{\pm1\}^n} \mu(x) \cdot p(x)\\
    \text{subject to}\;\;\;\;& \mu \text{ a degree-$d$ pseudo-distribution over $\{\pm1\}^n\mper$}
\end{align*}
\end{definition}

In the context of CSPs, we write $\SoS_d(\calI)$ to denote the Sum-of-Squares relaxation value of the instance $\calI$'s Fourier decomposition over the hypercube.

The characterization of Sum-of-Squares through pseudo-expectations is particularly fruitful from an algorithmic perspective, but the paradigm earned its moniker as a way to prove the non-negativity of functions by expressing them as sums-of-squares. While we do not detail this dual view of Sum-of-Squares, we borrow the following definition in order to talk about Sum-of-Squares proofs. See \cite{BarakS14} for a more formal treatment of this dual view.

\begin{definition}[Sum-of-Squares proofs over the hypercube]
    For any polynomial $p : \{\pm1\}^n \to \R$ with degree-at-most-$d$, we say there is a degree-$d$ Sum-of-Squares proof that $p \geq 0$ if for every degree-$d$ pseudo-expectation $\pE_\mu$, $\pE_\mu\sbra{p} \geq 0$.
\end{definition}

Finally, in the process of proving positivity of a pseudo-expectation operator, it is often fruitful to go via connection to its moment matrix.

\begin{definition}[Moment matrix]
    \label{def:moment}
    Given a degree-$d$ pseudo-expectation operator $\pE_\mu$ we define its moment matrix $M_{\pE_\mu} \in \R^{\binom{[n]}{\leq d/2} \times \binom{[n]}{\leq d/2}}$ via $M_{\pE_\mu}(S, T) = \pE_\mu\sbra{x_{S \oplus T}}$.
\end{definition}

\begin{fact}[Positivity over the moment matrix]
    \label{fact:moment}
    Let $\pE_\mu$ be a degree-$d$ pseudo-expectation operator. Then positivity of $\pE_\mu$ is equivalent to $M_{\pE_\mu} \succeq 0$.
\end{fact}

\begin{proof}
Every polynomial $p$ of degree at most $d/2$ can be written in the standard
multilinear basis as
\begin{align*}
    p(x)=\sum_{S:\abs{S}\leq d/2} a_S x_S\mper
\end{align*}
Using linearity and the Boolean relations $x_i^2=1$, we have
\begin{align*}
    \pE_\mu[p^2]
    &= \sum_{S,T} a_Sa_T \pE_\mu[x_Sx_T] \\
    &= \sum_{S,T} a_Sa_T \pE_\mu[x_{S\oplus T}] \\
    &= a^\top M_{\pE_\mu} a.
\end{align*}
Consequently, $\pE_\mu[p^2]\geq 0$ for every polynomial $p$ of degree
at most $d/2$ if and only if $a^\top M_{\pE_\mu}a\geq 0$ for every
vector $a$, which is equivalent to $M_{\pE_\mu}\succeq 0$.
\end{proof}

\subsubsection{Non-commutative Sum-of-Squares}

The quantum analogue of a distribution is a density matrix $\rho$ such that (1) positivity $\rho \succeq 0$ and (2) unit trace, $\tr(\rho) = 1$. To define a pseudo-distribution $\rho$ we relax the positivity condition and consider Hermitian matrices of unit trace.


\begin{definition}[Pseudo-density operators] \label{def:pseudo-density}
A degree-$d$ pseudo-expectation $\pE_\rho$ operator defined for pseudo-density operator $\rho$ on $(\C^2)^{\otimes n}$ is the linear operator $\bH \mapsto \tr(\bH\rho)$ satisfying the properties:
\begin{enumerate}
    \item (Normalization) $\pE_\rho[\Id] = 1$.
    \item (Positivity) For any operator $M \in \End((\C^2)^{\otimes n})$ of degree at most $d/2$, $\pE_\rho[M^\dagger M] \geq 0$. 
\end{enumerate} 

The degree-$d$ non-commutative Sum-of-Squares relaxation value $\ncSoS_d(\bH)$ of a degree-$d$ Hermitian operator $\bH \in \End((\C^2)^{\otimes n})$ is defined by the following maximization over degree-$d$ pseudo-density operators:
\begin{align*}
    \max_\rho\;\;\;&\tr(\bH\rho) = \sum_{i =1}^{2^n} \lambda_i(\rho) \cdot \tr(\bH\ketbra{\psi_i}) = \pE_{i \sim \lambda(\rho)} \sbra{\bra{\psi_i}\bH \ket{\psi_i}}\\
    \text{subject to}\;\;\;\;& \rho \text{ a degree-$d$ pseudo-density operator on $(\C^2)^{\otimes n}\mper$}
\end{align*}
\end{definition}


The notion of non-commutative Sum-of-Squares proofs and moment matrices carries over in a straightforward manner.

\begin{definition}[Non-commutative Sum-of-Squares proofs]
    For any Hermitian operator $M \in \End((\C^2)^{\otimes n})$ of degree-at-most-$d$, we say there is a degree-$d$ non-commutative Sum-of-Squares proof that $M \succeq 0$ if for every degree-$d$ pseudo-expectation $\pE_\rho$, $\pE_\rho\sbra{M} \geq 0$.
\end{definition}

\begin{definition}[Moment matrix]
    \label{def:ncmoment}
    Given a degree-$d$ pseudo-expectation operator $\pE_\rho$ we define its moment matrix $M_{\pE_\rho} \in \C^{\calP_{\leq d/2}(n) \times \calP_{\leq d/2}(n)}$ via $M_{\pE_\rho}(P, Q) = \pE_\rho\sbra{P^\dagger Q}$.
\end{definition}

\begin{fact}[Positivity over the moment matrix]
    \label{fact:ncmoment}
    Let $\pE_\rho$ be a degree-$d$ pseudo-expectation operator. Then positivity of $\pE_\rho$ is equivalent to $M_{\pE_\rho} \succeq 0$.
\end{fact}

\begin{proof}
The proof is similar to that of \Cref{fact:moment}. Write any $R\in \End((\C^2)^{\otimes n})$ with $\deg(R) \leq d/2$ as \(
R=\sum_{P\in\calP_{\leq d/2}(n)} a_P P, \)
where $a_P\in\C$. Then
\begin{align*}
\pE_\rho[R^\dagger R]
= \sum_{P,Q} \overline{a_P}a_Q\pE_\rho[P^\dagger Q] = a^\dagger M_{\pE_\rho}a.
\end{align*}
Thus $\pE_\rho[R^\dagger R]\geq 0$ for every such $R$ if and only if
$a^\dagger M_{\pE_\rho}a\geq 0$ for every $a\in\C^{\calP_{\leq d/2}(n)}$, which is equivalent to
$M_{\pE_\rho}\succeq 0$.
\end{proof}

\subsection{High-dimensional expanders and quantum codes}
\label{sec:codes}

In this section we define the basics of high-dimensional expanders and their relation to quantum codes. While we do not explicitly import any machinery from the high-dimensional expansion literature, we introduce in \Cref{sec:ncsos} a correspondence between $XZ$-$k$-$\XOR$ and signed quantum codes that parallels the standard correspondence between chain complexes and quantum codes.
The primary high-dimensional object we reference is the chain complex.

\begin{definition}[Chain complex]
    A chain complex is a sequence of vector spaces $\calC_0, \dots, \calC_d$ and linear operators $\delta_0, \dots, \delta_{d-1}$ where each $\delta_i : \calC_i \to \calC_{i+1}$ and satisfies $\delta_{i+1}\delta_i = 0$. We call the $\delta_i$ \textit{coboundary operators} and call their transposes $\partial_{i+1} = \delta_i^\top$ \textit{boundary operators} and notate the chain as follows:
    \begin{align*}
        \calC_0 \underset{\partial_1}{\overset{\delta_0}{\rightleftarrows}} \calC_1 \underset{\partial_2}{\overset{\delta_1}{\rightleftarrows}} \dots \underset{\partial_d}{\overset{\delta_{d-1}}{\rightleftarrows}} \calC_d\mper
    \end{align*}
    In this work, we consider only the special case of 3-term chain complexes for which the vector spaces are function spaces $\calC_i \simeq \F_2^{n_i}$:
    \begin{align*}
        \F_2^{m_0} \underset{\partial_1}{\overset{\delta_0}{\rightleftarrows}} \F_2^{m_1} \underset{\partial_2}{\overset{\delta_1}{\rightleftarrows}} \F_2^{m_2}\mper
    \end{align*}
\end{definition}
An in-depth explanation of notions of expansion for chain complexes can be found in~\cite{HopkinsL22}. This last restriction of a chain complex ($3$-term, binary) is well-known to be equivalent to a quantum CSS code.

\begin{definition}[Quantum CSS codes]
    A quantum CSS code $\calC$ is specified by a pair of $X$ and $Z$-parity check matrices $(H_X, H_Z)$ with $H_X : \F_2^n \to \F_2^{m_X}$ and $H_Z : \F_2^n \to \F_2^{m_Z}$ satisfying the CSS condition $H_X H_Z^\top = 0$.
    
    The codespace of the quantum CSS code is the subspace of $(\C^2)^{\otimes n}$ that is the kernel of the following code Hamiltonian. Convert each row of $H_X$ from a vector in $\F_2^n$ to $C \subseteq [n]$ and collect them into a multiset $\calH_X$, and do the same for $H_Z$ and $\calH_Z$. The code Hamiltonian for $\calC$ is given by
    \begin{align*}
        \bH_\calC &= \frac{1}{\abs{\calH_X}+\abs{\calH_Z}} \sum_{C \in \calH_X} \frac{\Id - X_C}{2} + \frac{1}{\abs{\calH_X}+\abs{\calH_Z}} \sum_{C \in \calH_Z} \frac{\Id - Z_C}{2}\\
        &\coloneqq \frac{\abs{\calH_X}}{\abs{\calH_X} + \abs{\calH_Z}} \bH_{\calI_X} + \frac{\abs{\calH_Z}}{\abs{\calH_X}+\abs{\calH_Z}} \bH_{\calI_Z}\mper
    \end{align*}
    Finally, we say a family of quantum CSS codes is low-density parity check (LDPC) if the row and column weights of their parity check matrices are bounded by a constant independent of $n$.
\end{definition}

Observe that the code Hamiltonian of a quantum CSS code is an instance of the Hamiltonian $k$-$\XOR$ problem with the additional condition that every sign $\{b_C\}_{C \in \calH}$ is $-1$. In this work, we use the energy maximization convention consistent with classical CSP theory, so we instead instantiate the Hamiltonian $k$-$\XOR$ definition with $\Id - \bH_\calC$ instead, in which case all the signs are $+1$ instead. 

As mentioned previously, we have suggestively written both $2$-dimensional (binary) chain complexes and quantum CSS codes with the symbol $\calC$. This is because there is a well-known equivalence between the two.

\begin{fact}[Chain complex $\iff$ quantum CSS codes]
    \label{fact:chaincodes}
    Let $\calC$ be a quantum CSS code specified by its parity check matrices $(H_X, H_Z)$ with $H_X : \F_2^n \to \F_2^{m_X}$ and $H_Z : \F_2^n \to \F_2^{m_Z}$. Then we can associate it with the following chain complex,
    \begin{align*}
        \F_2^{m_X} \underset{H_X}{\overset{H_X^\top}{\rightleftarrows}} \F_2^n \underset{H_Z^\top}{\overset{H_Z}{\rightleftarrows}} \F_2^{m_Z}\mper
    \end{align*}
\end{fact}

This fact follows by observing the CSS condition $H_X H_Z^\top = 0$ translates exactly to the chain complex condition $\delta_1\delta_0 = 0$. It is not hard to check that it, importantly, corresponds to the condition their code Hamiltonians are commuting as well.

We introduce here a third equivalence to the chain complex and quantum CSS code paradigm. Namely, we will transform the parity checks into $\XOR$ instances to craft an $XZ$-$k$-$\XOR$ instance.

\begin{fact}[Parity checks as $\XOR$ instances]
        \label{fact:parity}
        Let $H : \F_2^n \to \F_2^m$ be a parity check matrix. Then for every row we can define the set of non-zero entries $C \subseteq [n]$ and collect them into a multiset $\calH$. There is then a trivial inhomogeneous $k$-$\XOR$ instance $\calI = (\calH, (+1, \dots, +1))$ associated when $k$ is the maximum density $\max_{C \in \calH} \abs{C}$ of each row.
    \end{fact}
    
    Note that as constructed, the $k$-$\XOR$ instance associated with a parity check matrix is trivially satisfiable by the codespace which, at the very least, always contains $0$.

    Now, given a quantum CSS code, we can associate the parity checks $H_X, H_Z$ with two classical $k$-$\XOR$ instances $(\calI_X, \calI_Z)$ via the natural embedding. Furthermore this immediately gives an associated $XZ$-$k$-$\XOR$ instance.

\subsection{Binomial coefficient inequalities}

\begin{fact}[Binomial-ratio estimates]
\label{fact:binomest}
Let $n,\ell,q$ be positive integers satisfying
\(
    q \leq \ell \leq \frac{n}{2}.
\)
Then
\begin{align*}
    \left(\frac{\ell-q+1}{n}\right)^q
    \leq
    \frac{\binom{n}{\ell-q}}{\binom{n}{\ell}}
    \leq
    \left(\frac{\ell}{n-\ell+1}\right)^q,
\quad\text{and}\quad
    \left(\frac{n-\ell-q+1}{n}\right)^q
    \leq
    \frac{\binom{n-q}{\ell}}{\binom{n}{\ell}}
    \leq
    \left(\frac{n-\ell}{n-q+1}\right)^q.
\end{align*}
Consequently, the following balanced estimate holds:
\begin{align*}
    \binom{2q}{q}
    \frac{\binom{n-2q}{\ell-q}}{\binom{n}{\ell}}
    \geq
    \left(\frac{\ell}{n}\right)^q.
    \label{eq:balanced-binomest}
\end{align*}

\end{fact}
\section{\texorpdfstring{Certifying Semirandom $k$-$\XOR$ Hamiltonians}{Certifying Semirandom XOR Hamiltonians even}}
\label{sec:refutation}

In this section, we first prove \Cref{thm:refutation} for even values of $k$ using the Kikuchi matrix method. Our certification algorithm largely follows the framework of \cite{GuruswamiKM22, HsiehKM23} for refuting semirandom $k$-$\XOR$ instances, with our main new ingredient being a spectral certificate in the form of a novel Kikuchi matrix built for Hamiltonians as opposed to classical CSPs, which we introduce here.

The proof for odd values of $k$ is considerably more involved due to a lack of symmetry in the adjacency matrices. The proof is deferred to~\Cref{sec:oddcase}.

\begin{definition}[Even-arity Kikuchi matrix for Hamiltonian $k$-$\XOR$]
    \label{def:kikuchimatrix}
    Fix $k \in \N$ even and $k/2 \leq \ell \leq n/2$. Let $P \in \{\Id, X, Y, Z\}^{\otimes n}$ be a weight-$k$ Pauli operator. We define the following adjacency matrix on the set of degree-$\ell$ Pauli operators, $\mathcal{P}_\ell(n)$:
    \begin{equation*}
        A_P(Q, R) = \begin{cases}
            1 & Q^\dagger R = P \text{ and } \abs{\supp(Q) \cap \supp(R)} = \ell-k/2\\
            0 & \text{otherwise}
        \end{cases}
    \end{equation*}
    For a present edge, we write $Q \xleftrightarrow{P} R$. Let $\calI = (\calH, \{(P_C, b_C)\}_{C \in \calH})$ be a Hamiltonian $k$-$\XOR$ instance on $n$ qubits. The level-$\ell$ Kikuchi matrix of $\bH_\calI$ is then defined by
    \begin{equation*}
        K_{\bH_\calI} = \underbrace{\left(\sum_{C \in \calH} b_C A_{P_C}  \right)}_{\coloneqq A^*_{\bH_\calI}} \otimes \hspace{0.05em}\Id_{2^n} \mper
    \end{equation*}
    The notation $A^*_{\bH_\calI}$ will be used for the linear combination signed by $b_C$, whereas the unsigned adjacency matrix will be $A_{\bH_\calI} := \sum_{C \in \calH}  A_{P_C}$, the adjacency matrix of the Kikuchi graph.
\end{definition}

Note that the matrix $K_{\bH_\calI}$ is defined with a $\otimes\hspace{0.05em} \Id_{2^n}$ term because we are lifting the adjacency matrix to apply Pauli moments.
To ultimately apply a matrix concentration bound, the Kikuchi matrix needs to be Hermitian, which we note here.

\begin{observation}
    \label{obs:symmetric}
    Let $\calI = (\calH, \{(P_C, b_C)\}_{C \in \calH})$ be a Hamiltonian $k$-$\XOR$ instance on $n$ qubits describing $\bH_\calI$. Then $K_{\bH_\calI}$ is symmetric and the underlying graph is undirected.
\end{observation}


Previously it was known how to build Kikuchi matrices for any finite Abelian group~\cite{KocurekM25}. Our Kikuchi matrix, on the other hand, is built on the non-Abelian group (or rather a phase-free representative set) of the Pauli operators. Nevertheless, the above shows we only ever take edges between commuting operators due to the specific nature of Pauli operators.

\subsection{\texorpdfstring{Step 1: Expressing $\bra{\psi}\bH_\calI\ket{\psi}$ as a quadratic form of a Kikuchi matrix}{Step 1}}

As the core of the Kikuchi matrix method, we show how the Kikuchi matrix captures Hamiltonian optimization directly in its quadratic form. 
In the following, for a state $\ket{\psi}$, we use the notation $\psi^{\odot \ell} \in ((\C^{2})^{\otimes n})^{\mathcal{P}_\ell(n)}$ to refer to the vector
\[
\psi^{\odot \ell} \coloneqq \bigoplus_{P \in \mathcal{P}_\ell(n)} P \ket{\psi}\mper
\]
which can be interpreted as the concatenation of copies of $\ket{\psi}$ each conjugated by a term $P \in \mathcal{P}_\ell(n)$.

\begin{observation}
    \label{obs:kikuchirelaxation}
    Let $\calI = (\calH, \{(P_C, b_C)\}_{C \in \calH})$ be a Hamiltonian $k$-$\XOR$ instance on $n$ qubits describing $\bH_\calI$. Let $\ket{\psi} \in (\C^2)^{\otimes n}$ be a pure state.
    Then
    \begin{equation*}
        \bra{\psi}\bH_\calI \ket{\psi} = \frac{1}{2} + \frac{1}{2\abs{\calH}\Delta}(\psi^{\odot \ell})^\dagger K_{\bH_\calI} \psi^{\odot \ell}\mcom
    \end{equation*}
    where $\Delta = \binom{k}{k/2} \binom{n-k}{\ell-k/2} 3^{\ell-k/2}$.
\end{observation}

\begin{proof}
    For a Kikuchi matrix $K_P = A^*_P \otimes \Id_{2^n}$ of a basis Pauli operator $P$ we can write
    \begin{align*}
        (\psi^{\odot \ell})^\dagger K_P \psi^{\odot \ell} &= \sum_{Q, R \in P_\ell(n)} \bra{\psi}{Q^\dagger R} \ket{\psi} \cdot \mathbf{1}\{Q \xleftrightarrow{P} R\} \\
        &= \ev{P}{\psi} \cdot \underbrace{\left| \left\{ Q, R \in P_\ell(n) : Q \xleftrightarrow{P} R \right\} \right|}_{:= \Delta}
    \end{align*} 
    What remains is to calculate $\Delta$ by counting how many pairs $(Q, R) \in \mathcal{P}_\ell(n) \times \mathcal{P}_\ell(n)$ have $Q \xleftrightarrow{P} R$. As in \Cref{obs:symmetric}, we note that $Q$ and $R$ must both act non-trivially on $\ell-k/2$ qubits and disjointly on the $k$ qubits $P$ acts on. $\Delta$ can be verified by observing that there are (1) $\binom{k}{k/2}$ choices of which qubits of $P$ that $Q$ acts non-trivially on, (2) $\binom{n-k}{\ell-k/2}$ choices of which qubits not of $P$ that $Q$ and $R$ both act non-trivially on, and (3) $3^{\ell-k/2}$ non-trivial operators on these registers. Then, by linearity
    \begin{align*}
        (\psi^{\odot \ell})^\dagger K_{\bH_\calI} \psi^{\odot \ell} &= \sum_{C \in \calH} b_C \cdot (\psi^{\odot \ell})^\dagger K_{P_C}\psi^{\odot \ell}\\
        &= \Delta\sum_{C \in \calH} b_C \bra{\psi}P_C\ket{\psi}\mper
    \end{align*}
    Since $\bH_\calI = \frac{1}{\abs{\calH}} \sum_{C \in \calH} \frac{\Id + b_C P_C}{2}$, we conclude $\bra{\psi}\bH_\calI \ket{\psi} = \frac{1}{2} + \frac{1}{2\Delta \abs{\calH}}(\psi^{\odot \ell})^\dagger K_{\bH_\calI} \psi^{\odot \ell}$.
\end{proof}

Given this relationship, it suffices to provide a spectral norm bound for the Kikuchi matrix when the underlying Hamiltonian instance is semirandom, which we do using matrix concentration bounds. In order to do so, we need to establish a few more basic facts about the Kikuchi matrix. For starters, it turns out to be much easier to get a tight bound on the Hamiltonian's energy after applying a particular regularization process to the Kikuchi graph, which we define here.

\begin{definition}[Degree-regularized Kikuchi matrix]
    \label{def:regularkikuchimatrix}
    Let $K = A^* \otimes \Id_{2^n}$ be a level-$\ell$ Kikuchi matrix. Let $\Gamma \in \R^{\mathcal{P}_\ell(n) \times \mathcal{P}_\ell(n)}$ be defined as $\Gamma = D + d\Id$ where $D$ is the diagonal degree matrix of the unsigned $A$ and $d = \E_{\bP \sim \mathcal{P}_\ell(n)}\sbra{\deg(\bP)}$ is the average degree. The degree-regularized Kikuchi matrix is then
    \begin{equation*}
        \wt{K} = \Gamma^{-1/2}A^*\Gamma^{-1/2} \otimes \Id_{2^n} \mper
    \end{equation*}
\end{definition}

To use this regularization effectively, we want the following bound on the average degree in a row/column in $A$.

\begin{observation}
    \label{fact:avgdegree}
    Let $\calI = (\calH, \{(P_C, b_C)\}_{C \in \calH})$ be a Hamiltonian $k$-$\XOR$ instance on $n$ qubits describing $\bH_\calI$ and let $K_{\bH_\calI}$ be the Kikuchi matrix. For $P \in \mathcal{P}_\ell(n)$ we define the graph degree according to the Kikuchi graph $A_{\bH_\calI}$. Then $\E_{\bP \sim \mathcal{P}_\ell(n)}[\deg(\bP)] \geq \left({\frac{\ell}{3n}}\right)^{k/2} \cdot \abs{\calH}$.
\end{observation}

\begin{proof}
    Each $C \in \calH$ contributes $\Delta$ to the total degree, so the average degree is $\E_{\bP \sim \mathcal{P}_\ell(n)}[\deg(\bP)] = \frac{\abs{\calH}\Delta}{\abs{\mathcal{P}_\ell(n)}}$. We then have
    \begin{equation*}
    \E_{\bP \sim \mathcal{P}_\ell(n)}[\deg(\bP)] = \frac{\Delta}{\abs{\mathcal{P}_\ell(n)}} \cdot \abs{\calH} = \frac{3^{\ell-k/2} \binom{k}{k/2}\binom{n-k}{\ell-k/2}}{3^\ell \binom{n}{\ell}} \cdot \abs{\calH} \geq  \left({\frac{\ell}{3n}}\right)^{k/2} \cdot \abs{\calH} \mcom
    \end{equation*}
    where the last inequality follows from \Cref{fact:binomest}.
\end{proof}

\subsection{\texorpdfstring{Step 2: Bounding the spectral norm of $\wt{K}$}{Step 2}}

We are now ready to state our main technical component, which is a spectral bound on the underlying Kikuchi adjacency matrix using \Cref{lem:matrixchernoff}, and use it to prove the even case of \Cref{thm:refutation}.

\begin{lemma}
    \label{lem:spectralnorm}
    Let $\mathbfcal{I} = (\calH, \{(P_C, \bb_C)\}_{C \in \calH})$ be a Hamiltonian $k$-$\XOR$ instance on $n$ qubits and $\bsH_{\mathbfcal{I}} \in \End((\C^2)^{\otimes n})$ the associated Hamiltonian. Let $\wt{K}_{\bsH_{\mathbfcal{I}}} = \Gamma^{-1/2}A^*_{\bsH_{\mathbfcal{I}}}\Gamma^{-1/2} \otimes \Id_{2^n}$ be the level-$\ell$ degree-regularized Kikuchi matrix, as defined in \Cref{def:regularkikuchimatrix}, for $\bsH_{\mathbfcal{I}}$. Suppose additionally that the interaction coefficients $\{\bb_C\}_{C \in \calH}$ are drawn independently and uniformly from $\{\pm 1\}$ or they are drawn as independent standard Gaussians. Then, with probability $\geq 1 - \frac{1}{\poly(n)}$, it holds that
    \begin{equation*}
        \norm{\Gamma^{-1/2} A^*_{\bsH_{\mathbfcal{I}}} \Gamma^{-1/2}}{2} \leq O\left(\sqrt{\frac{\ell\log n}{d}}\right)\mper
    \end{equation*}
\end{lemma}

\begin{proof}[Proof of \Cref{lem:spectralnorm}]
Note for any fixed $C \in \calH$, $A_{P_C}$ is a partial matching since once $Q$ is fixed at most one neighbor $R$ satisfies $Q^\dagger R = P_C$. Hence $A_{P_C}^2$ is diagonal. We then have for any $Q \in \calP_\ell(n)$
\begin{align*}
    \pbra{\sum_{C \in \calH} (\Gamma^{-1/2} A_{P_C} \Gamma^{-1/2})^2}_{Q,Q} = \sum_{R \in N(Q)} \frac{1}{(\deg(Q)+d)(\deg(R)+d)} \leq \frac{\deg(Q)}{(\deg(Q)+d)d} \leq \frac{1}{d}\mper
\end{align*}
As a result, $\norm{\sum_{C \in \calH} (\Gamma^{-1/2} A_{P_C} \Gamma^{-1/2})^2}{2} \leq \frac{1}{d}$. Letting $N \coloneqq \abs{\calP_\ell(n)}$, we have immediately by \Cref{lem:matrixchernoff} and $\Gamma^{-1/2} A^*_{\bsH_{\mathbfcal{I}}} \Gamma^{-1/2} = \sum_{C \in \calH} \bb_C (\Gamma^{-1/2} A_{P_C} \Gamma^{-1/2})$ that
\begin{align*}
    \Pr\sbra{\norm{\Gamma^{-1/2} A^*_{\bsH_{\mathbfcal{I}}} \Gamma^{-1/2}}{2} \geq t} \leq4N\exp(-dt^2/2)\mper
\end{align*}
Since $\log N=O(\ell\log n)$, this gives the desired $O(\sqrt{\ell\log n/d})$ bound for both the Rademacher and Gaussian cases.
\end{proof}

\Cref{thm:refutation} now follows.
\begin{proof}[Proof of \Cref{thm:refutation} for even $k$]
    Let $\wt{\boldsymbol{A}} = \Gamma^{-1/2} \boldsymbol{A}^* \Gamma^{-1/2}$. For any unit norm state $\ket{\psi} \in (\C^2)^{\otimes n}$ we have by \Cref{obs:kikuchirelaxation}
    \begin{align*}
        \bra{\psi}\bsH_{\mathbfcal{I}}\ket{\psi} 
        &= \frac{1}{2} + \frac{1}{2\abs{\calH}\Delta} \cdot (\psi^{\odot \ell})^\dagger K_{\bsH_{\mathbfcal{I}}} \psi^{\odot \ell} \\
        &= \frac{1}{2} +\frac{1}{2\abs{\calH}\Delta} \cdot (\psi^{\odot \ell})^\dagger (\Gamma^{1/2} \otimes \Id_{2^n})(\wt{\boldsymbol{A}} \otimes \Id_{2^n}) (\Gamma^{1/2} \otimes \Id_{2^n})\psi^{\odot \ell}\\
        &\leq \frac{1}{2} +\frac{1}{2\abs{\calH}\Delta} \cdot \lVert \wt{\boldsymbol{A}} \otimes \Id_{2^n}\rVert_{2} \cdot \norm{(\Gamma^{1/2} \otimes \Id_{2^n}) \psi^{\odot \ell}}{2}^2\\
        &= \frac{1}{2} +\frac{1}{2\abs{\calH}\Delta} \cdot\lVert\wt{\boldsymbol{A}}\rVert_{2} \cdot \tr(\Gamma)\mper
    \end{align*}
    The last step follows from the multiplicativity of spectral norm across tensor product and the fact
    \begin{equation*}
        \norm{(\Gamma^{1/2} \otimes \Id_{2^n}) \psi^{\odot \ell}}{2}^2 = (\psi^{\odot \ell})^\dagger(\Gamma \otimes \Id_{2^n})\psi^{\odot \ell} = \sum_{P \in \mathcal{P}_\ell(n)} (\psi^{\odot \ell}_P)^\dagger (\Gamma_{P, P} \Id_{2^n})\psi_P^{\odot \ell} = \tr(\Gamma)\mper
    \end{equation*}
    Here we use the fact that $P^2 = \Id$ for any $P \in \mathcal{P}_\ell(n)$. Now we note $\tr(\Gamma) = \sum_{P \in \mathcal{P}_\ell(n)} (\deg(P) + d) = 2\abs{\calH}\Delta$ which is just twice the total degree. Putting it all together and invoking \Cref{lem:spectralnorm} we have
    \begin{equation*}
        \lambda_{\mathrm{max}}(\bsH_{\mathbfcal{I}}) = \max_{\ket{\psi} \in (\C^2)^{\otimes n}}\bra{\psi}\bsH_{\mathbfcal{I}}\ket{\psi} \leq \frac{1}{2} + \lVert \wt{\boldsymbol{A}} \rVert_2 \leq \frac{1}{2} + O\left(\sqrt\frac{\ell \log n}{d}\right)\mper
    \end{equation*}
    Finally, we recall \Cref{fact:avgdegree} that $d \geq \left(\frac{\ell}{3n}\right)^{k/2} \cdot \abs{\calH}$ and observe that if \[\abs{\calH} \geq O(1) \cdot n \log n \left(\frac{3n}{\ell}\right)^{k/2-1} \varepsilon^{-2}\] for a sufficiently large universal constant, the bound becomes $\frac{1}{2} + \varepsilon$ as desired.
\end{proof}

\section{\texorpdfstring{Non-commutative Sum-of-Squares Lower Bounds for $k$-$\XOR$ Hamiltonians}{Non-commutative Sum-of-Squares Lower Bounds for XOR Hamiltonians}}
\label{sec:ncsos}

In this section, we prove a collection of results characterizing non-commutative Sum-of-Squares bounds for certain classes of worst-case and average-case Hamiltonian $k$-$\XOR$.

\subsection{Warmup: Lifting Sum-of-Squares lower bounds}

We begin with a simple warm-up showing that any integrality gap for standard Sum-of-Squares can be lifted to a corresponding integrality gap for non-commutative Sum-of-Squares by treating the Hamiltonian as a CSP.

\begin{observation}[SoS-hardness $\implies$ (one-basis) ncSoS-hardness]
    \label{thm:lifting}
    Fix $k \geq 2$. Given a classical $k$-CSP instance $\calI = (\calH, \{\varphi_C\}_{C \in \calH})$, the standard basis Hamiltonian embedding $\bH_{\calI} \in \End((\C^2)^{\otimes n})$ satisfies:
    \begin{enumerate}
        \item $\lambda_{\mathrm{max}}(\bH_\calI) = \mathrm{val}(\calI)$;
        \item For any $d \geq k$, $\ncSoS_d(\bH_\calI) = \SoS_d(\calI)$.
    \end{enumerate}
\end{observation}

Note that \Cref{thm:lifting} immediately lifts the refutation threshold to the semirandom setting by instantiating it with the classical $k$-$\XOR$ integrality gap of \cite{Grigoriev01, Schoenebeck08, KothariMOW17}.

\begin{proof}[Proof of \Cref{thm:lifting}]
    The mapping we will take here is the $Z$-basis simulation $S \mapsto Z_S$, but we just as easily could have used $X_S$ or $Y_S$. As the first step in the proof, we show that this does not change the value of the instance, which just follows the canonical proof that such a Hamiltonian embedding is faithful.
    
    Let $\ket{\psi}$ be a top eigenvector for $\bH_\calI$. Suppose $\ket{\psi}$ decomposes in the standard basis as $\ket \psi = \sum_{z \in \{0,1\}^n} \alpha_z \ket z$. Identify $\{0,1\}^n \simeq \{\pm 1\}^n$ via the standard isomorphism. We calculate directly
    \begin{align*}
        \bra{\psi} \bH_\calI \ket{\psi} &= \sum_{z,z' \in \{\pm 1\}^n} \overline{\alpha_{z}}\alpha_{z'} \cdot \bra{z} \bH_\calI \ket{z'}\\
        &= \sum_{z, z' \in \{\pm1\}^n} \overline{\alpha_{z}}\alpha_{z'} \cdot \sum_{S \subseteq [n]} \widehat{\val}_\calI(S) \cdot \bra{z}Z_S\ket{z'}\\
        &= \sum_{z \in \{\pm 1\}^n} \abs{\alpha_z}^2 \cdot \sum_{S \subseteq [n]} \widehat{\val}_\calI(S) \cdot \chi_S(z)\\
        &= \sum_{z \in \{\pm1\}^n} \abs{\alpha_z}^2 \cdot \val(\calI, z)\mper
    \end{align*}
    Since $\sum_{z \in \{\pm1\}^n} \abs{\alpha_z}^2 = 1$, this forms a probability distribution over $\val(\calI, z)$. This is achieved and maximized by $\ket{x^*}$ for any maximizer $x^*$ for $\calI$.

    For the second half of the proof, we show equivalence of the non-commutative and standard Sum-of-Squares values by deriving matching pseudo-expectations in both directions.
    While this equivalence may a priori seem trivial, proving the equivalence requires some care; this is because the notions of Boolean pseudo-expectation and density matrix pseudo-expectations are different. The former fools degree-$d$ operators in the span of $\{I, Z\}^{\otimes n}$ while the latter fools degree-$d$ operators in the span of general Paulis. We show that, nevertheless, for Hamiltonians corresponding to $k$-$\XOR$ instances, the two optimizations yield the same value.
    
    First, assume there is some degree-$d$ pseudo-expectation operator $\pE_\rho$ with value $\pE_\rho\sbra{\bH_\calI}$ for $\bH_\calI$. We show the natural Boolean pseudo-expectation $\pE_\mu$ built by dropping non-$Z$-basis values from $\pE_\rho$ also has value $\pE_\rho\sbra{\bH_\calI}$ for $\calI$. More explicitly, we define for any parity function $\chi_S : \{\pm1\}^n \to \{\pm 1\}$ for $S \subseteq [n]$, $\pE_\mu\sbra{\chi_S} = \pE_\rho\sbra{Z_S}$. This extends by linearity to define $\pE_\mu$ on all degree-at-most-$d$ polynomials.

As an immediate consequence of this construction, we get that the value of $\bH_\calI$ under $\pE_\rho$ is exactly that of $\calI$ under $\pE_\mu$:
\begin{align*}
        \pE_\mu\sbra{\val(\calI)} &= \sum_{S \subseteq [n]} \hat{\val_\calI}(S) \cdot \pE_\mu\sbra{x_S}\\
        &= \sum_{S \subseteq [n]} \hat{\val_\calI}(S) \cdot \pE_\rho\sbra{Z_S}\\
        &= \pE_\rho\sbra{\sum_{S \subseteq [n]} \hat{\val_\calI}(S) \cdot Z_S}\\
        &= \pE_\rho\sbra{\bH_\calI}\mper
    \end{align*}

    Normalization, positivity, and Booleanity are also immediate in this direction; given any polynomial $p : \{\pm 1\}^n \to \R$ with $\deg(p) \leq \frac{d}{2}$ translating to its Hamiltonian simulation inherits these properties from the guarantees of $\pE_\rho$ and the involutory property of Pauli operators.

We now show the other direction, how to build a pseudo-expectation $\pE_\rho$ that is at least as good for $\calH_\calI$ as any pseudo-expectation $\pE_\mu$ for $\calI$. Let $\pE_\mu$ be a degree-$d$ Boolean pseudo-expectation. We construct the natural pseudo-expectation $\pE_\rho$ from $\pE_\mu$ as follows. For $Z$-basis $P \in \calP_{\leq d}(n)$, so $P = Z_S$ for some $S \subseteq [n]$, $\abs{S} \leq d$, define $x_P = \prod_{i \in \text{supp}(P)} x_i$ and let $\pE_\rho\sbra{P} = \pE_\mu\sbra{x_P}$. For any other operator in $P \in \calP_{\leq d}(n)$, let $\pE_\rho\sbra{P} = 0$. Any degree-$d$ operator can be derived via linearity from this set. 

    As before, the value is immediate,
    \begin{align*}
        \pE_\rho\sbra{\bH_\calI} &= \sum_{S \subseteq [n]} \hat{\val_\calI}(S) \cdot \pE_\rho\sbra{Z_S}\\
        &= \sum_{S \subseteq [n]} \hat{\val_\calI}(S) \cdot \pE_\mu\sbra{x_S}\\
        &= \pE_\mu\sbra{\sum_{S \subseteq [n]} \hat{\val_\calI}(S) \cdot x_S}\\
        &= \pE_\mu\sbra{\val(\calI)}\mper
    \end{align*}
    Here normalization is inherited directly from $\pE_\mu$, but positivity requires more work and is necessary to show $\pE_\rho$ is a valid pseudo-expectation. Let $\bH$ be an operator on $(\C^2)^{\otimes n}$ of degree-at-most-$d/2$, and we show $\pE_\rho\sbra{\bH^\dagger \bH} \geq 0$. We construct the following equivalence relation on $\calP_{\leq d/2}(n)$.

    \begin{definition}[$Z$-basis equivalence on $\calP_{\leq d/2}(n)$]
        Let $Q, R \in \calP_{\leq d/2}(n)$ be related $Q \sim R$ if $Q^\dagger R$ is $Z$-basis, i.e. is a complex scalar multiple of $\{\Id, Z\}^{\otimes n}$. Moreover, we can identify the equivalence classes of $\sim$ with the elements of $\{\Id, X\}^{\otimes n}$ with degree-at-most-$d/2$.
    \end{definition}

    We prove that $\sim$ is indeed an equivalence relation on $\calP_{\leq d/2}(n)$.

    \begin{proof}
        $Q^2 = \Id$ so $Q \sim Q$. If $Q \sim R$, $Q^\dagger R = \alpha T$ for $\alpha \in \C$ and $T \in \{\Id, Z\}^{\otimes n}$ and $R^\dagger Q = \overline{\alpha} T$, so $R \sim Q$. Suppose then $Q \sim R$ and $R \sim T$, then $Q^\dagger R$ is $Z$-basis and $R^\dagger T$ is $Z$-basis, and multiplying $Z$-basis operators $Q^\dagger RR^\dagger T = Q^\dagger T$ remains $Z$-basis so $Q \sim T$.

        To see the bijection with $\{\Id, X\}^{\otimes n}$, take any $Q \in \{\Id, X, Y, Z\}^{\otimes n}$ and take $Z \mapsto \Id$ and $Y \mapsto X$ in $Q$ and normalize the scalar coefficient to be $1$. Call this $\wt{Q}$ and observe $\wt{Q} \in \{\Id, X\}^{\otimes n}$. By construction, $Q^\dagger  \wt{Q}$ is $Z$-basis, so $\wt{Q}$ is in the equivalence class of $Q$. To see that $\wt{Q}$ is unique, suppose that two operators $R, T$ from $\{\Id, X\}^{\otimes n}$ fall in the same equivalence class. Since $R \neq T$, they differ at some index $i \in [n]$, and at this index $R_i^\dagger T_i = X$, meaning that $R^\dagger T$ is not $Z$-basis.
    \end{proof}

    Let $\{\calF_i\}_{i \in [r]}$ be the set of equivalence classes induced by $\sim$ and decompose $\bH$ in the Pauli basis partitioned by $\{\calF_i\}_{i \in [r]}$.
    \begin{align*}
        \pE_\rho\sbra{\bH^\dagger \bH} &= \pE_\rho\sbra{\left(\sum_{i =1}^r \sum_{P \in \calF_i} \langle P, \bH \rangle \cdot P\right)^\dagger\left(\sum_{i =1}^r \sum_{P \in \calF_i} \langle P, \bH \rangle \cdot P\right)}\\
        &= \sum_{i =1}^r \sum_{j = 1}^r \sum_{\substack{P \in \calF_i \\ P' \in \calF_j}} \overline{\langle P, \bH \rangle} \cdot \langle P', \bH \rangle \cdot \pE_\rho\sbra{P^\dagger P'}\\
        &= \sum_{i =1}^r \sum_{P, P' \in \calF_i} \overline{\langle P, \bH \rangle} \cdot \langle P', \bH \rangle \cdot \pE_\rho\sbra{P^\dagger P'}\mper
    \end{align*}

    The last line follows from the definition of $\sim$, that non-$Z$-basis operators have value $0$ under $\pE_\rho$. Now let's zoom in on the equivalence class $\calF_1$, which we define to be the one including $\Id^{\otimes n}$. Observe that all elements of $\calF_1$ must be from $\{\Id, Z\}^{\otimes n}$, so $PP' = P''$ and $x_Px_{P'} = x_{P''}$ coincide. As a result
    \begin{align*}
        \sum_{P, P' \in \calF_1} \overline{\langle P, \bH \rangle} \cdot \langle P', \bH \rangle \cdot \pE_\rho\sbra{P^\dagger P'} &= \sum_{P, P' \in \calF_1} \overline{\langle P, \bH \rangle} \cdot \langle P', \bH \rangle \cdot \pE_\mu\sbra{x_P x_{P'}}\\
        &= \pE_\mu\sbra{\sum_{P, P' \in \calF_1} \overline{\langle P, \bH \rangle} \cdot \langle P', \bH \rangle \cdot x_P x_{P'}}\\
        &= \pE_\mu\sbra{\abs{\sum_{P \in \calF_1} \langle P, \bH \rangle \cdot x_P}^2}\\
        &\geq 0\mper
    \end{align*}
    In the last line we appeal to the positivity of $\pE_\mu$. We would like to repeat the same argument for the remaining $\calF_i$, but we do not have that $P \in \calF_i$ is in $\{\Id, Z\}^{\otimes n}$ necessarily, so $x_P$ is undefined. What we do know is that we can fix the equivalence class representative $Q_i \in \{\Id, X\}^{\otimes n}$ and see $Q_i^\dagger P = \alpha_P T_P$ where $\alpha_P \in \{\pm 1, \pm i\}$ and $T_P \in \{\Id, Z\}^{\otimes n}$. This allows us to write:
    \begin{align*}
        \sum_{P, P' \in \calF_i} \overline{\langle P, \bH \rangle} \cdot \langle P', \bH \rangle \cdot \pE_\rho\sbra{P^\dagger P'} &= \sum_{P, P' \in \calF_i} \overline{\langle P, \bH \rangle} \cdot \langle P', \bH \rangle \cdot \pE_\rho\sbra{(Q_i^\dagger P)^\dagger Q_i^\dagger P'}\\
        &= \sum_{P, P' \in \calF_i} \overline{\alpha_P\langle P, \bH \rangle} \cdot \alpha_{P'}\langle P', \bH \rangle \cdot \pE_\rho\sbra{T_P^\dagger T_{P'}}\\
        &= \pE_\mu\sbra{\sum_{P, P' \in \calF_i} \overline{\alpha_P\langle P, \bH \rangle} \cdot \alpha_{P'}\langle P', \bH \rangle \cdot x_{T_P}x_{T_{P'}}}\\
        &= \pE_\mu\sbra{\abs{\sum_{P \in \calF_i} \alpha_P\langle P, \bH \rangle \cdot x_{T_P}}^2}\\
        &\geq 0\mper
    \end{align*}
    The last line appeals to the positivity of $\pE_\mu$.
\end{proof}

\subsection{\texorpdfstring{$\NP$-hard NLTS Hamiltonians with non-commutative Sum-of-Squares lower bounds}{NP-hard NLTS Hamiltonians with non-commutative Sum-of-Squares lower bounds}}
\label{sec:nlts}

In this section we prove \Cref{thm:ncsosforr}, building a family of $XZ$-$k$-$\XOR$ instances that simultaneously have (1) an NLTS property, (2) $\NP$-hardness-of-approximation for their maximum energy, and (3) $\Omega(n)$-level non-commutative Sum-of-Squares lower bounds. To motivate our construction, we start by observing the following simple fact about the energy of an $XZ$-$k$-$\XOR$ Hamiltonian.

\begin{fact}[Energy of $XZ$-Hamiltonians]
    \label{cor:forrelatedvalue}
    Let $\calI = (\calI_X, \calI_Z)$ be a pair of classical $k$-CSPs, and let $\bH_\calI$ be associated $XZ$-Hamiltonian. Then
    \begin{align*}
        \lambdamax(\bH_\calI) \leq \frac{\abs{\calH_X}}{\abs{\calH_X} + \abs{\calH_Z}} \val(\calI_X) + \frac{\abs{\calH_Z}}{\abs{\calH_X}+\abs{\calH_Z}} \val(\calI_Z)\mper
    \end{align*}
    Moreover, commutation of all local terms in $\bH_\calI$ is sufficient for equality.
\end{fact}

\begin{proof}
    We prove the inequality part of the claim. For a proof that commutation implies equality, see the appendix of \cite{MaN25}, which itself adapts the trace technique of \cite{Schuch11}. The former follows from the following observation.
    
    \begin{observation}
        Let $\calI = (\calI_X, \calI_Z)$ be an instance of $XZ$-$k$-$\XOR$ describing Hamiltonian $\bH_\calI$. Let $\ket{\psi} \in (\C^2)^{\otimes n}$ with standard and Hadamard measurement distributions $D_Z$ and $D_X$ respectively. Then the energy of $\bH_\calI$ under $\ket{\psi}$ is given by
        \begin{align*}
            \bra{\psi} \bH_\calI \ket{\psi} = \frac{\abs{\calH_{\calI_X}}}{\abs{\calH_{\calI_X}}+\abs{\calH_{\calI_Z}}} \E_{\boldsymbol{x} \sim D_X}\sbra{\val(\calI_X, \boldsymbol{x})} + \frac{\abs{\calH_{\calI_Z}}}{\abs{\calH_{\calI_X}}+\abs{\calH_{\calI_Z}}} \E_{\boldsymbol{z} \sim D_Z}\sbra{\val(\calI_Z, \boldsymbol{z})}\mper
        \end{align*}
    \end{observation}

    To see this, just compute
    \begin{align*}
        \bra{\psi} \bH_\calI \ket{\psi} &= \frac{\abs{\calH_{\calI_X}}}{\abs{\calH_{\calI_X}}+\abs{\calH_{\calI_Z}}}\bra{\psi} \bH_{\calI_X} \ket{\psi} + \frac{\abs{\calH_{\calI_Z}}}{\abs{\calH_{\calI_X}}+\abs{\calH_{\calI_Z}}}\bra{\psi} \bH_{\calI_Z} \ket{\psi}\mper
    \end{align*}
    From the proof of \Cref{thm:lifting} we see $\bra{\psi} \bH_{\calI_Z} \ket{\psi}$ is exactly the value $\val(\calI_Z, z)$ under the measurement distribution $D_Z$ of $\ket{\psi}$. The proof is analogous for $\bH_{\calI_X}$.
\end{proof}

Intuitively, what these observations tell us is that the maximal energy of $\bH_\calI$ in $XZ$-Hamiltonians is governed by the two classical values $\val(\calI_X)$ and $\val(\calI_Z)$ but can in general be lower if the two Hamiltonians anticommute in some way.

Our goal in this section is to understand how the non-commutative Sum-of-Squares value of an $XZ$-Hamiltonian decomposes in terms of the standard Sum-of-Squares values. In particular, we show the Sum-of-Squares analogue of \Cref{cor:forrelatedvalue}.

\begin{theorem}[Non-commutative Sum-of-Squares value of $XZ$-Hamiltonians]
    \label{thm:ncsosvalue}
    Let $\calI = (\calI_X, \calI_Z)$ be a pair of classical CSPs, and let $\bH_\calI$ be associated $XZ$-Hamiltonian. Then for any degree $d \geq k$
    \begin{align*}
        \ncSoS_d(\bH_\calI) \leq \frac{\abs{\calH_X}}{\abs{\calH_X} + \abs{\calH_Z}} \SoS_d(\calI_X) + \frac{\abs{\calH_Z}}{\abs{\calH_X}+\abs{\calH_Z}} \SoS_d(\calI_Z)\mper
    \end{align*}
    Moreover, this inequality is tight when all local terms in $\bH_\calI$ commute.
\end{theorem}

\begin{proof}[Proof of \Cref{thm:ncsosvalue}]
    We begin with a simple proof of the inequality. Suppose $\pE_\rho$ is a maximizing degree-$d$ pseudo-expectation for $\bH_\calI$ achieving value $\ncSoS_d(\bH_\calI)$. Let $\pE_{\mu_X}$ be the Boolean degree-$d$ pseudo-expectation built by restricting $\pE_\rho$ to the $X$-basis as in \Cref{thm:lifting}, and define $\pE_{\mu_Z}$ analogously. Then we have
    \begin{align*}
        \ncSoS_d(\bH_\calI) &= \frac{\abs{\calH_X}}{\abs{\calH_X} + \abs{\calH_Z}}\pE_\rho\sbra{\bH_{\calI_X}} + \frac{\abs{\calH_Z}}{\abs{\calH_X} + \abs{\calH_Z}}\pE_\rho\sbra{\bH_{\calI_Z}}\\
        &= \frac{\abs{\calH_X}}{\abs{\calH_X} + \abs{\calH_Z}}\pE_{\mu_X}\sbra{\val(\calI_X)} + \frac{\abs{\calH_Z}}{\abs{\calH_X} + \abs{\calH_Z}}\pE_{\mu_Z}\sbra{\val(\calI_Z)}\\
        &\leq \frac{\abs{\calH_X}}{\abs{\calH_X} + \abs{\calH_Z}}\SoS_d(\calI_X) + \frac{\abs{\calH_Z}}{\abs{\calH_X} + \abs{\calH_Z}}\SoS_d(\calI_Z)\mper
    \end{align*}
    We move to proving equality under the commutation assumption $[\bH_{\calI_X}, \bH_{\calI_Z}] = 0$. Our goal will simply be to fix degree-$d$ Boolean pseudo-expectations maximizing $\calI_X$ and $\calI_Z$ and use them to build a matching non-commutative pseudo-expectation $\pE_\rho$.
    
    Towards this, fix $\pE_{\mu_X'}$ and $\pE_{\mu_Z'}$ achieving values $\SoS_d(\calI_X)$ and $\SoS_d(\calI_Z)$ for $\calI_X$ and $\calI_Z$. We define them as prime versions as we would really like to define the following modifications. For any CSP $\calI$ let
    \begin{align*}
        V(\calI) = \spn_{\F_2^n}\{S \mid \wh{\val}_\calI(S) \neq 0\}\mcom
    \end{align*}
    where we define $\spn_{\F_2^n}$ as the span of subsets $S \subseteq [n]$ when treated as vectors over $\F_2^n$. For any degree-$d$ polynomial $p : \{\pm1\}^n \to \R$ we define
    \begin{align*}
        \pE_{\mu_X}\sbra{p} = \E_{\boldsymbol{S} \sim V(\calI_X)^\perp}\sbra{\pE_{\mu_X'}\sbra{p((-1)^{\boldsymbol{S}} \odot x)}}\mcom
    \end{align*}
    and $\pE_{\mu_Z}$ analogously. The motivation for this is two-fold: (1) we introduce the proper normalization factor and (2) we neutralize $\pE_{\mu_X}$ outside of $V(\calI_X)$. First and foremost, we argue this modified pseudo-expectation preserves pseudo-expectation properties, which inherits from $\pE_{\mu_X'}$ since pseudo-expectation properties are closed under signing.

    Now observe for any $S \subseteq [n]$,
    \begin{equation*}
        \pE_{\mu_X}\sbra{x_S} = \pbra{\frac{1}{\abs{V(\calI_X)^\perp}} \sum_{T \in V(\calI_X)^\perp} (-1)^{\abs{S \cap T}} } \pE_{\mu_X'}\sbra{x_S} = \begin{cases}
            \pE_{\mu_X'}\sbra{x_S} & S \in V(\calI_X)\\
            0 & \text{otherwise}
        \end{cases}
    \end{equation*}

    So this normalization has the effect of zeroing the pseudo-expectation outside of the values relevant for $\val(\calI_X)$.
    From this point forward, our construction is now the following
    \begin{align*}
        \pE_\rho\sbra{X_S Z_T} = \pE_{\mu_X}\sbra{x_S} \pE_{\mu_Z}\sbra{z_T}\mper
    \end{align*}
    This is a fully specified degree-$d$ pseudo-expectation since $X_SZ_T$ orthogonally span $\calP_{\leq d}(n)$ for $\{(S, T) \in [n] \mid \abs{S \cup T} \leq d\}$.
    We can now check this construction achieves the equality value. Observe that
    \begin{align*}
        \ncSoS_d(\bH_\calI) \geq \frac{\abs{\calH_X}}{\abs{\calH_X} + \abs{\calH_Z}}\pE_\rho\sbra{\bH_{\calI_X}} + \frac{\abs{\calH_Z}}{\abs{\calH_X} + \abs{\calH_Z}}\pE_\rho\sbra{\bH_{\calI_Z}}\mcom
    \end{align*}
    and
    \begin{align*}
        \pE_{\mu_X}\sbra{\val(\calI_X)} = \SoS_d(\calI_X)\mcom \qquad \pE_{\mu_Z}\sbra{\val(\calI_Z)} = \SoS_d(\calI_Z)\mcom
    \end{align*}
    since we did not modify the pseudo-expectations on their corresponding value functions. It suffices to show, without loss of generality, $\pE_\rho\sbra{\bH_{\calI_X}} = \pE_{\mu_X}\sbra{\val(\calI_X)}$. This is immediate by construction as we have
    \begin{align*}
        \pE_\rho\sbra{\bH_{\calI_X}} &= \sum_{S \subseteq [n]} \wh{\val}_{\calI_X}(S) \cdot \pE_\rho\sbra{X_S}\\
        &= \sum_{S \subseteq [n]} \wh{\val}_{\calI_X}(S) \cdot \pE_{\mu_X}\sbra{x_S}\\
        &= \pE_{\mu_X}\sbra{\sum_{S \subseteq [n]} \wh{\val}_{\calI_X}(S) \cdot x_S}\\
        &= \pE_{\mu_X}\sbra{\val(\calI_X)}\mper
    \end{align*}
    All that remains is to prove positivity of the pseudo-expectation $\pE_\rho$. The key structural observation will be that $V(\calI_X) \perp V(\calI_Z)$ by the commutativity assumption. To see this, we write the commutator out as
    \begin{align*}
        [\bH_{\calI_X}, \bH_{\calI_Z}] = \sum_{S, T \subseteq [n]} \wh{\val}_{\calI_X}(S) \wh{\val}_{\calI_Z}(T) (1-(-1)^{\abs{S \cap T}}) X_S Z_T\mper
    \end{align*}
    The $(-1)^{\abs{S \cap T}}$ accounts for the identity $X_S Z_T = (-1)^{\abs{S \cap T}} Z_T X_S$. Since this is an orthogonal decomposition, $[\bH_{\calI_X}, \bH_{\calI_Z}] = 0$ implies each term is $0$, so $\wh{\val}_{\calI_X}(S), \wh{\val}_{\calI_Z}(T) \neq 0$ implies $\abs{S \cap T}$ is even, or in other words, for every $S \in V(\calI_X)$ and $T \in V(\calI_Z)$, $S$ and $T$ are orthogonal under $\F_2^n$-inner product.

    The positivity proof will then follow through the moment matrix point of view, which by \Cref{fact:ncmoment} requires us to prove the moment matrix $M_{\pE_\rho}$ is positive semidefinite. It is a quick check to convert $M_{\pE_\rho}$ from being indexed by the Pauli operators $\calP_{\leq d/2}(n)$ to the basis indexed by $X_S Z_T$ by basis rotation. We then have for any $P = X_S Z_T$, $Q = X_{S'} Z_{T'}$,
    \begin{align*}
        M_{\pE_\rho}(P, Q) = \pE_\rho\sbra{P^\dagger Q} = \pE_\rho\sbra{Z_T X_{S \oplus S'} Z_{T'}} = (-1)^{\abs{(S \oplus S') \cap T}} \pE_{\mu_X}\sbra{x_{S \oplus S'}} \pE_{\mu_Z}\sbra{z_{T \oplus T'}}\mper
    \end{align*}
    Define the signless version
    \begin{align*}
        M'_{\pE_\rho}(P, Q) \coloneqq \pE_{\mu_X}\sbra{x_{S \oplus S'}} \pE_{\mu_Z}\sbra{z_{T \oplus T'}}\mper
    \end{align*}
    This matrix $M'_{\pE_\rho}$ is then a principal submatrix of $M_{\pE_{\mu_X}} \otimes M_{\pE_{\mu_Z}} \in \R^{\binom{[n]}{\leq d/2}^{\otimes 2} \times \binom{[n]}{\leq d/2}^{\otimes 2}}$ restricted to entries $(S, T) \in \binom{[n]}{\leq d/2}^{\otimes 2}$ with $\abs{S \cup T} \leq d/2$. By \Cref{fact:moment} and the positivity assumption we have $M_{\pE_{\mu_X}}, M_{\pE_{\mu_Z}} \succeq 0$ and as a result $M'_{\pE_\rho} \succeq 0$. To show the true moment matrix is positive semidefinite, we will conjugate from $M'_{\pE_\rho}$ using a diagonal matrix.

    To design this diagonal matrix, we start by partitioning $2^{[n]}$ by the cosets of $V(\calI_Z)$ and choosing a canonical element $U(T)$ for the coset $T \oplus V(\calI_Z)$. Letting $P = X_S Z_T$ we define the diagonal matrix $\Lambda \in \R^{\calP_{\leq d/2}(n) \times \calP_{\leq d/2}(n)}$ via
    \begin{align*}
        \Lambda(P, P) = (-1)^{\abs{S \cap U(T)}}\mper
    \end{align*}
    The key insight is that we may write for any non-zero entry $M_{\pE_\rho}(P, Q)$ with $P = X_SZ_T$ and $Q = X_{S'} Z_{T'}$,
    \begin{align*}
        (-1)^{\abs{(S \oplus S') \cap T}} = (-1)^{\abs{S \cap U(T)}} (-1)^{\abs{S' \cap U(T')}}\mper
    \end{align*}
    As an immediate result, we have $M_{\pE_\rho} = \Lambda  M'_{\pE_\rho} \Lambda$ and $M_{\pE_\rho} \succeq 0$, proving positivity of $\pE_\rho$ via \Cref{fact:ncmoment}.

    To finish, we argue this insight. Note the support condition immediately requires $\pE_{\mu_X}\sbra{x_{S \oplus S'}} \neq 0$ and therefore $S \oplus S' \in V(\calI_X) \perp V(\calI_Z)$, as well as the same for $T \oplus T'$ in $V(\calI_Z)$. This means $U(T) = U(T')$ and $T \oplus U(T) \in V(\calI_Z)$ by the coset structure. Consequently, the role of $T$ in $\abs{(S \oplus S') \cap T}$ can be traded with $T', U(T), U(T')$ by the fact $V(\calI_X) \perp V(\calI_Z)$.
\end{proof}

In the rest of this subsection we show how to apply \Cref{thm:ncsosvalue} in our construction for \Cref{thm:ncsosforr}.

\begin{proof}[Proof of \Cref{thm:ncsosforr}]
    Our goal is to build an $XZ$-$k$-$\XOR$ instance satisfying simultaneously \Cref{item:nlts}, \Cref{item:nphard}, and \Cref{item:ncsos}. The starting point of our construction is the quantum Tanner code-based construction of NLTS Hamiltonians from \cite{AnshuBN23}.

    \begin{theorem}[NLTS Hamiltonians \cite{AnshuBN23}]
        \label{thm:nlts}
        There exists an infinite, polynomial-time constructible family of $XZ$-$k_0$-$\XOR$ instances $\calI = (\calI_X, \calI_Z)$ defining $n$-qubit Hamiltonians of the form
        \begin{align*}
            \bH_{\calI} &= \frac{1}{\abs{\calH_X}+\abs{\calH_Z}} \sum_{C \in \calH_X} \frac{\Id + X_C}{2} + \frac{1}{\abs{\calH_X}+\abs{\calH_Z}} \sum_{C \in \calH_Z} \frac{\Id + Z_C}{2}\mcom
        \end{align*}
        and satisfying the NLTS property for some constant $\eta > 0$: for any state $\ket{\psi} \in (\C^2)^{\otimes n}$ with $\bra{\psi} \bH_{\calI} \ket{\psi} \geq 1 - \eta$, the circuit depth complexity of $\ket{\psi}$ is $\Omega(\log n)$. Moreover, the $k_0$-$\XOR$ instances $\calI_X$ and $\calI_Z$ are bounded occurrence, meaning each variable appears in at most $d_0$ clauses independent of $n$. The checks are nonempty and their total number $m_0=\abs{\calH_X}+\abs{\calH_Z}$ is $\Theta(n)$, and can moreover be made $s \leq m_0 \leq s^{O(1)}$ for any sufficiently large target size $s$.
    \end{theorem}

    While satisfying the NLTS property, this construction fails \Cref{item:nphard}, since all Hamiltonians in the family are frustration-free and therefore have $\lambdamax(\bH_{\calI}) = 1$ trivially. At a high-level, our construction fixes this by embedding a linear system PCP into signed linear combinations of local Pauli terms of this base Hamiltonian. The main technical piece is then balancing this embedding in the right way as to preserve the NLTS property and the non-commutative Sum-of-Squares perfect completeness.

    To accomplish this modification, we are going to build a new Hamiltonian $\bH$ from the original $X$ and $Z$-checks of $\calH_\calI$. Let us consider a matrix $A \in \F_2^{s \times r}$ and a vector $b \in \{\pm 1\}^s$ given some subset of the Pauli terms $P_1, \dots, P_r$. For each $j = 1, \ldots, s$, we construct a new Hamiltonian term
    \begin{equation*}
        \bH_j \coloneqq \frac{\Id + b_j \displaystyle \prod_{i = 1}^r P_i^{A_{ji}}}{2} \mper
    \end{equation*}
    In English, the new Pauli term corresponding to $\bH_j$ is the product of the Pauli terms of $\bH_\calI$ as identified by the $j$th row of $A$, adjusted by a sign identified by the $j$th index of vector $b$. Observe, some natural properties: since we assume that the terms in $\bH_\calI$ commute, the order in which the product is computed is irrelevant. Moreover, if the original Paulis commuted as in \Cref{thm:nlts}, any products of them also commute. Second, if the matrix $A$ is row-sparse, then the locality of the new Hamiltonian is implied from the locality of the Hamiltonian $\bH_\calI$.

    \begin{theorem}[$\NP$-hardness of $\XOR$ \cite{Hastad01}]
        \label{thm:nphard}
        For any $0 < \eps < 1/4$, $3$-$\XOR[1-\varepsilon, \frac{1}{2}+\varepsilon]$, the promise problem of deciding whether a $3$-$\XOR$ has value at least $1-\varepsilon$ or at most $\frac{1}{2}+\varepsilon$, is $\NP$-hard. Furthermore, we can assume the $3$-$\XOR$ instance is bounded occurrence.
    \end{theorem}

    A first idea is the following: take a bounded occurrence $3$-$\XOR$ instance $\calI_{\NP} = (\calH_{\NP}, \{b_C\}_{C \in \calH_{\NP}})$ with $m_0 = \abs{\calH_X} + \abs{\calH_Z}$ variables and form a parity check matrix $H_{\NP} \in \F_2^{M \times m_0}$ from $\calH_{\NP}$. Numbering the terms $C_1, \dots, C_{m_0}$ and $b_1, \dots, b_{M}$, one could then form the local terms
    \begin{equation*}
        \bH_j \coloneqq \frac{\Id + b_j \displaystyle \prod_{i = 1}^{m_0} P_{C_i}^{(H_{\NP})_{ji}}}{2} \mper
    \end{equation*}
    The maximum eigenvalue of this Hamiltonian is the maximum of the classical objective over the jointly attainable eigenvalue assignments of the original checks. In the case of independent checks, this means $\lambdamax(\bH) = \val(\calI_{\NP})$. More importantly, however, one can check that taking arbitrary combinations may not preserve NLTS or a non-commutative Sum-of-Squares integrality gap in general.

    Our fix for this is conceptually quite simple but requires a bit of care to execute. Rather than embed the $\NP$-hard instance across the signs of the entire instance, we take only a constant fraction of the $Z$-checks to do so. To guarantee compatibility between the rest of the NLTS Hamiltonian and the $\NP$-hard piece, we add slack variables to the latter term in the form of re-encoded stabilizer checks from the original NLTS Hamiltonian and then re-encode those same checks within the original instance to induce a Sum-of-Squares integrality gap.

    \paragraph{Our construction.}
Fix a constant $c_Z\in(0,1]$ such that $\abs{\calH_Z}\geq c_Zm_0$ for all sufficiently large members of \Cref{thm:nlts}. Such a constant follows already from NLTS: the product state $\ket{+}^{\otimes n}$ has energy $1-\abs{\calH_Z}/(2m_0)$, so $\abs{\calH_Z}>2\eta m_0$ for sufficiently large $n$. Put $K=12/c_Z$. Before choosing the input size, fix an integer $L\geq24Kd_0$, a sufficiently large integer $\Delta\geq10(L+1)$ divisible by $4$ for \Cref{lem:expansion} with parameters $(k_0,3d_0,L)$, and constants
    \begin{equation*}
        0<\rho\leq\min\left\{\frac{c_Z}{12k_0d_0^2},\frac{\eta}{4L}\right\},
        \qquad
        0<\varepsilon_\NP\leq\frac{1}{8(K+1)}\mper
    \end{equation*}
    All these choices are independent of the input instance.

    Given a bounded-occurrence $3$-$\XOR$ input from \Cref{thm:nphard} with $M_{\mathrm{in}}$ clauses, choose a sufficiently large base instance $\calI_{\mathrm{NLTS}}=(\calI_X,\calI_Z)$ from \Cref{thm:nlts} with $m_0=\abs{\calH_X}+\abs{\calH_Z}\geq\max\{4,M_{\mathrm{in}}/\rho\}$ and $m_0$ polynomially bounded in the input size. Replace the input by $\lceil\rho m_0/M_{\mathrm{in}}\rceil$ disjoint copies on fresh variables, and denote the resulting instance by $\calI_\NP=(\calH_\NP,\{b_C\}_{C\in\calH_\NP})$. This preserves its value and bounded occurrence, and gives
    \begin{equation*}
        M\coloneqq \abs{\calH_\NP} \in[\rho m_0,2\rho m_0]\mper
    \end{equation*}
    After discarding unused variables, let $p\leq3M$ be its number of variables and let $H_\NP\in\F_2^{M\times p}$ be its parity check matrix.

    We greedily choose $C_1, \dots, C_p \in \calH_Z$ that are pairwise disjoint, $C_i \cap C_j = \varnothing$ when $i \neq j$. For each, choose a representative variable $x_j \in C_j$. We then partition $\calH_Z$ as follows. Define
    \begin{align*}
        W \coloneqq \bigcup_{\substack{C \in \calH_Z \\ x_j \in C, j \in [p]}} C\mcom \qquad \calH_{\mathrm{int}} \coloneqq \{C \in \calH_Z \mid C \subseteq W\}\mper
    \end{align*}
    Let $\calH_Z = \calH_{\mathrm{int}} \sqcup \calH_{\mathrm{slack}}$. Importantly, $C_1, \dots, C_p \in \calH_{\mathrm{int}}$.
    The greedy selection succeeds because each selected check excludes at most $k_0d_0$ intersecting checks, while
    \begin{equation*}
        p\leq6\rho m_0\leq\frac{c_Zm_0}{2k_0d_0^2}
        \leq\frac{\abs{\calH_Z}}{k_0d_0}\mper
    \end{equation*}
    Moreover, bounded occurrence and nonempty checks give
    \begin{equation*}
        |W|\leq k_0d_0p,\qquad
        \abs{\calH_{\mathrm{int}}}\leq d_0|W|
        \leq6k_0d_0^2\rho m_0\leq\frac{c_Zm_0}{2}\mper
    \end{equation*}
    Thus $\abs{\calH_{\mathrm{slack}}}\geq c_Zm_0/2$. After the padding below, we have
    \begin{equation*}
        \frac{c_Zm_0}{2}\leq r\leq m_0+2,
        \qquad 3M\leq6\rho m_0\leq\frac{c_Zm_0}{2}\leq r\mper
    \end{equation*}
    In particular, $r=\Theta(n)$ and $M\leq r/3$, as required to sample $R$ and apply \Cref{lem:expansion}.

    Now number the elements $D_1, \dots, D_r \in \calH_{\mathrm{slack}}$, guaranteeing $r$ is divisible by $3$ by repeating at most two terms if needed. Let $A \in \F_2^{\Delta r/3 \times r}$ be any $3$-row sparse matrix and let $R \in \F_2^{M \times r}$ be the same.

    Now, define the total number of terms $m = m_{\mathrm{NLTS}} + m_{\NP}$ where $m_{\mathrm{NLTS}} = \abs{\calH_X} + \abs{\calH_{\mathrm{int}}} + \Delta r/3$ and $m_{\NP} = L\abs{\calH_{\NP}}$ for some integer $L \geq 1$ and consider the Hamiltonian $\bH$ defined by
    \begin{align*}
        &\bH_{\mathrm{NLTS}} = \frac{1}{m_{\mathrm{NLTS}}}\sum_{C \in \calH_X} \frac{\Id + X_C}{2} + \frac{1}{m_{\mathrm{NLTS}}}\sum_{C \in \calH_{\mathrm{int}}} \frac{\Id + Z_C}{2} + \frac{1}{m_{\mathrm{NLTS}}}\sum_{j = 1}^{\Delta r/3} \frac{\Id + \prod_{i = 1}^r Z_{D_i}^{A_{ji}}}{2}\\
        &\bH_{\NP} = \frac{L}{m_\NP}\sum_{j = 1}^{M} \frac{\Id + b_j \prod_{i = 1}^p Z_{C_i}^{(H_{\NP})_{ji}} \prod_{i=1}^r Z_{D_i}^{R_{ji}} }{2}\\
        &\bH = \frac{m_{\mathrm{NLTS}}}{m}\bH_{\mathrm{NLTS}} + \frac{m_\NP}{m}\bH_{\NP}\mper
    \end{align*}
    
    For fixed integers $L, \Delta \geq 1$, $\bH$ exists on the same $n$ variables as $\calI_{\mathrm{NLTS}}$ and consists of $m$ terms satisfying roughly
    \begin{equation*}
        \frac{\Delta c_Z}{6}m_0\leq m
        \leq m_0+\frac{\Delta(m_0+2)}{3}+2L\rho m_0
        \leq\Delta m_0\mper
    \end{equation*}
    The last inequality here uses $m_0\geq4$, $\rho\leq1$, and $\Delta\geq10(L+1)$. Since each new term contains a product of at most $6$ original Pauli operators, the locality $k \leq 6k_0 = \Theta(1)$. Finally, since we only take products of the original commuting operators within each basis, the new instance is a commuting $XZ$-$k$-$\XOR$ instance. 
    
    The motivation for this construction is the following: as a first move, we essentially just combine the original NLTS Hamiltonian with an $\NP$-hard-to-approximate $3$-$\XOR$ instance in the signs of a fraction of the $Z$-terms, along with a few slack $Z$-checks. From the viewpoint of non-commutative Sum-of-Squares, these disjoint terms add to the $\NP$-hard instance some slack variables. If it were this alone, the slack variables would allow a cheap solution to the system, which would lead the non-commutative Sum-of-Squares value to be $1$. These slack variables, however, are designed so they do not actually allow any real improvement to the value, since the added checks are already used in the original NLTS Hamiltonian and should therefore contribute nothing to the instance. The ``re-encoding'' trick, where we take linear combinations of $D_1, \dots, D_r$ in $\calH_{\mathrm{slack}}$, is then just a way to hide the uselessness of these slack variables from low-degree Sum-of-Squares, by blowing up the variable expansion required to refute them and thus keeping the relaxation value at $1$.
    
    We now specify the two matrices left open in the construction, which will ultimately be random. Sample
    $\Delta+1$ independent uniformly random partitions of $[r]$ into triples.
    Stack the incidence matrices of the first $\Delta$ partitions to obtain
    $\mathbf A$, and take the first $M$ triples of the last partition to obtain
    $\mathbf R$. Every column of $\mathbf A$ has weight $\Delta$, and every
    column of $\mathbf R$ has weight at most one. This random construction achieves a few expansion properties that we invoke as needed and ultimately prove in \Cref{sec:expansion}. We believe a derandomization of this construction is achievable but focus instead on the clarity of exposition allowed by this construction. In the rest of the proof, we work on the simultaneous expansion event of \Cref{lem:expansion} for the matrix $B$ defined in the proof of \Cref{item:ncsos}.

    \paragraph{NLTS property.} We prove the NLTS property in our construction by appealing to the original NLTS property of the code Hamiltonian, and showing this is not greatly affected by the addition of our $\NP$-hard $\XOR$ encoding.

    \begin{proof}[Proof of \Cref{item:nlts}]
        Assume we have some state $\ket{\psi} \in (\C^2)^{\otimes n}$ satisfying $\bra{\psi} \bH \ket{\psi} \geq \lambdamax(\bH) - \varepsilon$. Our goal is to show the circuit depth complexity of $\ket{\psi}$ is $\Omega(\log n)$. We do this by proving $\bra{\psi} \bH_{\calI_{\mathrm{NLTS}}} \ket{\psi} \geq 1-\eta$ and appealing to \Cref{thm:nlts}.

        We start by rewriting $\ket{\psi} = \sum_{i=1}^{2^n} \alpha_i\ket{\psi_i}$ with $\sum_{i=1}^{2^n} \abs{\alpha_i}^2=1$ for some orthonormal simultaneous eigenbasis $\ket{\psi_i}$ for the original Paulis in $\calI_{\mathrm{NLTS}}$. Note this is possible since the code Hamiltonian is commuting. We then write for each term $C \in \calH_X \sqcup \calH_Z$, $P_C \ket{\psi_i} = (-1)^{p_C^{(i)}} \ket{\psi_i}$. Put all these bits in a vector $p^{(i)}\in \F_2^{m_0}$. Under this notation, one has
        \begin{align*}
            \bra{\psi_i} \bH_{\calI_{\mathrm{NLTS}}}\ket{\psi_i} = 1-\frac{|p^{(i)}|}{m_0}\mper
        \end{align*}
        Now since
        \begin{align*}
            \bH_{\mathrm{NLTS}} = \frac{1}{m_{\mathrm{NLTS}}}\sum_{C \in \calH_X} \frac{\Id + X_C}{2} + \frac{1}{m_{\mathrm{NLTS}}}\sum_{C \in \calH_{\mathrm{int}}} \frac{\Id + Z_C}{2} + \frac{1}{m_{\mathrm{NLTS}}}\sum_{j = 1}^{\Delta r/3} \frac{\Id + \prod_{i = 1}^r Z_{D_i}^{A_{ji}}}{2}\mcom
        \end{align*}
        we can explicitly calculate its energy under $\ket{\psi_i}$ as follows. Let $p^{(i)}_1$ be the restriction of the original syndrome $p^{(i)}$ to $\calH_X \sqcup \calH_{\mathrm{int}}$, and define $v^{(i)}\in\F_2^r$ by $v^{(i)}_j=p^{(i)}_{D_j}$ for $j\in[r]$. Here $D_1,\ldots,D_r$ is the padded list of slack checks, so any repeated labels give repeated coordinates of $v^{(i)}$. Since every original slack check occurs at least once in this list,
        \begin{equation*}
            |p^{(i)}|\leq |p^{(i)}_1|+|v^{(i)}|\mper
        \end{equation*}
        The number of violated re-encoded checks is $|Av^{(i)}|$. In total, we then have
        \begin{align*}
            \bra{\psi_i} \bH_{\mathrm{NLTS}}\ket{\psi_i} = 1-\frac{|p^{(i)}_1| + |Av^{(i)}|}{m_{\mathrm{NLTS}}}\mper
        \end{align*}
        The key inequality we desire is then $|Av| \geq \abs{v}$, which gives
        \begin{equation*}
            |p^{(i)}_1|+|Av^{(i)}|\geq |p^{(i)}_1|+|v^{(i)}|\geq |p^{(i)}|\mper
        \end{equation*}
        Both Hamiltonians are diagonal in the chosen eigenbasis, so averaging with weights $|\alpha_i|^2$ yields
        \begin{align*}
            m_{\mathrm{NLTS}}(1-\bra{\psi} \bH_{\mathrm{NLTS}}\ket{\psi}) \geq m_0(1-\bra{\psi} \bH_{\calI_\mathrm{NLTS}}\ket{\psi})\mper
        \end{align*}
        Similarly, we have
        \begin{align*}
            m (1-\bra{\psi} \bH\ket{\psi}) \geq m_{\mathrm{NLTS}}(1-\bra{\psi} \bH_{\mathrm{NLTS}}\ket{\psi})\mcom
        \end{align*}
        and since $\bH_\NP \succeq 0$, $\lambdamax(\bH) \geq 1-\frac{m_\NP}{m}$ we conclude that
        \begin{align*}
            \bra{\psi}\bH_{\calI_\mathrm{NLTS}}\ket{\psi} \geq 1 - \frac{m_\NP}{m_0} - \frac{\varepsilon m}{m_0}\mper
        \end{align*}
        Our choice of $\rho$ gives $m_\NP/m_0=LM/m_0\leq2L\rho\leq\eta/2$. Also $m\leq\Delta m_0$, so taking $\varepsilon\leq\eta/(2\Delta)$ bounds the latter term by $\eta/2$. We may decrease $\varepsilon$ further in the next part to use a common constant for NLTS and the hardness gap. In total, this gives $\bra{\psi}\bH_{\calI_\mathrm{NLTS}}\ket{\psi}\geq1-\eta$, completing the proof.
    \end{proof}

    \paragraph{$\NP$-hardness-of-approximation.} Our overall goal is to show that the energy of the instance $\calI_\NP$ is encoded in $\lambdamax(\bH)$, which leads to $\NP$-hardness via \Cref{thm:nphard}.

    \begin{proof}[Proof of \Cref{item:nphard}]
        Our main goal is then to prove the following bounds for $\lambdamax(\bH)$:
        \begin{align*}
            \frac{m_{\mathrm{NLTS}}}{m} + \frac{m_\NP}{m} \val(\calI_\NP) - \frac{d_0p}{m} \leq \lambdamax(\bH) \leq \frac{m_{\mathrm{NLTS}}}{m} + \frac{m_\NP}{m} \val(\calI_\NP)\mper
        \end{align*}
        Intuitively, this is saying (1) the slack terms do not help since the upper bound is still controlled by $\val(\calI_\NP)$ directly and (2) the interference between $\bH_{\mathrm{NLTS}}$ and $\bH_\NP$ is captured by at most an energy discrepancy of $\frac{d_0p}{m}$.
    
To obtain thresholds independent of the instance, put $w=m_\NP/m=LM/m$. Using $p\leq3M$, the displayed bounds imply
        \begin{align*}
            \val(\calI_\NP)\geq1-\varepsilon_\NP
            &\quad\Longrightarrow\quad
            \lambdamax(\bH)\geq1-w\left(\varepsilon_\NP+\frac{3d_0}{L}\right),\\
            \val(\calI_\NP)\leq\frac12+\varepsilon_\NP
            &\quad\Longrightarrow\quad
            \lambdamax(\bH)\leq1-w\left(\frac12-\varepsilon_\NP\right)\mper
        \end{align*}
        The uniform bounds on $m$ and $M$ give
        \begin{equation*}
            w_-\coloneqq\frac{L\rho}{\Delta}\leq w\leq
            w_+\coloneqq\frac{12L\rho}{\Delta c_Z},
            \qquad w_+=Kw_-\mper
        \end{equation*}
        Define fixed thresholds
        \begin{equation*}
            \alpha=1-w_-\left(\frac12-\varepsilon_\NP\right),
            \qquad
            \beta=1-w_+\left(\varepsilon_\NP+\frac{3d_0}{L}\right)\mper
        \end{equation*}
        Our choices of $L$ and $\varepsilon_\NP$ ensure
        \begin{equation*}
            \beta-\alpha
            =w_-\left(\frac12-(K+1)\varepsilon_\NP-\frac{3Kd_0}{L}\right)
            \geq\frac{w_-}{4}>0\mper
        \end{equation*}
        Thus $1/2\leq\alpha<\beta\leq1$, all YES instances have maximum energy at least $\beta$, and all NO instances have maximum energy at most $\alpha$. Fixing
        $0<\varepsilon\leq\min\{\eta/(2\Delta),w_-/4\}$ supplies the same constant for the hardness gap and the preceding NLTS conclusion.

        For the lower bound, let $x^* \in \F_2^p$ be an optimal assignment for $\calI_\NP$ in $\F_2$ notation and let
        $\ket{\psi_0}$ be any original code state. Put
        \begin{align*}
            \ket{\psi_{x^*}}
               =\prod_{i=1}^p X_{x_i}^{\,x_i^*}\ket{\psi_0},
        \end{align*}
        where the subscripts $x_i\in C_i$ are the representative variables. In English, we are flipping the assignment of any representative variable set to $1$ in $x^*$. Every original $Z$-check containing a representative variable is supported entirely in $W$ by construction, so no slack check is flipped. Thus
        all re-encoded checks remain satisfied and all slack parities are $+1$.
        At most $d_0p$ flipped $Z$-checks are violated, and every $X$-check
        remains satisfied. Consequently,
        \begin{align*}
            \bra{\psi_{x^*}}\bH_{\mathrm{NLTS}}\ket{\psi_{x^*}}
               &\geq1-\frac{d_0p}{m_{\mathrm{NLTS}}}\mcom \qquad
            \bra{\psi_{x^*}}\bH_{\mathrm{NP}}\ket{\psi_{x^*}} = \val(\calI_\NP)\mper
        \end{align*}
        Taking their weighted average proves the lower bound immediately.

        For the upper bound, we once again work in a simultaneous eigenspace of the original checks and in particular fix a top eigenvector $\ket{\psi}$. Let $x \in \F_2^p$ be the vector of $\F_2$-parities for the $Z$-Paulis corresponding to $C_1, \dots, C_p$ in $\calH_{\mathrm{int}}$ under $\ket{\psi}$. Similarly, let $v \in \F_2^r$ be the $\F_2$-parities of $\ket{\psi}$ for $D_1, \dots, D_r \in \calH_{\mathrm{slack}}$. Observe that under these definitions the eigenvalue of $\bH$ for $\ket{\psi}$ is computed as simply
        \begin{align*}
            \lambdamax(\bH) = 1-\frac{u + \abs{Av} + L\abs{H_\NP x + Rv + b^{(\NP)}}}{m}\mcom
        \end{align*}
        if $u$ is the number of terms violated from $\calH_X$ and $\calH_{\mathrm{int}}$ in total. We also reinterpret $b^{(\NP)}$ to live in $\F_2^M$ instead of $\{\pm1\}^M$ for convenience. The term $\abs{Av}$ comes directly from $\calH_{\mathrm{slack}}$, the re-encoded Paulis under $A$. The term $L\abs{H_\NP x + Rv + b^{(\NP)}}$ comes from the $\NP$-hard term with its slack variables and signs.
        
        The key observation is that by triangle inequality we can write:
        \begin{align*}
            \abs{Av} + L\abs{H_\NP x + Rv + b^{(\NP)}} &\geq \abs{Av} - L\abs{Rv} + L\abs{H_\NP x + b^{(\NP)}}\\
            &\geq \abs{Av} - L\abs{v} + L\abs{H_\NP x + b^{(\NP)}}\\
            &\geq m_\NP (1-\val(\calI_\NP))\mper
        \end{align*}
        The second line uses that $R$ is formed by a permutation and the last line uses the expansion guarantee $\abs{Av} \geq L \abs{v}$. Intuitively this just captures the fact that deviating from $v = 0$, where all slack variables are set to $0$ is never useful. Rearranging from here yields the desired bound.

        Note sampling $A,R$ and forming the listed Pauli products take polynomial time and the expansion properties hold with high probability, forming the randomized reduction.
    \end{proof}

    \paragraph{Non-commutative Sum-of-Squares lower bound.} Our overall approach here is to appeal to \Cref{thm:ncsosvalue} for our commuting construction and prove classical Sum-of-Squares lower bounds individually.

    \begin{proof}[Proof of \Cref{item:ncsos}]
        Let $\calJ_X,\calJ_Z$ be the two classical instances underlying the new Hamiltonian $\bH$. As noted earlier, all terms commute since they are products of the original checks, allowing us to invoke \Cref{thm:ncsosvalue},
        \begin{align*}
            \ncSoS_{d(n)}(\bH) = \frac{\abs{\calH_X}}{m} \SoS_{d(n)}(\calJ_X) + \frac{m-\abs{\calH_X}}{m} \SoS_{d(n)}(\calJ_Z)\mcom
        \end{align*}
        and show $\SoS_{d(n)}(\calJ_X) = 1$ and $\SoS_{d(n)}(\calJ_Z) = 1$ for each instance separately. Moreover, since $\calJ_X$ is just the original code Hamiltonian containing only positive terms, we have $\val(\calJ_X) = 1$, meaning the relaxation value $\SoS_{d(n)}(\calJ_X) = 1$ for any $d(n) \geq k_0$. Our efforts go to proving an $\SoS_{d(n)}(\calJ_Z) = 1$ for some $d(n) = \Omega(n)$.

        For the $Z$-instance, we start outside the pseudo-expectation view and just suppose we set every variable in $W$ to $0$. This immediately satisfies all clauses in $\calH_{\mathrm{int}}$, since the sign is $+1$ for each, so the $\F_2$ equation is set equal to $0$. For each clause $\calH_{\mathrm{slack}}$, all variables within $W$ are simply removed, but since they are not fully supported on $W$ some term remains. Finally, the clauses corresponding to $\calH_{\NP}$ combine both checks from $\calH_{\mathrm{int}}$ and $\calH_{\mathrm{slack}}$ with right hand sides varying between $\{0, 1\}$. The $\calH_{\mathrm{int}}$ is removed during this process, deriving some simplified equations involving $\calH_{\mathrm{slack}}$ only.

Since this is a hard assignment, corresponding to a point mass on these variables, it can be tensored into any pseudo-expectation on the remaining variables $[n] \setminus W$. Our goal is thus to design a high-degree pseudo-expectation on these remaining variables that believes the remaining equations are satisfiable.

Towards this, define
\begin{align*}
    B\in\F_2^{r\times [n] \setminus W}\mcom
    \qquad B_{i,a}=1\ \Longleftrightarrow\ a\in D_i\setminus W\mper
\end{align*}
This represents in matrix form the remaining clauses $D_1, \dots, D_r \in \calH_{\mathrm{slack}}$. Every row is nonzero, since $D_i\not\subseteq W$, and has weight at most
$k_0$. Each column has weight at most $3d_0$, including the at most two
repeated slack labels. In the remaining binary variables $y\in\F_2^{[n] \setminus W}$, the
restricted equations are exactly
\begin{align*}
    ABy = 0\mcom \qquad RBy = b\mcom
\end{align*}
in matrix form, with the latter counted $L$ times. We can write these together by putting $J = \begin{bmatrix}A\\R\end{bmatrix}B$ and
$s = \begin{bmatrix}
    0 \\ b
\end{bmatrix}$. To prove that this has an $\Omega(n)$-degree Sum-of-Squares relaxation value of $1$, we appeal to the $\XOR$-resolution framework via the expansion of the parity check matrix $J$.

\begin{theorem}[Small-set expansion $\implies$ Sum-of-Squares lower bounds \cite{Grigoriev01, Schoenebeck08, KothariMOW17, HopkinsL22}]
        \label{thm:kmow}
        Let $\calI = (\calH, \{b_C\}_{C \in \calH})$ be an instance of classical inhomogeneous $k$-$\XOR$. Let $H : \F_2^n \to \F_2^\calH$ be the associated parity check matrix for $\calH$. We say $H$ is a $(\rho_1, \rho_2)$-small-set expanding if it satisfies
        \begin{align*}
            \abs{H^\top y} \geq \rho_2\abs{y}\mcom \qquad \forall y \in \F_2^\calH\mcom \abs{y} \leq \rho_1\abs{\calH} \mper
        \end{align*}
        Then the $\Theta(\rho_1\rho_2 \abs{\calH})$-level Sum-of-Squares relaxation value for $\calI$ satisfies $\SoS_{\Theta(\rho_1\rho_2 \abs{\calH})}(\calI) = 1$.
    \end{theorem}

Note this is independent of the signs in $b$. Further, this is exactly where the re-encoding trick comes to play. If we simply used $D_1, \dots, D_r$ for slack directly, there would be dependence between the rows of $J$ ruining the expansion result. By taking random, short linear combinations, $J$ becomes expanding, allowing \Cref{thm:kmow} to go through. We prove this quantitatively in \Cref{sec:expansion}.

Apply \Cref{lem:expansion} to $B$ with row-weight bound $k_0$ and column-weight bound $3d_0$. The preceding construction guarantees $M\leq r/3$, $3\mid r$, and $r=\Theta(n)$. On the expansion event, the matrix $J$ containing one copy of each $RB$-equation has
\begin{equation*}
    q=\frac{\Delta r}{3}+M=\Theta(r),\qquad
    |J^\top u|\geq\frac{|u|}{2}\quad(0<|u|\leq\tau r)\mper
\end{equation*}
Thus \Cref{thm:kmow} applies to this one-copy instance with
\begin{equation*}
    \rho_1=\frac{\tau r}{q},\qquad \rho_2=\frac12,
    \qquad \rho_1\rho_2q=\frac{\tau r}{2}=\Omega(n)\mper
\end{equation*}
We may therefore fix a uniform even degree $d(n)=\Omega(n)$ at which its SoS value is $1$. Each individual equation has pseudo-satisfaction at most $1$, so value $1$ for their average means every equation has pseudo-satisfaction $1$. Repeating the $RB$-equations $L$ times consequently preserves this value; no expansion claim is made for the matrix with repeated rows. Extending the pseudo-expectation by the fixed assignment on $W$ also preserves its degree and gives value $1$ to every clause of $\calJ_Z$. Hence $\SoS_{d(n)}(\calJ_Z)=1$, completing the proof of \Cref{item:ncsos}.
    \end{proof}

    This completes the proof of \Cref{thm:ncsosforr} bar the proofs in \Cref{sec:expansion}.
\end{proof}

\subsection{\texorpdfstring{Non-commutative Sum-of-Squares bounds for random $k$-$\XOR$ Hamiltonians}{Non-commutative Sum-of-Squares bounds for random XOR Hamiltonians}}

In this section we prove \Cref{thm:randomxor}, bounding the non-commutative Sum-of-Squares value for random $k$-$\XOR$ Hamiltonians in a regime distinct from classical random $k$-$\XOR$.

We begin with a rather innocuous observation that non-commutative Sum-of-Squares notices an obstruction to frustration-freeness of anticommuting local Paulis in low-degree.

\begin{observation}
    \label{obs:ncsosknows}
    Let $\calI = (\calH, \{(P_C, b_C)\}_{C \in \calH})$ be an instance of Hamiltonian $k$-$\XOR$. Suppose there exist terms $C, C'$ in $\calH$ such that $P_C P_{C'} = -P_{C'}P_C$. Then there is a degree-$2k$ non-commutative Sum-of-Squares proof that $\lambdamax(\bH_\calI) < 1$.
\end{observation}

\begin{proof}
    Assume there exists a degree-$2k$ pseudo-expectation $\pE_\rho$ satisfying $\pE_\rho\sbra{\bH_\calI} = 1$. This immediately requires $\pE_\rho\sbra{P_C} = b_C$ and $\pE_\rho\sbra{P_{C'}} = b_{C'}$. Let $R = \Id - b_CP_C - b_{C'}P_{C'}$. By the positivity assumption on $\pE_\rho$ one has $\pE_\rho\sbra{R^\dagger R} \geq 0$, but we can alternatively compute 
    \begin{align*}
        \pE_\rho\sbra{R^\dagger R} &= \pE_\rho\sbra{3\Id - 2b_CP_C - 2b_{C'}P_{C'} + b_Cb_{C'}P_C P_{C'} + b_Cb_{C'}P_{C'}P_C }\\
        &= 3\pE_\rho\sbra{\Id} - 2b_C\pE_\rho\sbra{P_C} - 2b_{C'}\pE_\rho\sbra{P_{C'}} \\
        &= -1\mper
    \end{align*}
    In the second line we use the anticommutation assumption and linearity of $\pE_\rho$ to cancel the latter terms. Finally, we use the normalization property and our assumption on the values of $P_C$ and $P_{C'}$.
\end{proof}

To see the utility of \Cref{obs:ncsosknows}, notice that if given a Hamiltonian $k$-$\XOR$ instance with even a single anticommuting pair of terms, the non-commutative Sum-of-Squares relaxation value falls below $1$ at degree-$2k$, ruling out perfect completeness integrality gaps. Our main goal in this section is to push this result further, showing how to utilize the \emph{many} anticommuting pairs arising in a random instance to achieve a stronger refutation for random $k$-$\XOR$ Hamiltonians.

\begin{proof}[Proof of \Cref{thm:randomxor}]
    Our proof of \Cref{thm:randomxor} is broken into three pieces: \Cref{item:random1}, \Cref{item:random2}, and \Cref{item:random3}. Note, \Cref{item:random1} is immediate by applying \Cref{fact:quantumxorconcentration} at the given clause density.

    \begin{proof}[Proof of \Cref{item:random2}]
        Our proof of \Cref{item:random2} begins with a more sophisticated version of \Cref{obs:ncsosknows} which we prove here.

        \begin{lemma}
            \label{lem:star}
            Let $\calI = (\calH, \{(P_C, b_C)\}_{C \in \calH})$ be an instance of Hamiltonian $k$-$\XOR$. Suppose there exists a partition of $\calH = \calH_1 \sqcup \calH_2 \sqcup \calH_3$ with the property that every term $C \in \calH_i, C' \in \calH_j$ for $i \neq j$ one has $P_C P_{C'} = -P_{C'} P_C$. Then there is a degree-$2k$ non-commutative Sum-of-Squares proof that $\lambdamax(\bH_\calI) \leq \frac{1}{2} + \frac{\sqrt{\abs{\calH_1}^2 +\abs{\calH_2}^2+\abs{\calH_3}^2 }}{2\abs{\calH}}$.
        \end{lemma}

        \begin{proof}
            Write $\bH_i = \sum_{C \in \calH_i} b_C P_C$. By the anticommutation assumption we can write
            \begin{align*}
                (\bH_1 + \bH_2 + \bH_3)^2 = \bH_1^2 + \bH_2^2 + \bH_3^2\mper
            \end{align*}
            We can also calculate
            \begin{align*}
                \sum_{\{C, C'\} \in \binom{\calH_i}{2}} (b_C P_C - b_{C'} P_{C'})^2 &= \sum_{\{C, C'\} \in \binom{\calH_i}{2}} 2\Id - b_C b_{C'} P_C  P_{C'} - b_C b_{C'} P_{C'}P_C\\
                &= \abs{\calH_i}^2\Id - \sum_{C, C' \in \calH_i} b_C b_{C'} P_C P_{C'}\\
                &= \abs{\calH_i}^2\Id - \bH_i^2\mper
            \end{align*}
            Combining these two observations allows us to write
            \begin{align*}
                (\abs{\calH_1}^2 + \abs{\calH_2}^2 + \abs{\calH_3}^2) \Id -  (\bH_1 + \bH_2 + \bH_3)^2 = \sum_{i=1}^3 \sum_{\{C, C'\} \in \binom{\calH_i}{2}} (b_C P_C - b_{C'} P_{C'})^2 \mper
            \end{align*}
            Letting $m = \sqrt{\abs{\calH_1}^2 + \abs{\calH_2}^2 + \abs{\calH_3}^2}$, expanding out the square $(m \Id -(\bH_1 + \bH_2 + \bH_3))^2$, and rearranging, one arrives at the equation
            \begin{align*}
                m \Id - (\bH_1 + \bH_2 + \bH_3) = \frac{(m \Id - (\bH_1 + \bH_2 + \bH_3))^2}{2m} + \frac{\sum_{i=1}^3 \sum_{\{C, C'\} \in \binom{\calH_i}{2}} (b_C P_C - b_{C'} P_{C'})^2}{2m}
            \end{align*}
            Since every operator on the right side is a square of a $k$-local operator, this serves as a degree-$2k$ non-commutative Sum-of-Squares proof that $m\Id -  (\bH_1 + \bH_2 + \bH_3) \succeq 0$. Since $\bH_\calI = \frac{1}{2}\Id + \frac{1}{2\abs{\calH}}(\bH_1 + \bH_2 + \bH_3)$, rearranging yields the desired result.
        \end{proof}

        As is, the anticommutation bucketing of \Cref{lem:star} is too restrictive to apply to a random $k$-$\XOR$ Hamiltonian, but we show there exists a related decomposition which suffices via a sequence of applications.

        \begin{lemma}
            \label{lem:stardecomp}
            Let $\calI = (\calH, \{(P_C, b_C)\}_{C \in \calH})$ be a random instance of Hamiltonian $k$-$\XOR$. Then with high probability there exists a decomposition
            \begin{align*}
               \calI = \bigsqcup_{i = 1}^t \calI_i \sqcup \calI_{\mathrm{trash}} \mcom
            \end{align*}
            for $t \leq n + \abs{\calH}/\sqrt{n}$ and satisfying:
            \begin{enumerate}
                \item Each $\calH_i$ decomposes as $(\calH_{i,X}, \calH_{i,Y}, \calH_{i,Z})$ adhering to the desired pairwise anticommutation condition of \Cref{lem:star}.
                \item $\abs{\calH_{\mathrm{trash}}} = o(\abs{\calH})$.
                \item $\sum_{i=1}^t \abs{\calH_{i,X}}^2 + \abs{\calH_{i,Y}}^2 + \abs{\calH_{i,Z}}^2 \leq (1+o(1))\pbra{\abs{\calH} + \frac{\abs{\calH}^2}{3t}}$.
            \end{enumerate}
        \end{lemma}
        
        An immediate application of \Cref{lem:stardecomp} and the fact $\pE_\rho\sbra{P} \leq 1$ for any degree-$2k$ pseudo-expectation operator, we immediately get the bound
        \begin{align*}
            \ncSoS_{2k}(\bH_\calI) &\leq \sum_{i=1}^t \frac{\abs{\calH_i}}{\abs{\calH}} \ncSoS_{2k}(\bH_{\calI_i}) + \frac{\abs{\calH_{\text{trash}}}}{\abs{\calH}}\\
            &\leq \frac{1}{2} + \frac{1}{2\abs{\calH}}\sum_{i=1}^t \sqrt{\abs{\calH_{i,X}}^2 + \abs{\calH_{i,Y}}^2 + \abs{\calH_{i,Z}}^2} + \frac{\abs{\calH_{\text{trash}}}}{2\abs{\calH}}\\
            &\leq \frac{1}{2} + \frac{1}{2\abs{\calH}} \sqrt{t \sum_{i=1}^t (\abs{\calH_{i,X}}^2 + \abs{\calH_{i,Y}}^2 + \abs{\calH_{i,Z}}^2}) + \frac{\abs{\calH_{\text{trash}}}}{2\abs{\calH}}\\
            &\leq \frac{1}{2} + \frac{1}{2} \sqrt{\frac{1}{3} + \frac{n}{\abs{\calH}} } + o(1)\mper
        \end{align*}
        The third line here uses the Cauchy-Schwarz inequality and the last uses the last two assumptions of the lemma together with $t/\abs{\calH} \leq n/\abs{\calH} + n^{-1/2}$.

        \begin{proof}[Proof of \Cref{lem:stardecomp}]
            We prove the existence of our decomposition algorithmically. First, set $m = \abs{\calH}$, $s = \lceil n^{3/2}\rceil$, and $b = \lceil m/s\rceil$. Partition the clause labels into $b$ batches $\{\calH^{(j)}\}_{j \in [b]}$, independently of their sampled supports and Pauli labels, with sizes $m_j = \abs{\calH^{(j)}}$ differing by at most one. Thus $m_j \in \{\lfloor m/b\rfloor,\lceil m/b\rceil\}$ and $m_j \leq s$ for every $j$. Now take each term $C \in \calH^{(j)}$, designate a random root vertex $i \in C$ and put $C$ in the bucket $\calH^{(j)}_i$. Now, while any bucket $\calH^{(j)}_i$ contains $C, C' \in \calH^{(j)}_i$ such that $\abs{C \cap C'} \geq 2$, remove $C'$ arbitrarily. Finally, once this process is done, partition each $\calH_i = \calH^{(j)}_{i, X} \sqcup \calH^{(j)}_{i, Y} \sqcup \calH^{(j)}_{i, Z}$ depending on whether the $i$th tensor-product entry of each $C \in \calH^{(j)}_i$ is $X$, $Y$, or $Z$. By construction we have $t = nb = n\lceil m/s\rceil \leq n + m/\sqrt{n}$.

            Note, the first pairwise anticommutation relation is immediately satisfied by this construction. To see this, observe that, without loss of generality, any $C \in \calH^{(j)}_{i, X}$ and $C' \in \calH^{(j)}_{i, Y}$ satisfy $P_C P_{C'} = -P_{C'}P_C$. This is because by assumption $C$ and $C'$ are both non-trivial on the root index, one taking value $X$ and one value $Y$, and otherwise intersect on no other entries since we greedily removed pairs with larger intersection.

            For the second criterion, we note that since each $\calH^{(j)}$ is drawn uniformly at random, any two terms $C$ and $C'$ intersect in at least two points with probability at most $\frac{4k^4}{n^2}$. Each deletion can be charged to a distinct such pair within a batch. Since $m_j \leq s$, the expected total number of deleted terms is at most
            \begin{align*}
                \E\sbra{\abs{\calH_{\mathrm{trash}}}} \leq \frac{4k^4}{n^2}\sum_{j=1}^b m_j^2 \leq \frac{4k^4sm}{n^2} = O(m/\sqrt{n})\mper
            \end{align*}
            Markov's inequality therefore gives $\abs{\calH_{\mathrm{trash}}} \leq n^{-1/4}m = o(m)$ with probability $1-O(n^{-1/4})$.

            We finish by arguing the third and final criterion. We analyze the total collision count across all batches rather than taking a union bound over the batches. Before deletion, let $\boldsymbol{N}^{(j)}_\ell$ for $\ell \in [3n]$ denote the number of terms in batch $j$ assigned to bin $\ell$, where a bin specifies both the root vertex and its Pauli label. The independent clauses and root choices give independent uniform assignments to these $3n$ bins. Deletion can only decrease the squared bucket sizes, so
            \begin{align*}
                \boldsymbol{Z} &\coloneqq \sum_{j=1}^b\sum_{i=1}^n \pbra{\abs{\calH^{(j)}_{i,X}}^2 + \abs{\calH^{(j)}_{i,Y}}^2 + \abs{\calH^{(j)}_{i,Z}}^2} \\
                &\leq \sum_{j=1}^b\sum_{\ell=1}^{3n} (\boldsymbol{N}^{(j)}_\ell)^2 = m + \boldsymbol{W}\mcom
                \qquad \boldsymbol{W} \coloneqq 2\sum_{j=1}^b\sum_{\ell=1}^{3n}\binom{\boldsymbol{N}^{(j)}_\ell}{2}\mper
            \end{align*}
            The balanced batch sizes and $b \leq m$ imply
            \begin{align*}
                \sum_{j=1}^b m_j^2 = \frac{m^2}{b} + \sum_{j=1}^b \pbra{m_j-\frac{m}{b}}^2 \leq \frac{m^2}{b}+b \leq \frac{m^2}{b}+m\mper
            \end{align*}
            Every unordered pair of distinct clause labels in the same batch collides with probability $1/3n$. These collision indicators are pairwise independent: for pairs sharing one label, the joint collision probability is $1/9n^2$, and disjoint pairs depend on independent assignments. Consequently,
            \begin{align*}
                \E\sbra{\boldsymbol{W}} &= \frac{1}{3n}\sum_{j=1}^b m_j(m_j-1) \leq \frac{m^2}{3nb} = \frac{m^2}{3t}\mcom \\
                \Var\sbra{\boldsymbol{W}} &= 4\sum_{j=1}^b\binom{m_j}{2}\frac{1}{3n}\pbra{1-\frac{1}{3n}} \leq \frac{2m^2}{3t}\mper
            \end{align*}
            Chebyshev's inequality and $t \geq n$ now give
            \begin{align*}
                \Pr\sbra{\boldsymbol{W} > \E\sbra{\boldsymbol{W}} + n^{-1/4}m}
                \leq \frac{2m^2/(3t)}{n^{-1/2}m^2} \leq \frac{2}{3\sqrt{n}}\mper
            \end{align*}
            Thus, with high probability,
            \begin{align*}
                \boldsymbol{Z} \leq m + \frac{m^2}{3t} + n^{-1/4}m \leq (1+n^{-1/4})\pbra{m+\frac{m^2}{3t}}\mcom
            \end{align*}
            which is the required bound.
        \end{proof}
        This completes the proof of \Cref{item:random2} of \Cref{thm:randomxor}.
    \end{proof}

    \begin{proof}[Proof of \Cref{item:random3}]
        In this part, we prove a non-commutative Sum-of-Squares lower bound putting limitations on the approach of \Cref{item:random2} by using an imperfect Boolean-to-Pauli pseudo-expectation built through a randomized lifting process.

        \begin{lemma}[SoS-hardness $\implies$ ncSoS-hardness]
            \label{lem:lifting2}
            Let $\pE_\mu$ be a $3n$-variable (grouped as $x_{i, P}$ for $i \in [n]$, $P \in \{X, Y, Z\}$) degree-$d$ Boolean pseudo-expectation. Define the degree-$d$ pseudo-expectation $\pE_\rho$ on $P \in \calP_{\leq d}(n)$ via
            \begin{align*}
                \pE_\rho\sbra{P} = \frac{1}{3^{\deg(P)}}\pE_\mu\sbra{x_P}\mper
            \end{align*}
            where
            \begin{align*}
                x_P = \prod_{i \in \supp(P)} x_{i, P_i}\mper
            \end{align*}
            Then $\pE_\rho$ is a valid pseudo-expectation.
        \end{lemma}

        It is useful to compare \Cref{lem:lifting2} to the single-basis embedding of \Cref{thm:lifting}. In that setting, we saw how to embed $n$-variable Boolean polynomials in a product Pauli basis to lift Sum-of-Squares lower bounds to non-commutative Sum-of-Squares lower bounds. In this setting, the object we want exists in many bases simultaneously. To address this, we simply take a Boolean pseudo-expectation and lift it to a random product Pauli basis.

        \begin{proof}[Proof of \Cref{lem:lifting2}]
            Let $\pE_\mu$ be a degree-$d$ pseudo-expectation on $3n$ variables. For any full-weight Pauli operator $Q \in \calP_n(n)$, we define for $P \in \calP_{\leq d}(n)$
            \begin{equation*}
                {\pE_\rho}^{(Q)}\sbra{P} = \begin{cases}
                    \pE_\mu\sbra{x_P} & P \sqsubseteq Q\\
                    0 & \text{otherwise}
                \end{cases}
            \end{equation*}
            Our key observation is that
            \begin{align*}
                \pE_\rho\sbra{P} = \E_{\boldsymbol{Q} \sim \calP_n(n)}\sbra{{\pE_\rho}^{(Q)}\sbra{P}}\mper
            \end{align*}
            From here, the proof is immediate. Since each ${\pE_\rho}^{(Q)}$ is just a single-basis restriction of $\pE_\mu$, it inherits degree-$d$ pseudo-expectation properties via \Cref{thm:lifting}, and since $\pE_\rho$ is just a convex combination of them, it too is a degree-$d$ pseudo-expectation.

            To argue the key observation, just note for any $P \in \calP_{\leq d}(n)$,
            \begin{align*}
                \E_{\boldsymbol{Q} \sim \calP_n(n)}\sbra{{\pE_\rho}^{(Q)}\sbra{P}} = \frac{1}{3^{\deg(P)}} \pE_\mu\sbra{x_P}\mcom
            \end{align*}
            since with probability exactly $3^{-\deg(P)}$, $P \sqsubseteq Q$, and in this case the value is always $\pE_\mu\sbra{x_P}$, and otherwise $0$.
        \end{proof}

        Now, to apply \Cref{lem:lifting2} is simple. Suppose there exists a degree-$d$ pseudo-expectation $\pE_\mu$ on the $3n$ variable encoding of a random $k$-$\XOR$ Hamiltonian instance $\calI = (\calH, \{(P_C, b_C)\}_{C \in \calH})$ with $\pE_\mu\sbra{b_C x_{P_C}} = 1$ for all $C \in \calH$. In other words, $\pE_\mu$ believes the corresponding classical $k$-$\XOR$ instance on $3n$ variables is satisfiable. Applying \Cref{lem:lifting2} yields
        \begin{align*}
            \pE_\rho\sbra{\bH_\calI} &= \frac{1}{2} + \frac{1}{2\abs{\calH}}\sum_{C \in \calH} \pE_\rho\sbra{b_C P_C}\\
            &= \frac{1}{2} + \frac{1}{2 \cdot 3^k\abs{\calH}}\sum_{C \in \calH} \pE_\mu\sbra{b_C x_{P_C}}\\
            &= \frac{1}{2} + \frac{1}{2 \cdot 3^k}\mper
        \end{align*}
        To argue $\pE_\mu$ exists, we would like to appeal to \cite[Theorem 1.2]{KothariMOW17} directly, which gives random $k$-$\XOR$ lower bounds matching the refutation threshold. The only problem is that the $3n$-variable embedding of random Hamiltonian $k$-$\XOR$ is not exactly a random $k$-$\XOR$ instance, since any $X, Y, Z$ copies of the same qubit term cannot exist simultaneously. 
        
        Luckily, we can easily show by rejection sampling that this effect is negligible. Consider sampling supports $C \sim \binom{3n}{k}$ that specify our random Pauli operators. This corresponds to a random classical $k$-$\XOR$ instance on $3n$ variables. We can thus apply \cite{KothariMOW17} to achieve a pseudo-expectation $\pE_\mu$ with $\pE_\mu\sbra{b_C x_{P_C}} = 1$ for all $C \in \calH$.
        
        The bad event then is when a collision $C \ni x_{i, X}, x_{i, Y}$ occurs without loss of generality. In this case, we throw out this term and resample the corresponding term from $\calP_k(n)$ directly. If $r$ clauses are thrown out then we still have
        \begin{align*}
            \frac{1}{m}\sum_{C \in \calH} b_C \pE_\mu\sbra{x_{P_C}} \geq 1-\frac{2r}{m}\mcom
        \end{align*}
        since the $r$ new clauses contribute $-1$ in the worst case. It then suffices to show $r = o(m)$, but since the collision probability is at most $O(k^2/n)$, this is immediate.
    \end{proof}

    With this, the proof of \Cref{thm:randomxor} is complete.
\end{proof}

\bibliographystyle{alpha}
\bibliography{references}

@article{Hastad01,
  title={{Some optimal inapproximability results}},
  author={H{\aa}stad, Johan},
  journal={Journal of the ACM (JACM)},
  volume={48},
  number={4},
  pages={798--859},
  year={2001},
  publisher={ACM New York, NY, USA}
}

@article{ImpagliazzoP01,
  author    = {Russell Impagliazzo and
               Ramamohan Paturi},
  title     = {{On the Complexity of k-SAT}},
  journal   = {J. Comput. Syst. Sci.},
  volume    = {62},
  number    = {2},
  pages     = {367--375},
  year      = {2001}
}

@inproceedings{CobleCoudronNelsonNezhadi23Stabilizer,
    author = {Nolan J. Coble and Matthew Coudron and Jon Nelson and Seyed Sajjad Nezhadi},
    title = {Local Hamiltonians with No Low-Energy Stabilizer States},
    booktitle = {18th Conference on the Theory of Quantum Computation, Communication and Cryptography (TQC 2023)},
    series = {Leibniz International Proceedings in Informatics (LIPIcs)},
    volume = {266},
    pages = {14:1--14:21},
    year = {2023},
    doi = {10.4230/LIPIcs.TQC.2023.14},
    eprint = {2302.14755},
    archivePrefix = {arXiv},
    primaryClass = {quant-ph}
}

@misc{CobleCoudronNelsonNezhadi23TCount,
    author = {Nolan J. Coble and Matthew Coudron and Jon Nelson and Seyed Sajjad Nezhadi},
    title = {Hamiltonians whose low-energy states require {\(\Omega\)}(n) {T} gates},
    year = {2023},
    eprint = {2310.01347},
    archivePrefix = {arXiv},
    primaryClass = {quant-ph},
    doi = {10.48550/arXiv.2310.01347}
}

@inproceedings{Parham26Magic,
    author = {Natalie Parham},
    title = {Quantum Circuit Lower Bounds in the Magic Hierarchy},
    booktitle = {Proceedings of the 58th Annual ACM Symposium on Theory of Computing},
    year = {2026},
    pages = {734--743},
    doi = {10.1145/3798129.3800788}
}

@misc{QuantumNPsurvey,
Author = {Dorit Aharonov and Tomer Naveh},
Title = {Quantum {NP} - {A} {S}urvey},
Year = {2002},
Eprint = {arXiv:quant-ph/0210077},
}

@article{qpcp-survey,
author = {Aharonov, Dorit and Arad, Itai and Vidick, Thomas},
title = {{Guest {C}olumn: The {Q}uantum {P}{C}{P} {C}onjecture}},
year = {2013},
issue_date = {June 2013},
publisher = {Association for Computing Machinery},
address = {New York, NY, USA},
volume = {44},
number = {2},
issn = {0163-5700},
url = {https://doi.org/10.1145/2491533.2491549},
doi = {10.1145/2491533.2491549},
journal = {SIGACT News},
month = jun,
pages = {47–79},
numpages = {33}
}

@inproceedings{AroraKK95,
  author    = {Sanjeev Arora and
               David R. Karger and
               Marek Karpinski},
  title     = {{Polynomial time approximation schemes for dense instances of \emph{NP}-hard
               problems}},
  booktitle = {Proceedings of the Twenty-Seventh Annual {ACM} Symposium on Theory
               of Computing, 29 May-1 June 1995, Las Vegas, Nevada, {USA}},
  pages     = {284--293},
  publisher = {{ACM}},
  year      = {1995},
}

@inproceedings{FotakisLP16,
  author    = {Dimitris Fotakis and
               Michael Lampis and
               Vangelis Th. Paschos},
  title     = {{Sub-exponential Approximation Schemes for CSPs: From Dense to Almost Sparse}},
  booktitle = {33rd Symposium on Theoretical Aspects of Computer Science, {STACS}
               2016, February 17-20, 2016, Orl{\'{e}}ans, France},
  series    = {LIPIcs},
  volume    = {47},
  pages     = {37:1--37:14},
  publisher = {Schloss Dagstuhl - Leibniz-Zentrum f{\"{u}}r Informatik},
  year      = {2016},
}

@inproceedings{AlrabiahGKM23,
  author       = {Omar Alrabiah and
                  Venkatesan Guruswami and
                  Pravesh K. Kothari and
                  Peter Manohar},
  title        = {{A Near-Cubic Lower Bound for 3-Query Locally Decodable Codes from
                  Semirandom CSP Refutation}},
  booktitle    = {Proceedings of the 55th Annual {ACM} Symposium on Theory of Computing,
                  {STOC} 2023, Orlando, FL, USA, June 20-23, 2023},
  pages        = {1438--1448},
  publisher    = {{ACM}},
  year         = {2023},
}

@article{TallerV25,
  title={{Approximating the quantum value of an LCS game is RE-hard}},
  author={Aviv Taller and Thomas Vidick},
  journal={arXiv},
  year={2025}
}

@inproceedings{KothariM23,
  author       = {Pravesh K. Kothari and
                  Peter Manohar},
  title        = {{An Exponential Lower Bound for Linear 3-Query Locally Correctable
                  Codes}},
  booktitle    = {Proceedings of the 56th Annual {ACM} Symposium on Theory of Computing,
                  {STOC} 2024, Vancouver, BC, Canada, June 24-28, 2024},
  pages        = {776--787},
  publisher    = {{ACM}},
  year         = {2024}
}

@inproceedings{AbascalGK21,
  author    = {Jackson Abascal and
               Venkatesan Guruswami and
               Pravesh K. Kothari},
  title     = {{Strongly refuting all semi-random Boolean CSPs}},
  booktitle = {Proceedings of the 2021 {ACM-SIAM} Symposium on Discrete Algorithms,
               {SODA} 2021, Virtual Conference, January 10 - 13, 2021},
  pages     = {454--472},
  publisher = {{SIAM}},
  year      = {2021},
}

@inproceedings{Schoenebeck08,
  author       = {Grant Schoenebeck},
  title        = {{Linear Level Lasserre Lower Bounds for Certain k-CSPs}},
  booktitle    = {49th Annual {IEEE} Symposium on Foundations of Computer Science, {FOCS}
                  2008, October 25-28, 2008, Philadelphia, PA, {USA}},
  pages        = {593--602},
  publisher    = {{IEEE} Computer Society},
  year         = {2008},
}

@article{Grigoriev01,
author = {Dima Grigoriev},
title = {{Linear lower bound on degrees of Positivstellensatz calculus proofs for the parity}},
journal = {Theoretical Computer Science},
volume = {259},
number = {1},
pages = {613-622},
year = {2001},
issn = {0304-3975},
}

@inproceedings{AllenOW15,
  author    = {Sarah R. Allen and
               Ryan O'Donnell and
               David Witmer},
  title     = {{How to Refute a Random CSP}},
  booktitle = {{IEEE} 56th Annual Symposium on Foundations of Computer Science, {FOCS}
               2015, Berkeley, CA, USA, 17-20 October, 2015},
  pages     = {689--708},
  publisher = {{IEEE} Computer Society},
  year      = {2015},
}

@inproceedings{RaghavendraRS17,
  author    = {Prasad Raghavendra and
               Satish Rao and
               Tselil Schramm},
  title     = {{Strongly refuting random CSPs below the spectral threshold}},
  booktitle = {Proceedings of the 49th Annual {ACM} {SIGACT} Symposium on Theory
               of Computing, {STOC} 2017, Montreal, QC, Canada, June 19-23, 2017},
  pages     = {121--131},
  publisher = {{ACM}},
  year      = {2017},
}

@inproceedings{MoshkovitzR08,
  author       = {Dana Moshkovitz and
                  Ran Raz},
  title        = {{Two Query PCP with Sub-Constant Error}},
  booktitle    = {49th Annual {IEEE} Symposium on Foundations of Computer Science, {FOCS}
                  2008, Philadelphia, PA, USA, October 25-28, 2008},
  pages        = {314--323},
  publisher    = {{IEEE} Computer Society},
  year         = {2008},
}

@inproceedings{Khot02,
  author       = {Subhash Khot},
  title        = {{On the Power of Unique 2-Prover 1-Round Games}},
  booktitle    = {Proceedings of the 17th Annual {IEEE} Conference on Computational
                  Complexity, Montr{\'{e}}al, Qu{\'{e}}bec, Canada, May 21-24,
                  2002},
  pages        = {25},
  publisher    = {{IEEE} Computer Society},
  year         = {2002}
}

@inproceedings{Raghavendra08,
  author       = {Prasad Raghavendra},
  title        = {{Optimal algorithms and inapproximability results for every CSP?}},
  booktitle    = {Proceedings of the 40th Annual {ACM} Symposium on Theory of Computing,
                  Victoria, British Columbia, Canada, May 17-20, 2008},
  pages        = {245--254},
  publisher    = {{ACM}},
  year         = {2008},
}

@inproceedings{BafnaM0Y25,
  author       = {Mitali Bafna and
                  Dor Minzer and
                  Nikhil Vyas and
                  Zhiwei Yun},
  title        = {{Quasi-Linear Size PCPs with Small Soundness from HDX}},
  booktitle    = {Proceedings of the 57th Annual {ACM} Symposium on Theory of Computing,
                  {STOC} 2025, Prague, Czechia, June 23-27, 2025},
  pages        = {45--53},
  publisher    = {{ACM}},
  year         = {2025},
}

@inproceedings{KothariMOW17,
  author    = {Pravesh K. Kothari and
               Ryuhei Mori and
               Ryan O'Donnell and
               David Witmer},
  title     = {{Sum of squares lower bounds for refuting any CSP}},
  booktitle = {Proceedings of the 49th Annual {ACM} {SIGACT} Symposium on Theory
               of Computing, {STOC} 2017, Montreal, QC, Canada, June 19-23, 2017},
  pages     = {132--145},
  publisher = {{ACM}},
  year      = {2017},
}

@inproceedings{GuruswamiKM22,
  author    = {Venkatesan Guruswami and
               Pravesh K. Kothari and
               Peter Manohar},
  title     = {{Algorithms and certificates for Boolean CSP refutation: smoothed
               is no harder than random}},
  booktitle = {{STOC} '22: 54th Annual {ACM} {SIGACT} Symposium on Theory of Computing,
               Rome, Italy, June 20 - 24, 2022},
  pages     = {678--689},
  publisher = {{ACM}},
  year      = {2022}
}

@inproceedings{HsiehKM23,
  author    = {Jun{-}Ting Hsieh and
               Pravesh K. Kothari and
               Sidhanth Mohanty},
  title     = {{A simple and sharper proof of the hypergraph Moore bound}},
  booktitle = {Proceedings of the 2023 {ACM-SIAM} Symposium on Discrete Algorithms,
               {SODA} 2023, Florence, Italy, January 22-25, 2023},
  pages     = {2324--2344},
  publisher = {{SIAM}},
  year      = {2023}
}

@inproceedings{KempeKR04,
  author       = {Julia Kempe and
                  Alexei Y. Kitaev and
                  Oded Regev},
  title        = {{The Complexity of the Local Hamiltonian Problem}},
  booktitle    = {{FSTTCS} 2004: Foundations of Software Technology and Theoretical
                  Computer Science, 24th International Conference, Chennai, India, December
                  16-18, 2004, Proceedings},
  series       = {Lecture Notes in Computer Science},
  volume       = {3328},
  pages        = {372--383},
  publisher    = {Springer},
  year         = {2004},
}

@inproceedings{WeinAM19,
  author    = {Alexander S. Wein and
               Ahmed El Alaoui and
               Cristopher Moore},
  title     = {{The Kikuchi Hierarchy and Tensor {PCA}}},
  booktitle = {60th {IEEE} Annual Symposium on Foundations of Computer Science, {FOCS}
               2019, Baltimore, Maryland, USA, November 9-12, 2019},
  pages     = {1446--1468},
  publisher = {{IEEE} Computer Society},
  year      = {2019}
}

@book{KitaevSV02,
  author       = {Alexei Y. Kitaev and
                  A. H. Shen and
                  Mikhail N. Vyalyi},
  title        = {{Classical and Quantum Computation}},
  series       = {Graduate studies in mathematics},
  volume       = {47},
  publisher    = {American Mathematical Society},
  year         = {2002},
}

@inproceedings{LeeP24,
  author       = {Eunou Lee and
                  Ojas Parekh},
  title        = {{An Improved Quantum Max Cut Approximation via Maximum Matching}},
  booktitle    = {51st International Colloquium on Automata, Languages, and Programming,
                  {ICALP} 2024, July 8-12, 2024, Tallinn, Estonia},
  series       = {LIPIcs},
  volume       = {297},
  pages        = {105:1--105:11},
  publisher    = {Schloss Dagstuhl - Leibniz-Zentrum f{\"{u}}r Informatik},
  year         = {2024},
}

@inproceedings{HwangNP0W23,
  author       = {Yeongwoo Hwang and
                  Joe Neeman and
                  Ojas Parekh and
                  Kevin Thompson and
                  John Wright},
  title        = {{Unique Games hardness of Quantum Max-Cut, and a conjectured vector-valued
                  Borell's inequality}},
  booktitle    = {Proceedings of the 2023 {ACM-SIAM} Symposium on Discrete Algorithms,
                  {SODA} 2023, Florence, Italy, January 22-25, 2023},
  pages        = {1319--1384},
  publisher    = {{SIAM}},
  year         = {2023},
}

@inproceedings{HothemP023,
  author       = {Daniel Hothem and
                  Ojas Parekh and
                  Kevin Thompson},
  title        = {{Improved Approximations for Extremal Eigenvalues of Sparse Hamiltonians}},
  booktitle    = {18th Conference on the Theory of Quantum Computation, Communication
                  and Cryptography, {TQC} 2023, July 24-28, 2023, Aveiro, Portugal},
  series       = {LIPIcs},
  volume       = {266},
  pages        = {6:1--6:10},
  publisher    = {Schloss Dagstuhl - Leibniz-Zentrum f{\"{u}}r Informatik},
  year         = {2023},
}

@inproceedings{HallgrenLP20,
  author       = {Sean Hallgren and
                  Eunou Lee and
                  Ojas Parekh},
  title        = {{An Approximation Algorithm for the MAX-2-Local Hamiltonian Problem}},
  booktitle    = {Approximation, Randomization, and Combinatorial Optimization. Algorithms
                  and Techniques, {APPROX/RANDOM} 2020, August 17-19, 2020, Virtual
                  Conference},
  series       = {LIPIcs},
  volume       = {176},
  pages        = {59:1--59:18},
  publisher    = {Schloss Dagstuhl - Leibniz-Zentrum f{\"{u}}r Informatik},
  year         = {2020},
}

@inproceedings{Lee22,
  author       = {Eunou Lee},
  title        = {{Optimizing Quantum Circuit Parameters via SDP}},
  booktitle    = {33rd International Symposium on Algorithms and Computation, {ISAAC}
                  2022, December 19-21, 2022, Seoul, Korea},
  series       = {LIPIcs},
  volume       = {248},
  pages        = {48:1--48:16},
  publisher    = {Schloss Dagstuhl - Leibniz-Zentrum f{\"{u}}r Informatik},
  year         = {2022},
}

@article{ChenDBBT23,
  title = {{Sparse Random Hamiltonians Are Quantumly Easy}},
  author = {Chen, Chi-Fang and Dalzell, Alexander M. and Berta, Mario and Brand\~ao, Fernando G. S. L. and Tropp, Joel A.},
  journal = {Phys. Rev. X},
  volume = {14},
  issue = {1},
  pages = {011014},
  numpages = {29},
  year = {2024},
}

@article{AnschuetzGK25,
  author  = {Eric R. Anschuetz and David Gamarnik and Bobak T. Kiani},
  title   = {{Bounds on the Ground State Energy of Quantum {$p$}-Spin Hamiltonians}},
  journal = {Communications in Mathematical Physics},
  volume  = {406},
  pages   = {232},
  year    = {2025}
}

@article{Schuch11,
  author       = {Norbert Schuch},
  title        = {Complexity of commuting Hamiltonians on a square lattice of qubits},
  journal      = {Quantum Inf. Comput.},
  volume       = {11},
  number       = {11-12},
  pages        = {901--912},
  year         = {2011}
}

@article{Tropp15,
  author       = {Joel A. Tropp},
  title        = {{An Introduction to Matrix Concentration Inequalities}},
  journal      = {Found. Trends Mach. Learn.},
  volume       = {8},
  number       = {1-2},
  pages        = {1--230},
  year         = {2015},
}

@inproceedings{AnshuBN23,
  author       = {Anurag Anshu and
                  Nikolas P. Breuckmann and
                  Chinmay Nirkhe},
  title        = {{NLTS Hamiltonians from Good Quantum Codes}},
  booktitle    = {Proceedings of the 55th Annual {ACM} Symposium on Theory of Computing,
                  {STOC} 2023, Orlando, FL, USA, June 20-23, 2023},
  pages        = {1090--1096},
  publisher    = {{ACM}},
  year         = {2023},
}

@inproceedings{KocurekM25,
  author       = {Nicholas Kocurek and
                  Peter Manohar},
  title        = {{Spectral Refutations of Semirandom k-LIN over Larger Fields}},
  booktitle    = {Approximation, Randomization, and Combinatorial Optimization. Algorithms
                  and Techniques, {APPROX/RANDOM} 2025, August 11-13, 2025, Berkeley,
                  CA, {USA}},
  series       = {LIPIcs},
  volume       = {353},
  pages        = {17:1--17:15},
  publisher    = {Schloss Dagstuhl - Leibniz-Zentrum f{\"{u}}r Informatik},
  year         = {2025},
}

@inproceedings{HopkinsL22,
  author       = {Max Hopkins and
                  Ting{-}Chun Lin},
  title        = {{Explicit Lower Bounds Against {\(\Omega\)}(n)-Rounds of Sum-of-Squares}},
  booktitle    = {63rd {IEEE} Annual Symposium on Foundations of Computer Science, {FOCS}
                  2022, Denver, CO, USA, October 31 - November 3, 2022},
  pages        = {662--673},
  publisher    = {{IEEE}},
  year         = {2022},
}

@book{ODonnell14,
  title={{Analysis of Boolean Functions}},
  author={{O'Donnell}, Ryan},
  year={2014},
  publisher={Cambridge University Press}
}

@article{BenabbasGMT12,
  title={{SDP gaps from pairwise independence}},
  author={Benabbas, Siavosh and Georgiou, Konstantinos and Magen, Avner and Tulsiani, Madhur},
  journal={Theory of Computing},
  volume={8},
  number={1},
  pages={269--289},
  year={2012},
  publisher={Theory of Computing Exchange}
}

@inproceedings{BarakCK15,
  author    = {Boaz Barak and
               Siu On Chan and
               Pravesh K. Kothari},
  title     = {{Sum of Squares Lower Bounds from Pairwise Independence}},
  booktitle = {Proceedings of the Forty-Seventh Annual {ACM} on Symposium on Theory
               of Computing, {STOC} 2015, Portland, OR, USA, June 14-17, 2015},
  pages     = {97--106},
  publisher = {{ACM}},
  year      = {2015},
}

@inproceedings{MoriW16,
  author    = {Ryuhei Mori and
               David Witmer},
  title     = {{Lower Bounds for {CSP} Refutation by {SDP} Hierarchies}},
  booktitle = {Approximation, Randomization, and Combinatorial Optimization. Algorithms
               and Techniques, {APPROX/RANDOM} 2016, September 7-9, 2016, Paris,
               France},
  series    = {LIPIcs},
  volume    = {60},
  pages     = {41:1--41:30},
  year      = {2016}
}

@inproceedings{BrandaoH13a,
  author       = {Fernando G. S. L. Brand{\~{a}}o and
                  Aram W. Harrow},
  title        = {{Product-state approximations to quantum ground states}},
  booktitle    = {Symposium on Theory of Computing Conference, STOC'13, Palo Alto, CA,
                  USA, June 1-4, 2013},
  pages        = {871--880},
  publisher    = {{ACM}},
  year         = {2013},
}

@inproceedings{AlevJT19,
  author       = {Vedat Levi Alev and
                  Fernando Granha Jeronimo and
                  Madhur Tulsiani},
  title        = {{Approximating Constraint Satisfaction Problems on High-Dimensional
                  Expanders}},
  booktitle    = {60th {IEEE} Annual Symposium on Foundations of Computer Science, {FOCS}
                  2019, Baltimore, Maryland, USA, November 9-12, 2019},
  pages        = {180--201},
  publisher    = {{IEEE} Computer Society},
  year         = {2019},
}

@misc{MaN25,
      title={{Two bases suffice for QMA1-completeness}}, 
      author={Henry Ma and Anand Natarajan},
      year={2025},
}

@inproceedings{MousaviS25,
  author       = {Hamoon Mousavi and
                  Taro Spirig},
  title        = {{A Quantum Unique Games Conjecture}},
  booktitle    = {16th Innovations in Theoretical Computer Science Conference, {ITCS}
                  2025, January 7-10, 2025, Columbia University, New York, NY, {USA}},
  series       = {LIPIcs},
  volume       = {325},
  pages        = {76:1--76:16},
  publisher    = {Schloss Dagstuhl - Leibniz-Zentrum f{\"{u}}r Informatik},
  year         = {2025},
  url          = {https://doi.org/10.4230/LIPIcs.ITCS.2025.76},
}

@inproceedings{GhoshJJPR20,
  author       = {Mrinalkanti Ghosh and
                  Fernando Granha Jeronimo and
                  Chris Jones and
                  Aaron Potechin and
                  Goutham Rajendran},
  title        = {{Sum-of-Squares Lower Bounds for Sherrington-Kirkpatrick via Planted
                  Affine Planes}},
  booktitle    = {61st {IEEE} Annual Symposium on Foundations of Computer Science, {FOCS}
                  2020, Durham, NC, USA, November 16-19, 2020},
  pages        = {954--965},
  publisher    = {{IEEE}},
  year         = {2020},
  url          = {https://doi.org/10.1109/FOCS46700.2020.00093}
}

@inproceedings{CulfMS24,
  author       = {Eric Culf and
                  Hamoon Mousavi and
                  Taro Spirig},
  title        = {{Approximation Algorithms for Noncommutative CSPs}},
  booktitle    = {65th {IEEE} Annual Symposium on Foundations of Computer Science, {FOCS}
                  2024, Chicago, IL, USA, October 27-30, 2024},
  pages        = {920--929},
  publisher    = {{IEEE}},
  year         = {2024},
  url          = {https://doi.org/10.1109/FOCS61266.2024.00061},
}

@inproceedings{HastingsO22,
  author       = {Matthew B. Hastings and
                  Ryan O'Donnell},
  title        = {{Optimizing strongly interacting fermionic Hamiltonians}},
  booktitle    = {{STOC} '22: 54th Annual {ACM} {SIGACT} Symposium on Theory of Computing,
                  Rome, Italy, June 20 - 24, 2022},
  pages        = {776--789},
  publisher    = {{ACM}},
  year         = {2022},
  url          = {https://doi.org/10.1145/3519935.3519960},
}

@inproceedings{SchmidhuberOKB25,
  author       = {Alexander Schmidhuber and
                  Ryan O'Donnell and
                  Robin Kothari and
                  Ryan Babbush},
  title        = {{Quartic quantum speedups for planted inference}},
  booktitle    = {Proceedings of the 2025 Annual {ACM-SIAM} Symposium on Discrete Algorithms,
                  {SODA} 2025, New Orleans, LA, USA, January 12-15, 2025},
  pages        = {905--913},
  publisher    = {{SIAM}},
  year         = {2025},
}

@misc{GuptaHOS25,
      title={{A Classical Quadratic Speedup for Planted $k$XOR}}, 
      author={Meghal Gupta and William He and Ryan O'Donnell and Noah G. Singer},
      year={2025},
      eprint={2508.09422},
      archivePrefix={arXiv},
      primaryClass={cs.DS},
      url={https://arxiv.org/abs/2508.09422}, 
}

@article{LaurentM00,
    author = {B. Laurent and P. Massart},
    title = {{Adaptive estimation of a quadratic functional by model selection}},
    volume = {28},
    journal = {The Annals of Statistics},
    number = {5},
    publisher = {Institute of Mathematical Statistics},
    pages = {1302 -- 1338},
    year = {2000},
}

@article{AnschuetzCKK24,
    author = "Anschuetz, Eric R. and Chen, Chi-Fang and Kiani, Bobak T. and King, Robbie",
    title = "{{Strongly Interacting Fermions Are Nontrivial yet Nonglassy}}",
    eprint = "2408.15699",
    archivePrefix = "arXiv",
    primaryClass = "quant-ph",
    doi = "10.1103/cbqf-d24r",
    journal = "Phys. Rev. Lett.",
    volume = "135",
    number = "3",
    pages = "030602",
    year = "2025"
}

@article{BarakS14,
  author       = {Boaz Barak and
                  David Steurer},
  title        = {{Sum-of-squares proofs and the quest toward optimal algorithms}},
  journal      = {CoRR},
  volume       = {abs/1404.5236},
  year         = {2014},
}

@article{BaldwinS20,
  author  = {C. L. Baldwin and Brian Swingle},
  title   = {Quenched vs Annealed: Glassiness from {SK} to {SYK}},
  journal = {Physical Review X},
  volume  = {10},
  number  = {3},
  pages   = {031026},
  year    = {2020}
}

@article{ErdosS14,
  author  = {L{\'a}szl{\'o} Erd{\H{o}}s and Dominik J. Schr{\"o}der},
  title   = {{Phase Transition in the Density of States of Quantum Spin Glasses}},
  journal = {Mathematical Physics, Analysis and Geometry},
  volume  = {17},
  number  = {3--4},
  pages   = {441--464},
  year    = {2014},
}

@article{Al-GattasG26,
  author       = {Omar Al{-}Ghattas and
                  David Gamarnik},
  title        = {Circuit complexity lower bounds for quantum spin glasses},
  journal      = {CoRR},
  volume       = {abs/2607.14384},
  year         = {2026}
}

\appendix
\section{\texorpdfstring{Certifying Semirandom $k$-$\XOR$ Hamiltonians for odd $k$}{Certifying Semirandom XOR Hamiltonians odd}}
\label{sec:oddcase}

In this section, we extend the framework built in \Cref{sec:refutation} and prove \Cref{thm:refutation} for odd $k$. Our proof again follows the Kikuchi matrix method approach of \cite{GuruswamiKM22, HsiehKM23} by crafting a more complicated odd-arity Kikuchi matrix capturing Hamiltonian optimization. As is usual with this method, straightforward reductions to the even case seem to fail, so we require some additional instance preprocessing that we explain now.

\subsection{\texorpdfstring{Refuting bipartite Hamiltonian $k$-$\XOR$ instances}{Refuting bipartite Hamiltonian XOR instances}}

We begin the proof by defining a structured family of \textit{bipartite} Hamiltonian $k$-$\XOR$ instances.

\begin{definition}[$\calU$-bipartite Hamiltonian $k$-$\XOR$]
   Given $\calU$ a multiset over $\calP_t(n)$, we say a Hamiltonian $k$-$\XOR$ instance $\calI = (\calH, \{(P_C, b_C)\}_{C \in \calH})$ is a $t$-sparse $\calU$-bipartite instance if we can write $\calH = \bigsqcup_{U \in \calU} \calH_U$ and each $\calH_U \subseteq \calH$ with the property that $U \sqsubseteq P_C$ for all $C \in \calH_U$. We call the collection $\{\calH_U\}_{U \in \calU}$ the $\calU$-bipartite decomposition of the instance.
\end{definition}

\begin{definition}[$(\varepsilon, \ell)$-regularity in $\calU$-bipartite instances]
    Given a $\calU$-bipartite decomposition $\{\calH_U\}_{U \in \calU}$, a partition $\calH_U$ is $(\varepsilon, \ell)$-regular if there does not exist non-identity $W \in \calP(n)$ with $\abs{W} > \abs{U}$ and a subset $\calH' \subseteq \calH_U$ with the property $W \sqsubseteq P_C$ for all $C \in \calH'$ and $\abs{\calH'} > \max\left(\left(\frac{3n}{\ell}\right)^{k/2-\abs{W}}, 1\right) \cdot \varepsilon^{-2}$. The entire collection $\{\calH_U\}_{U \in \calU}$ is said to be $(\varepsilon, \ell)$-regular if all partitions $\calH_U$ are $(\varepsilon, \ell)$-regular. If $\varepsilon = 1$, we abbreviate to just $\ell$-regular.
\end{definition}

Given the structure of a bipartite Hamiltonian $k$-$\XOR$ instance $\calI$, the following Cauchy-Schwarz trick, adapted from the analogous classical CSP refutation trick, gives us a new way to bound the maximum eigenvalue of $\bH_\calI$.

\begin{lemma}[Cauchy-Schwarz Trick]
\label{lem:cauchyschwarz}
    For a Hamiltonian $k$-$\XOR$ instance $\calI = (\calH, \{(P_C, b_C)\}_{C \in \calH})$ we recall
    \begin{equation*}
        \lambda_{\mathrm{max}}(\bH_\calI) = \frac{1}{2} + \frac{1}{2 \abs{\calH}}\lambda_{\mathrm{max}}\left(\sum_{C \in \calH} b_CP_C\right) \coloneqq \frac{1}{2} + \frac{1}{2\abs{\calH}} \lambda_{\mathrm{max}}(\bH^*_\calI)\mper
    \end{equation*}
    Given a bipartite decomposition $\{\calH_U\}_{U \in \calU}$ of $\calH$ we have the bound
    \begin{equation*}
        \lambda_{\mathrm{max}}(\bH^*_\calI)^2 \leq \max_{\substack{\ket{\psi} \in (\C^2)^{\otimes n} \\ \ket{\psi} \text{ pure state}}} \abs{\calU} \sum_{U \in \calU} \sum_{C, C' \in \calH_U} b_C b_{C'} \bra{\psi} P_{\wt{C}} P_{\wt{C}'}\ket{\psi} \mcom
    \end{equation*}
    where $P_{\wt{C}} = P_C U$ for $C \in \calH_U$, $U \in \calU$.
\end{lemma}

\begin{proof}
    We start by writing $\lambda_{\mathrm{max}}(\bH^*_\calI)$ as the maximum quadratic form and partition $\bH^*_\calI$ according to $\calU$. 
    \begin{align*}
        \lambda_{\mathrm{max}}(\bH^*_\calI)^2 &\leq \max_{\substack{\ket{\psi} \in (\C^2)^{\otimes n} \\ \ket{\psi} \text{ pure state}}} \bra{\psi}\sum_{U \in \calU} U\sum_{C \in \calH_U} b_C P_{\wt{C}}\ket{\psi}^2\\
        &= \max_{\substack{\ket{\psi} \in (\C^2)^{\otimes n} \\ \ket{\psi} \text{ pure state}}} \left(\sum_{U \in \calU}\bra{\psi} U \sum_{C \in \calH_U} b_C P_{\wt{C}}\ket{\psi}\right)^2\\
        &\leq \max_{\substack{\ket{\psi} \in (\C^2)^{\otimes n} \\ \ket{\psi} \text{ pure state}}} \abs{\calU} \sum_{U \in \calU}\bra{\psi} U \sum_{C \in \calH_U} b_C P_{\wt{C}}\ket{\psi}^2\\
        &\leq \max_{\substack{\ket{\psi} \in (\C^2)^{\otimes n} \\ \ket{\psi} \text{ pure state}}} \abs{\calU} \sum_{U \in \calU} \sum_{C, C' \in \calH_U} b_C b_{C'} \bra{\psi} P_{\wt{C}} P_{\wt{C}'}\ket{\psi}\mper
    \end{align*}
    In the third and fourth lines we applied the Cauchy-Schwarz inequality.
\end{proof}

Given the definition of bipartite instances and a tool to bound their value, our natural goal is to design a way to find non-trivial bipartite decompositions for \textit{arbitrary} Hamiltonian $k$-$\XOR$ instances, which we accomplish through the following.

\begin{lemma}[Regularity decomposition algorithm for Hamiltonian $k$-$\XOR$]
    \label{lem:decompositionalg} 
    There is an algorithm that takes as input a Hamiltonian $k$-$\XOR$ instance $\calI = (\calH, \{(P_C, b_C)\}_{C \in \calH})$ and parameters $\ell, \varepsilon > 0$ assuming $\abs{\calH} \geq \Omega(n) \cdot \pbra{\frac{3n}{\ell}}^{k/2-1} \log n \cdot \varepsilon^{-2}$ and outputs a partition of $\calI$ into subinstances of the form $\calI^{(t)} = (\calH^{(t)}, \{(P_C, b_C)\}_{C \in \calH^{(t)}})$ with $t$-sparse $\calU^{(t)}$-bipartite decompositions $\{\calH^{(t)}_U\}_{U \in \calU^{(t)}}$ for each $\calH^{(t)}$ in time $n^{O(\ell)}$ with the guarantees:
    \begin{enumerate}
        \item \label{item:decompositionalg2} For $t \neq 1$ and all $U \in \calU^{(t)}$, $\abs{\calH^{(t)}_U} = \tau_t \coloneqq \lceil \max(1, \left(\frac{3n}{\ell}\right)^{k/2-t}) \cdot 4k^2\varepsilon^{-2} \rceil$.
        \item \label{item:decompositionalg3} For all $U \in \calU^{(1)}$, $\abs{\calH^{(1)}_U} \leq \tau_1$.
        \item \label{item:decompositionalg4} For all $t$, $\abs{\calU^{(t)}} \leq \frac{2\abs{\calH}}{\tau_t}$.
        \item \label{item:decompositionalg5} Each $\calH^{(t)}$ is $\left(\frac{\varepsilon}{2k}, \ell\right)$-regular.
    \end{enumerate}
\end{lemma}

We prove \Cref{lem:decompositionalg} in \Cref{sec:decompositionalg}. It turns out given an instance with a bipartite decomposition, the resulting polynomial can be refuted provided the decomposition satisfies the regularity property. Formally, we show the following.

\begin{lemma}[$\varepsilon^2$-refutation of semirandom bipartite Hamiltonian $k$-$\XOR$]
    \label{lem:mainrefutation}
    Fix $k \geq 2$, $1 \leq t \leq k$, and $k-t \leq \ell \leq n/2$. There is an algorithm that takes as input a Hamiltonian $k$-$\XOR$ instance $\mathbfcal{I}^{(t)} = (\calH^{(t)}, \{(P_C, \bb_C)\}_{C \in \calH^{(t)}})$ with a $\calU^{(t)}$-bipartite decomposition $\calH^{(t)} = \{\calH^{(t)}_U\}_{U \in \calU^{(t)}}$ a subset of $\calH$ describing an associated operator $\bU_t = \frac{k^2\abs{\calU^{(t)}}}{4\abs{\calH}^2} \sum_{U \in \calU^{(t)}} \sum_{C, C' \in \calH_U} \bb_C \bb_{C'} P_{\wt{C}} P_{\wt{C}'}$ and outputs a certificate $\algval(\bU_t) \in \R$ in time $n^{O(\ell)}$ with the guarantees:
    \begin{enumerate}
        \item \label{item:ref1} $\algval(\bU_t) \geq \lambdamax(\bU_t)$.
        \item \label{item:ref2} $\algval(\bU_t) < \varepsilon^2$ with high probability (over the randomness in $\bb_C$) given:
        \begin{enumerate}
            \item \label{item:ref2a} The hypothesis of \Cref{item:refutation2} in \Cref{thm:refutation} holds.
            \item \label{item:ref2b} The output guarantees of \Cref{lem:decompositionalg} hold.
        \end{enumerate}
    \end{enumerate}
\end{lemma}

With this, we have all the tools needed to build our refutation algorithm.

\begin{proof}[Proof of \Cref{thm:refutation} for odd $k$]
    To achieve an $\varepsilon$-refutation for an arbitrary odd-arity Hamiltonian $k$-$\XOR$ instance $\mathbfcal{I} = (\calH, \{(P_C, \bb_C)\}_{C \in \calH})$, we begin by applying \Cref{lem:decompositionalg} to achieve a partition of $\mathbfcal{I}$ into subinstances $\{\mathbfcal{I}^{(t)}\}_{t \in [k]}$ and associated $\calU^{(t)}$-bipartite decompositions. We can rewrite the maximum eigenvalue of $\bsH_{\mathbfcal{I}}$ as follows.
    \begin{align*}
        \lambda_{\mathrm{max}}(\bsH_{\mathbfcal{I}}) &= \frac{1}{2} + \frac{1}{2 \abs{\calH}}\lambda_{\mathrm{max}}\left(\sum_{C \in \calH} \bb_CP_C\right)\\
        &\leq \frac{1}{2} + \frac{1}{2 \abs{\calH}}\sum_{t \in [k]} \lambda_{\mathrm{max}}\left(\sum_{C \in \calH^{(t)}} \bb_CP_C\right)\\
        &\leq \frac{1}{2} + \sum_{t \in [k]} \frac{1}{2 \abs{\calH}}\lambda_{\mathrm{max}}\left(\sum_{C \in \calH^{(t)}} \bb_CP_C\right)\\
        &\coloneqq \frac{1}{2} + \sum_{t \in [k]} \varepsilon^{(t)}\mper
    \end{align*}
    Our goal is now to bound each term $\varepsilon^{(t)}$ in the sum by $\frac{\varepsilon}{k}$. Applying the Cauchy-Schwarz trick to $\sum_{C \in \calH^{(t)}} \bb_CP_C$ yields an immediate upper bound
    \begin{equation*}
        (k\varepsilon^{(t)})^2 \leq \abs{\lambdamax \left(\frac{k^2\abs{\calU^{(t)}}}{4\abs{\calH}^2} \sum_{U \in \calU^{(t)}} \sum_{C, C' \in \calH_U} \bb_C \bb_{C'} P_{\wt{C}} P_{\wt{C}'}\right)} \coloneqq \lambdamax(\bU_t)\mper
    \end{equation*}
    Since the output guarantees of \Cref{lem:decompositionalg} hold for our decomposition $\calU^{(t)}$, assuming the semirandom setting of \Cref{thm:refutation} allows us to invoke \Cref{lem:mainrefutation}, a Kikuchi matrix refutation of $\bU_t$ and conclude $(k\varepsilon^{(t)})^2 \leq \varepsilon^2$ and further $\lambdamax(\bsH_{\mathbfcal{I}}) \leq \frac{1}{2} + \varepsilon$ as desired.
\end{proof}

It suffices to prove the auxiliary lemmas, \Cref{lem:decompositionalg} and \Cref{lem:mainrefutation} to conclude this section.

\subsection{\texorpdfstring{Hamiltonian $k$-$\XOR$ regularity decomposition}{Hamiltonian XOR regularity decomposition}}
\label{sec:decompositionalg}

We state here the regularity decomposition algorithm of \Cref{lem:decompositionalg}.

\begin{tcolorbox}[
    width=\textwidth,   
    colframe=black,  
    colback=white,   
    title=Hamiltonian $k$-$\XOR$ Regularity Decomposition Algorithm,
    colbacktitle=white, 
    coltitle=black,      
    fonttitle=\bfseries,
    center title,   
    enhanced,       
    frame hidden,           
    borderline={1pt}{0pt}{black},
    sharp corners,
    toptitle=2.5mm
]
\textbf{Input:} A Hamiltonian $k$-$\XOR$ instance $\calI = (\calH, \{(P_C, b_C)\}_{C \in \calH})$.\\

\textbf{Output:} A set of instances $\cbra{\calI^{(t)}}_{t \in [k]}$ satisfying the criteria of \Cref{lem:decompositionalg}.\\

\textbf{Algorithm:}
\begin{enumerate}
    \item Let $t = k$, and while $\exists U \in \calP_t(n)$ such that $\abs{\{C \in \calH \mid U \sqsubseteq P_C\}} \geq \tau_t \coloneqq \max(1, (\frac{3n}{\ell})^{k/2-t}) \cdot 4k^2 \varepsilon^{-2}$, do the following. Otherwise, decrement $t$.
    \begin{enumerate}
        \item Let $\calH_U^{(t)}$ hold $C \in \calH$ for exactly $\tau_t$ such $P_C$ and move the set $\calH_U^{(t)}$ from $\calH$ to $\calH^{(t)}$.
    \end{enumerate}
    \item When $t = 0$, add all remaining $C \in \calH$ to $\calH^{(1)}_{(P_C)_1}$, where $(P_C)_1 \in \calP_1(n)$ is $P_C$ with all but its first non-trivial operator set to identity.
\end{enumerate}

\end{tcolorbox}

\begin{proof}[Proof of \Cref{lem:decompositionalg}]
    The $t$-sparsity of $\calU^{(t)}$ and \Cref{item:decompositionalg2} follow simply from the loop condition, which enforces that only $\tau_t$ sized sets with $\abs{U} = t$ get added to $\calH^{(t)}$. For \Cref{item:decompositionalg3}, note by the greediness of the algorithm any set added outside the for loop must have size less than $\tau_1$ otherwise it would have been added within.
    
    For \Cref{item:decompositionalg4}, assume $t > 1$ and note that by \Cref{item:decompositionalg2} we have $\abs{\calU^{(t)}} \leq \frac{\abs{\calH}}{\tau_t}$. When $t = 1$, we have two kinds of sets $\calH_U$, ones added in the for loop and those added outside. The number added in the for loop is $\leq \frac{\abs{\calH}}{\tau_1}$ following the case above. The number added outside is at most $\abs{\calP_1(n)} = 3n$. Since by assumption $\abs{\calH} \geq C n \log n \cdot \tau_1$ we have $3n \leq \frac{\abs{\calH}}{\tau_1}$ for $C \geq 1$, which gives in total $\abs{\calU^{(1)}} \leq \frac{2 \abs{\calH}}{\tau_1}$.

    \Cref{item:decompositionalg5} follows by the greediness of the algorithm. Assume for sake of contradiction there is some partition $\calH^{(t)}$ which is not $\left(\frac{\varepsilon}{2k}, \ell\right)$-regular. By definition, there is some $\calH' \subseteq \calH_U^{(t)}$ and $W \in \calP(n)$ such that all $V \in \calH'$ have $W \sqsubseteq V$ and moreover $\abs{W} = t' > t$ and $\abs{\calH'} \geq \tau_{t'}$. Since iteration $t'$ happens before $t$, such a set would have been available on iteration $t'$ and would have been added then instead, a contradiction.
\end{proof}

\subsection{Odd-arity Kikuchi matrices}

In this section, we prove our main technical component, \Cref{lem:mainrefutation}. The main component is an odd-arity Kikuchi matrix for the bipartite operators appearing as a result of applying the Cauchy-Schwarz trick, \Cref{lem:cauchyschwarz}. Our construction is inspired by the odd-arity constructions of \cite{GuruswamiKM22, HsiehKM23, KocurekM25}, but requires modifications inherent to the Hamiltonian setting.

Recall our goal is to bound the maximum eigenvalue of an operator $\bU_t \propto \sum_{U \in \calU^{(t)}} \sum_{C, C' \in \calH_U} b_C b_{C'} P_{\wt{C}} P_{\wt{C}'}$ built from a Hamiltonian $k$-$\XOR$ instance $\calI^{(t)} = (\calH^{(t)}, \{(P_C, b_C)\}_{C \in \calH^{(t)}})$ with a $\calU^{(t)}$-bipartite decomposition $\calH^{(t)} = \{\calH^{(t)}_U\}_{U \in \calU^{(t)}}$. Towards this, we define the following odd-arity Kikuchi matrix, meant to capture the optimization of $\bU_t$ over quantum states.

\begin{definition}[Odd-arity Kikuchi matrix]
    \label{def:kikuchimatrixodd}
    Let $k-t \leq \ell \leq n/2$ be a parameter and let $N = 3^\ell \binom{2n}{\ell}$. For each pair $(P, P') \in \calP_{k-t}(n)^{\otimes 2}$, we define a matrix $A_{P, P'} \in \R^{N \times N}$ as follows. Identifying $N$ with $\calQ_\ell(n)$ the set of pairs $(Q^{(1)}, Q^{(2)}) \in \calP(n)^{\otimes 2}$ with $\abs{Q^{(1)}} + \abs{Q^{(2)}} = \ell$, we define the entry $(Q,R) \in \calQ_\ell(n)^{\otimes 2}$ by
        \begin{equation*}
            A_{P, P'}(Q, R) = \begin{cases}
                          -1  & Q\xrightarrow{\text{$P, P'$}} R \text{ and } (Q^{(2)})^\dagger (Q^{(1)})^\dagger R^{(1)} R^{(2)} = (R^{(2)})^\dagger (R^{(1)})^\dagger Q^{(1)} Q^{(2)} = -PP'\\
                          1 & Q\xrightarrow{\text{$P, P'$}} R\\
                          0 & \text{otherwise}
                        \end{cases}
        \end{equation*}
        where we say $Q \xrightarrow{\text{$P, P'$}} R$ if the following conditions hold
        \begin{enumerate}
            \item $(Q^{(1)})^\dagger R^{(1)} = P$ and $(Q^{(2)})^\dagger R^{(2)} = P'$.
            \item $\abs{\supp(Q^{(1)}) \cap \supp(P)} = \lfloor\frac{k-t}{2}\rfloor$ and $\abs{\supp(R^{(1)}) \cap \supp(P)} = \lceil\frac{k-t}{2}\rceil$ and $\abs{\supp(Q^{(2)}) \cap \supp(P')} = \lceil\frac{k-t}{2}\rceil$ and $\abs{\supp(R^{(2)}) \cap \supp(P')} = \lfloor\frac{k-t}{2}\rfloor$ or vice versa.
        \end{enumerate}
        
    Define $\bU_t \propto \sum_{U \in \calU^{(t)}} \sum_{C, C' \in \calH_U} b_C b_{C'} P_{\wt{C}} P_{\wt{C}'}$ via Hamiltonian $k$-$\XOR$ instance $\calI^{(t)} = (\calH^{(t)}, \{(P_C, b_C)\}_{C \in \calH^{(t)}})$ with a $\calU^{(t)}$-bipartite decomposition $\calH^{(t)} = \{\calH^{(t)}_U\}_{U \in \calU^{(t)}}$. We define the level-$\ell$ Kikuchi matrix for this instance to be $K_{\bU_t} = \sum_{U \in \calU^{(t)}}  \sum_{C \neq C' \in \calH^{(t)}_U} b_C b_{C'} A_{P_{\wt{C}}, P_{\wt{C}'}} \otimes \Id_{2^n}$ where again $P_{\wt{C}} = U^\dagger P_C$ for $C \in \calH^{(t)}_U$. We shorthand $A_{P_{\wt{C}}, P_{\wt{C}'}} \coloneqq A_{C, C'}^U$ and refer to the graph of $A_{\bU_t} = \sum_{U \in \calU^{(t)}}  \sum_{C \neq C' \in \calH^{(t)}_U} \abs{A_{C, C'}^U}$ the underlying adjacency matrix as the Kikuchi graph and $A_{\bU_t}^* = \sum_{U \in \calU^{(t)}}  \sum_{C \neq C' \in \calH^{(t)}_U} b_C b_{C'} A_{C, C'}^U$ the signed version.
\end{definition}

While this definition is significantly more complicated than the even-arity case, the blueprint is still largely similar to the odd-arity construction for classical $k$-$\XOR$ as seen in \cite{GuruswamiKM22, HsiehKM23}. The main complication in the quantum setting is that we must introduce signs into $A_{P, P'}(Q,R)$ based jointly on the commutation relations of $P, P'$ and also the parts of $(Q,R)$. This sign is introduced exactly to cancel the effects of anticommutation between $P$ and $(Q,R)$ and is an artifact of our construction rather than something inherent within the underlying ensemble of Paulis, but it will fortunately be taken out with a simple triangle inequality later in the proof.

We now show that this Kikuchi matrix is Hermitian so that we may later apply the trace moment method.

\begin{observation}
    Let $\bU_t \propto \sum_{U \in \calU^{(t)}} \sum_{C, C' \in \calH_U} b_C b_{C'} P_{\wt{C}} P_{\wt{C}'}$ be built from a Hamiltonian $k$-$\XOR$ instance $\calI^{(t)} = (\calH^{(t)}, \{(P_C, b_C)\}_{C \in \calH^{(t)}})$ with a $\calU^{(t)}$-bipartite decomposition $\calH^{(t)} = \{\calH^{(t)}_U\}_{U \in \calU^{(t)}}$. Then the corresponding level-$\ell$ Kikuchi matrix $K_{U_t}$ is Hermitian for any $k/2 \leq \ell \leq n/2$.
\end{observation}

\begin{proof}
    Note it suffices to prove Hermiticity for the adjacency matrix $A_{P, P'}$ for Paulis $P, P' \in \calP_{k-t}(n)$ by linearity. This reduces to showing for arbitrary $(Q,R) \in \calQ_\ell(n)^{\otimes 2}$ that $Q \xrightarrow{\text{$P, P'$}} R$ implies $R \xrightarrow{\text{$P, P'$}} Q$, since clearly the criteria governing whether $A_{P, P'}(Q,R) = \pm1$ is symmetric in $Q$ and $R$. Note however $Q \xrightarrow{\text{$P, P'$}} R$ immediately guarantees the support condition in (2), and since (1) and (2) together imply $Q^{(1)}$ and $R^{(1)}$ have disjoint support on $\supp(P)$, commutativity is immediate to yield $(R^{(1)})^\dagger Q^{(1) = P}$ (and analogously for the second part of (1)).
\end{proof}

With the Kikuchi matrix well-defined, the key observation is that this matrix captures the energy of the Cauchy-Schwarz trick operator $\bU_t$ in its bilinear forms.

\begin{observation}
    \label{obs:kikuchirelaxationodd}
    Let $\bU_t = \frac{k^2\abs{\calU^{(t)}}}{4\abs{\calH}^2} \sum_{U \in \calU^{(t)}} \sum_{C, C' \in \calH_U} b_C b_{C'} P_{\wt{C}} P_{\wt{C}'}$ for some Hamiltonian $k$-$\XOR$ instance $\calI^{(t)} = (\calH^{(t)}, \{(P_C, b_C)\}_{C \in \calH^{(t)}})$ with a $\calU^{(t)}$-bipartite decomposition $\calH^{(t)} = \{\calH^{(t)}_U\}_{U \in \calU^{(t)}}$. Let $\ket{\psi} \in (\C^2)^{\otimes n}$ be a pure state. Consider $\psi^{\odot \ell}$ in $((\C^2)^{\otimes n})^N$ defined by $\psi^{\odot \ell}_{(Q, Q')} = Q^\dagger Q'\ket{\psi}$ for $(Q, Q') \in \calQ_\ell(n)$. Then
    \begin{equation*}
        \bra{\psi}\bU_t\ket{\psi} = \frac{k^2\abs{\calU^{(t)}}}{4\abs{\calH}^2} \sum_{C \in \calH^{(t)}} b_C^2 + \frac{k^2\abs{\calU^{(t)}}}{4\Delta^{(t)}\abs{\calH}^2} (\psi^{\odot \ell})^\dagger K_{\bU_t} \psi^{\odot \ell}\mcom
    \end{equation*}
    where $\Delta^{(t)} \coloneqq \binom{k-t}{\lceil \frac{k-t}{2} \rceil}\binom{k-t}{\lfloor \frac{k-t}{2} \rfloor}\binom{2n-2(k-t)}{\ell-k+t} \cdot 3^{\ell - k+t} \cdot 2^{\mathbf{1}\{\text{$k-t$ odd}\}}$
\end{observation}

\begin{proof}
    For any Kikuchi matrix we compute
    \begin{align*}
        (\psi^{\odot \ell})^\dagger K_{\bU_t} \psi^{\odot \ell}
        &= \tr\left(\left(\sum_{U \in \calU^{(t)}}  \sum_{C \neq C' \in \calH^{(t)}_U} b_C b_{C'} A_{C, C'}^U \otimes \Id_{2^n}\right) \cdot \psi^{\odot \ell} (\psi^{\odot \ell})^\dagger\right) \\
        &= \sum_{(Q, R) \in \calQ_\ell(n)^{\otimes 2}} \sum_{U \in \calU^{(t)}} \sum_{C \neq C' \in \calH^{(t)}_U} b_C b_{C'} \cdot \tr(((A_{C, C'}^U)_{Q, R} \cdot \Id_{2^n}) \cdot \psi_Q^{\odot \ell} (\psi_R^{\odot \ell})^\dagger)\\
        &= \sum_{(Q, R) \in \calQ_\ell(n)^{\otimes 2}} \sum_{U \in \calU^{(t)}} \sum_{C \neq C' \in \calH^{(t)}_U}  b_C b_{C'} \cdot (A_{C,C'}^U)_{Q,R} \cdot \tr((Q^{(1)})^\dagger Q^{(2)} \ketbra{\psi}{\psi} (R^{(2)})^\dagger R^{(1)})\\
        &= \sum_{U \in \calU^{(t)}}\sum_{C \neq C' \in \calH^{(t)}_U} \sum_{(Q, R) \in \calQ_\ell(n)^{\otimes 2}}  \mathbf{1}\{Q \xrightarrow{\text{$P_{\wt{C}}$, $P_{\wt{C}'}$}} R\} \cdot  b_C b_{C'} \bra{\psi} P_{\wt{C}} P_{\wt{C}'}\ket{\psi}\\
        &= \sum_{U \in \calU^{(t)}}\sum_{C \neq C' \in \calH^{(t)}_U} \sum_{(Q, R) \in \calQ_\ell(n)^{\otimes 2}}  \mathbf{1}\{Q \xrightarrow{\text{$P_{\wt{C}}$, $P_{\wt{C}'}$}} R\} \cdot  b_C b_{C'} \bra{\psi} P_C P_{C'}\ket{\psi}\mper
    \end{align*}

    We point out in the third line we use trace cyclicity and that $(A^U_{C, C'})_{Q,R}$ cancels the sign associated with $(Q^{(2)})^\dagger (Q^{(1)})^\dagger R^{(1)} R^{(2)}$ by construction. Note also that if $P_C$ and $P_{C'}$ anticommute, their contribution to both $\boldsymbol{U}_t$ and this sum is $0$, since for any pair $(Q, R)$ with $Q \xrightarrow{\text{$P_{\wt{C}}$, $P_{\wt{C}'}$}} R$ contributing $P_C P_{C'}$ we have $R \xrightarrow{\text{$P_{\wt{C}}$, $P_{\wt{C}'}$}} Q$ contributing $-P_C P_{C'}$. Our primary motivation for including these terms in the definition of the Kikuchi matrix was simply to make the degree of the Kikuchi graph independent of the number of commuting terms for later calculations.

    Now we claim the coefficient $\sum_{(Q,R) \in \calQ_\ell(n)^{\otimes 2}} \mathbf{1}\{Q \xrightarrow{\text{$P_{\wt{C}}$, $P_{\wt{C}'}$}} R\} = \binom{k-t}{\lceil \frac{k-t}{2} \rceil} \binom{k-t}{\lfloor \frac{k-t}{2} \rfloor} \binom{2n-2(k-t)}{\ell-k+t} \cdot 3^{\ell - k+t} \cdot 2^{\mathbf{1}\{k-t \text{ odd}\}} \coloneqq \Delta^{(t)}$. Plugging this directly into the equation yields the result as follows.
    \begin{align*}
        (\psi^{\odot \ell})^\dagger K_{\bU_t} \psi^{\odot \ell} &= \Delta^{(t)} \sum_{U \in \calU^{(t)}}\sum_{C \neq C' \in \calH^{(t)}_U} b_C b_{C'} \bra{\psi} P_C P_{C'}\ket{\psi}\mper
    \end{align*}
    Comparing this to
    \begin{align*}
        \bra{\psi} \bU_t \ket{\psi} = \frac{k^2\abs{\calU^{(t)}}}{4 \abs{\calH}^2} \sum_{U \in \calU^{(t)}}\sum_{C, C' \in \calH^{(t)}_U} b_C b_{C'} \bra{\psi} P_C P_{C'}\ket{\psi}\mcom
    \end{align*}
    we see they match up to a factor $\frac{k^2 \abs{\calU^{(t)}}}{4 \Delta^{(t)} \abs{\calH}^2}$ and a diagonal term when $C = C'$. This latter is exactly the constant term in the statement of the theorem when simplifying $P_C^2 = \Id$ by self-inverse.
    
    To justify the value of $\Delta^{(t)}$, we fix an arbitrary pair $P_{\wt{C}}, P_{\wt{C}'}$ and show how to count. We first specify the supports of $Q^{(1)}$, $Q^{(2)}$, $R^{(1)}$, and $R^{(2)}$ and then the non-zero elements $Q^{(1)}_i$ for $i \in \supp(Q^{(1)})$ (and the same for $Q^{(2)}$, $R^{(1)}$, and $R^{(2)}$). The conditions $\abs{\supp(Q^{(1)}) \cap \supp(P_{\wt{C}})} = \lfloor\frac{k-t}{2}\rfloor$ along with $\abs{\supp(Q^{(2)}) \cap \supp(P_{\wt{C}'})} = \lceil \frac{k-t}{2} \rceil$ or the other way around gives $2$ options when $k-t$ is odd (otherwise the conditions are the same). Without loss of generality we let $\abs{\supp(Q^{(1)}) \cap \supp(P_{\wt{C}})} = \lfloor\frac{k-t}{2}\rfloor$, then there are $\binom{k-t}{\lfloor \frac{k-t}{2} \rfloor}$ choices for $\supp(Q^{(1)}) \cap \supp(P_{\wt{C}})$. Similarly we have $\binom{k-t}{\lceil \frac{k-t}{2} \rceil}$ ways to pick the intersection $\supp(Q^{(2)}) \cap \supp(P_{\wt{C}'})$. Among $\supp(Q^{(1)}) \setminus \supp(P_{\wt{C}})$ and $\supp(Q^{(2)}) \setminus \supp(P_{\wt{C}'})$, there are then $\ell - k+t$ indices, and each falls inside of $[n] \setminus \supp(P_{\wt{C}})$ and $[n] \setminus \supp(P_{\wt{C}'})$, giving $\binom{2n-2(k-t)}{\ell - k+t}$ options. The condition $(Q^{(1)})^\dagger R^{(1)} = P_{\wt{C}}$ enforces that $\supp(Q^{(1)}) \oplus \supp(R^{(1)}) = \supp(P_{\wt{C}})$ which requires $\supp(R^{(1)}) \cap \supp(P_{\wt{C}}) = \supp(P_{\wt{C}}) \setminus (\supp(Q^{(1)}) \cap \supp(P_{\wt{C}}))$ and $\supp(R^{(1)}) \setminus \supp(P_{\wt{C}}) = \supp(Q^{(1)}) \setminus \supp(P_{\wt{C}})$, fully determining $\supp(R^{(1)})$. Likewise $\supp(R^{(2)})$ is fully determined. Now we specify the actual values. Note for $i \in \supp(Q^{(1)}) \cap \supp(P_{\wt{C}})$ we have $Q^{(1)}_i = (P_{\wt{C}})_i$, and likewise for the other intersections. For the remaining $\ell - k+t$ non-zero entries across $Q^{(1)}$ and $Q^{(2)}$, we may choose any non-trivial Paulis $X, Y, Z$. The corresponding entries in $R^{(1)}$ and $R^{(2)}$ are then determined as they must cancel, totaling to $3^{\ell - k+t}$ choices. Combining these counts yields exactly $\Delta^{(t)}$ above.
\end{proof}

As before, we will define and work with a degree-regularized Kikuchi matrix (which is defined analogously to \Cref{def:regularkikuchimatrix}) and utilize a few basic facts about the degree in the Kikuchi graph. Following the even-arity case, we provide a lower bound to the average degree in the Kikuchi graph.

\begin{observation}
    \label{fact:avgdegreeodd}
    Let $\bU_t = \frac{k^2\abs{\calU^{(t)}}}{4\abs{\calH}^2} \sum_{U \in \calU^{(t)}} \sum_{C, C' \in \calH_U} b_C b_{C'} P_{\wt{C}} P_{\wt{C}'}$ for some Hamiltonian $k$-$\XOR$ instance $\calI^{(t)} = (\calH^{(t)}, \{(P_C, b_C)\}_{C \in \calH^{(t)}})$ with a $\calU^{(t)}$-bipartite decomposition $\calH^{(t)} = \{\calH^{(t)}_U\}_{U \in \calU^{(t)}}$ and denote by $K_{\bU_t}$ the associated Kikuchi matrix. For a vertex $(Q^{(1)}, Q^{(2)}) \in \calQ_\ell(n)$ we define the weighted graph degree according to the Kikuchi graph $A_{\bU_t}$. Then $d^{(t)} \coloneqq \E_{(\boldsymbol{Q}^{(1)}, \boldsymbol{Q}^{(2)}) \sim \calQ_\ell(n)}[\deg(\boldsymbol{Q}^{(1)}, \boldsymbol{Q}^{(2)})]$ satisfies $d^{(t)} \geq \frac{1}{2}\left({\frac{\ell}{6n}}\right)^{k-t} \cdot \sum_{U \in \calU^{(t)}} \binom{\abs{\calH^{(t)}_U}}{2}$.
\end{observation}

\begin{proof}
    Note that every distinct choice $(C, C') \in \calH_U^{(t)}$ uniformly adds $\Delta^{(t)} = \binom{k-t}{\lceil \frac{k-t}{2} \rceil}\binom{k-t}{\lfloor \frac{k-t}{2} \rfloor}\binom{2n-2(k-t)}{\ell-k+t} \cdot 3^{\ell - k+t} \cdot 2^{\mathbf{1}\{\text{$k-t$ odd}\}}$ edges. The total degree can then be written as $2\Delta^{(t)} \sum_{U \in \calU^{(t)}} \binom{\abs{\calH^{(t)}_U}}{2}$ and the average using standard binomial approximations is:
    \begin{flalign*}
        &d^{(t)} = \frac{2\Delta^{(t)}}{N} \sum_{U \in \calU^{(t)}} \binom{\abs{\calH^{(t)}_U}}{2}
        = \frac{2\binom{k-t}{\lceil \frac{k-t}{2} \rceil}\binom{k-t}{\lfloor \frac{k-t}{2} \rfloor}\binom{2n-2(k-t)}{\ell-k+t} \cdot 3^{\ell - k+t} \cdot 2^{\mathbf{1}\{\text{$k-t$ odd}\}}}{3^\ell \binom{2n}{\ell}} \sum_{U \in \calU^{(t)}} \binom{\abs{\calH^{(t)}_U}}{2}\\
        &\;\;\;\;\geq \frac{1}{2}\left({\frac{\ell}{6n}}\right)^{k-t}  \sum_{U \in \calU^{(t)}} \binom{\abs{\calH^{(t)}_U}}{2}\mper
    \end{flalign*}
    The inequality follows from \Cref{fact:binomest}.
\end{proof}

There is one other property important to track, which is the local degree of each vertex, which is roughly for any operator $P_C$ and vertex, the number of $P_{C'}$ for which an edge typed as $(C, C')$ is incident.

\begin{definition}[Local degree]
    Let $Q = (Q^{(1)}, Q^{(2)}) \in \calQ_\ell(n)$ be a vertex of the level-$\ell$ Kikuchi graph and $C \in \calH^{(t)}$. We define the $C$-local degrees of a vertex as
    \begin{align*}
        d_{Q, C, 0} \coloneqq \abs{\{C' \neq C \in \calH^{(t)}_U \mid \exists R \in \calQ_\ell(n), Q \xrightarrow{\text{$P_{\wt{C}}, P_{\wt{C}'}$}} R\}}\mcom
    \end{align*}
    and
    \begin{align*}
        d_{Q, C, 1} \coloneqq \abs{\{C' \neq C \in \calH^{(t)}_U \mid \exists R \in \calQ_\ell(n), Q \xrightarrow{\text{$P_{\wt{C}'}, P_{\wt{C}}$}} R\}}\mcom
    \end{align*}
    We say a Kikuchi graph has $\eta$-bounded local degree if $d_{Q, C, b} \leq \eta$ for all $Q \in \calQ_\ell(n)$, $C \in \calH^{(t)}$, and $b \in \{0,1\}$.
\end{definition}

With our Kikuchi matrix properly defined, we are ready to prove our main spectral refutation lemma.

\begin{lemma}
    \label{lem:spectralnormodd}
    Assume $d^{(t)}>0$. Let $\bU_t = \frac{k^2\abs{\calU^{(t)}}}{4\abs{\calH}^2} \sum_{U \in \calU^{(t)}} \sum_{C, C' \in \calH_U} \bb_C \bb_{C'} P_{\wt{C}} P_{\wt{C}'}$ for some Hamiltonian $k$-$\XOR$ instance $\mathbfcal{I}^{(t)} = (\calH^{(t)}, \{(P_C, \bb_C)\}_{C \in \calH^{(t)}})$ with a $\calU^{(t)}$-bipartite decomposition $\calH^{(t)} = \{\calH^{(t)}_U\}_{U \in \calU^{(t)}}$. Let $\tilde{K}_{\bU_t} = \Gamma_t^{-1/2}A^*_{\bU_t}\Gamma_t^{-1/2} \otimes \Id_{2^n}$ be the level-$\ell$ degree-regularized Kikuchi matrix, as defined in \Cref{def:kikuchimatrixodd}. Suppose additionally the interaction coefficients $\{\bb_C\}_{C \in \calH}$ are drawn independently and uniformly from $\{\pm 1\}$ or they are drawn as independent standard Gaussians and assume further the Kikuchi graph has $\eta_t$-bounded local degree. Then, with probability $\geq 1 - \exp(-\Theta(\sqrt{\log n}))$, it holds that
    \begin{equation*}
        \norm{\Gamma_t^{-1/2} A^*_{\bU_t} \Gamma_t^{-1/2}}{2} \leq C_0\pbra{\sqrt{\frac{\eta_t\ell\log n}{d^{(t)}}} +  \frac{\eta_t\ell\log n}{d^{(t)}}}\mper
    \end{equation*}
    Here $C_0>0$ is a universal constant.
\end{lemma}

From \Cref{lem:spectralnormodd}, we are able to prove \Cref{lem:mainrefutation}.

\begin{proof}[Proof of \Cref{lem:mainrefutation} from \Cref{lem:spectralnormodd}]
    Write $B_t \coloneqq \sum_{U\in\calU^{(t)}}\binom{\abs{\calH_U^{(t)}}}{2}$.
    If $d^{(t)}=0$, then $B_t=0$ by \Cref{fact:avgdegreeodd}, so every bucket has at most one clause and
    \begin{equation*}
        \bU_t=\frac{k^2\abs{\calU^{(t)}}}{4\abs{\calH}^2}
        \sum_{C\in\calH^{(t)}}\bb_C^2\Id\mper
    \end{equation*}
    In this case, output this scalar coefficient and omit the spectral term. For the remainder of the proof assume $d^{(t)}>0$.

    We thus begin with \Cref{obs:kikuchirelaxationodd}, which allows us to claim
    \begin{equation*}
        \max_{\substack{\ket{\psi} \in (\C^2)^{\otimes n} \\ \ket{\psi} \text{ pure state}}} \bra \psi \bU_t \ket \psi = \max_{\substack{\ket{\psi} \in (\C^2)^{\otimes n} \\ \ket{\psi} \text{ pure state}}} \frac{k^2\abs{\calU^{(t)}}}{4\abs{\calH}^2} \sum_{C \in \calH^{(t)}} \bb_C^2 + \frac{k^2\abs{\calU^{(t)}}}{4\Delta^{(t)}\abs{\calH}^2} (\psi^{\odot \ell})^\dagger K_{\bU_t} \psi^{\odot \ell}
    \end{equation*}
    This allows us to bound the value in terms of the spectral norm of the degree-regularized adjacency matrix $\tilde{A}_{\bU_t} = \Gamma_t^{-1/2} A_{\bU_t}^* \Gamma^{-1/2}$ of the Kikuchi graph. For any $\ket \psi \in (\C^2)^{\otimes n}$ a quantum state we have
    \begin{align*}
        (\psi^{\odot \ell})^\dagger K_{\bU_t} \psi^{\odot \ell}
        &= (\psi^{\odot \ell})^\dagger (\Gamma_t^{1/2} \otimes \Id_{2^n})(\tilde{A}_{\bU_t} \otimes \Id_{2^n}) (\Gamma_t^{1/2} \otimes \Id_{2^n})\psi^{\odot \ell}\\
        &\leq \lVert \tilde{A}_{\bU_t} \otimes \Id_{2^n}\rVert_{2} \cdot \norm{(\Gamma_t^{1/2} \otimes \Id_{2^n}) \psi^{\odot \ell}}{2}^2\\
        &= \lVert\tilde{A}_{\bU_t}\rVert_{2} \cdot \tr(\Gamma_t)\mper
    \end{align*}
    We use here that
    \begin{equation*}
        \norm{(\Gamma_t^{1/2} \otimes \Id_{2^n}) \psi^{\odot \ell}}{2}^2 = (\psi^{\odot \ell})^\dagger(\Gamma_t \otimes \Id_{2^n})\psi^{\odot \ell} = \sum_{Q \in \mathcal{Q}_\ell} (\psi^{\odot \ell}_Q)^\dagger ((\Gamma_t)_{Q, Q} \Id_{2^n})\psi_Q^{\odot \ell} = \tr(\Gamma_t)\mper
    \end{equation*}
    Observing $\tr(\Gamma_t) = 4\Delta^{(t)} \sum_{U \in \calU^{(t)}} \binom{\abs{\calH_U^{(t)}}}{2}$, twice the total degree, we bound our original equation as
    \begin{equation*}
        \max_{\substack{\ket{\psi} \in (\C^2)^{\otimes n} \\ \ket{\psi} \text{ pure state}}} \bra \psi \bU_t \ket \psi \leq \frac{k^2\abs{\calU^{(t)}}}{4\abs{\calH}^2} \sum_{C \in \calH^{(t)}} \bb_C^2 + \frac{k^2\abs{\calU^{(t)}}}{\abs{\calH}^2} \sum_{U \in \calU^{(t)}} \binom{\abs{\calH_U^{(t)}}}{2} \cdot \lVert\tilde{A}_{\bU_t}\rVert_{2} \mper
    \end{equation*}
    To fulfill \Cref{item:ref1} of \Cref{lem:mainrefutation}, we can simply compute and output $\algval^*(\bU_t) \coloneqq \frac{k^2\abs{\calU^{(t)}}}{4\abs{\calH}^2} \sum_{C \in \calH^{(t)}} \bb_C^2 + \frac{k^2\abs{\calU^{(t)}}}{\abs{\calH}^2} \sum_{U \in \calU^{(t)}} \binom{\abs{\calH_U^{(t)}}}{2} \cdot \lVert\tilde{A}_{\bU_t}\rVert_{2}$ which is possible in time $n^{O(\ell)}$ with the bulk of the computation coming from computing the spectral norm $\lVert\tilde{A}_{\bU_t}\rVert_{2}$, an $N = 3^\ell \binom{2n}{\ell}$ dimensional matrix. We spend the rest of the proof justifying \Cref{item:ref2} of \Cref{lem:mainrefutation}.

    We start by showing the constant term $\frac{k^2\abs{\calU^{(t)}}}{4\abs{\calH}^2} \sum_{C \in \calH^{(t)}} \bb_C^2$ is less than $\frac{\varepsilon^2}{2}$. A first observation is when the $\bb_C$'s are $\pm 1$, the sum simplifies to just $\abs{\calH^{(t)}}$. By \Cref{item:decompositionalg4} of \Cref{lem:decompositionalg}, we have $\abs{\calU^{(t)}} \leq \frac{2\abs{\calH}}{\tau_t}$, which immediately implies the term $\frac{k^2\abs{\calH^{(t)}}\abs{\calU^{(t)}}}{4\abs{\calH}^2} \leq \frac{k^2}{2\tau_t} \leq \frac{\varepsilon^2}{8}$. If the $\bb_C$'s are instead Gaussian, we do not have cancellation, but we do have the same expectation which can be combined with the following concentration bound.
    \begin{lemma}[Laurent-Massart inequality, Lemma 1 \cite{LaurentM00}]
        Let $\boldsymbol{X}_1, ..., \boldsymbol{X}_m$ be a set of i.i.d. standard Gaussians. Then $\boldsymbol{X} =\sum_{i=1}^m \boldsymbol{X}_i^2$ satisfies
        \begin{equation*}
            \Pr\sbra{\abs{\boldsymbol{X}} \geq m + 2\sqrt{mx} + 2x} \leq 2\exp(-x)\mper
        \end{equation*}
    \end{lemma}
    Setting $2x = m = \abs{\calH^{(t)}} \geq \sqrt{n}$ in the lemma tells us this term does not exceed its expectation more than fourfold with high probability, which gives a bound $\frac{\varepsilon^2}{2}$ from above. If $\abs{\calH^{(t)}}$ is less than $\sqrt{n}$, we can simply union bound over every coefficient $\abs{\bb_C}$ from $C \in \calH^{(t)}$ being at most $\sqrt{\log n}$ and safely ignore since $\abs{\calH} \geq \Omega(n \log n) \cdot \varepsilon^{-4}$.
    
    Our goal now becomes bounding the spectral norm term. Let $\eta_t$ be the maximum local degree of the unsigned Kikuchi graph. Applying \Cref{lem:spectralnormodd}, we split this contribution as
    \begin{align*}
        \frac{k^2\abs{\calU^{(t)}}}{\abs{\calH}^2}B_t
        \lVert\tilde{A}_{\bU_t}\rVert_2
        &\leq S_t+L_t\mcom
    \end{align*}
    where
    \begin{align*}
        S_t&\coloneqq \frac{C_0k^2\abs{\calU^{(t)}}}{\abs{\calH}^2}B_t
        \sqrt{\frac{\eta_t\ell\log n}{d^{(t)}}}\mcom\\
        L_t&\coloneqq \frac{C_0k^2\abs{\calU^{(t)}}}{\abs{\calH}^2}B_t
        \frac{\eta_t\ell\log n}{d^{(t)}}\mper
    \end{align*}
    For the square-root contribution, \Cref{fact:avgdegreeodd} gives
    \begin{align*}
        S_t 
        &\leq \frac{C_0k^2\abs{\calU^{(t)}}}{\abs{\calH}^2} \sum_{U \in \calU^{(t)}} \binom{\abs{\calH_U^{(t)}}}{2} \cdot \sqrt{\frac{\eta_t \ell\log n}{d^{(t)}}}\\
        &\leq \frac{2C_0k^2\abs{\calU^{(t)}}}{\abs{\calH}^2} \sum_{U \in \calU^{(t)}} \binom{\abs{\calH_U^{(t)}}}{2} \cdot \sqrt{\frac{\eta_t \ell\log n}{\left({\frac{\ell}{6n}}\right)^{k-t}  \sum_{U \in \calU^{(t)}} \binom{\abs{\calH^{(t)}_U}}{2}}}\\
        &\leq \frac{2C_0k^2\abs{\calU^{(t)}}}{\abs{\calH}^2} \cdot \sqrt{\eta_t \ell\log n \left({\frac{6n}{\ell}}\right)^{k-t}  \sum_{U \in \calU^{(t)}} \binom{\abs{\calH^{(t)}_U}}{2}}\\
        &\leq \frac{2C_0k^2\abs{\calU^{(t)}}}{\abs{\calH}^2} \cdot \sqrt{\eta_t \ell\log n \left({\frac{6n}{\ell}}\right)^{k-t}  \abs{\calU^{(t)}} \tau_t^2}\\
        &\leq 8C_0k^2 \sqrt{\frac{\eta_t \ell\log n}{\abs{\calH} \tau_t} \left({\frac{6n}{\ell}}\right)^{k-t}}\mper
    \end{align*}
    In the second to last line, we invoke \Cref{item:decompositionalg2} and \Cref{item:decompositionalg3} of \Cref{lem:decompositionalg} to claim $ \binom{\abs{\calH^{(t)}_U}}{2} \leq \tau_t^2$. In the last line, taking $\frac{\abs{\calU^{(t)}}}{\abs{\calH}^2}$ inside the root yields $\frac{\abs{\calU^{(t)}}^3 \tau_t^2}{\abs{\calH}^4}$, which by \Cref{item:decompositionalg4} of \Cref{lem:decompositionalg} we can bound by $\frac{8}{\abs{\calH} \tau_t}$. Taking $\abs{\calH} \geq O(1)^k \cdot n \log n \left(\frac{n}{\ell}\right)^{k/2-1} \cdot \varepsilon^{-4}$ for a large enough universal constant as in the statement of the theorem, we cancel almost all terms to yield a bound $S_t\leq \frac{\varepsilon^2}{4}\sqrt{\frac{\eta_t}{\tau_t}\left(\frac{n}{\ell}\right)^{k/2-t}}$.

    The linear contribution is controlled by the same degree estimate:
    \begin{align*}
        L_t
        &\leq \frac{2C_0k^2\abs{\calU^{(t)}}}{\abs{\calH}^2}
        \eta_t\ell\log n\left(\frac{6n}{\ell}\right)^{k-t}\\
        &\leq \frac{4C_0k^2}{\abs{\calH}}\frac{\eta_t}{\tau_t}
        \ell\log n\left(\frac{6n}{\ell}\right)^{k-t}\\
        &\leq \frac{\varepsilon^4}{4}\frac{\eta_t}{\tau_t}\left(\frac{n}{\ell}\right)^{k/2-t}\mper
    \end{align*}
    The second line uses \Cref{item:decompositionalg4} of \Cref{lem:decompositionalg}, and the last follows by increasing the universal constant in the same density hypothesis. Thus it suffices to bound $\frac{\eta_t}{\tau_t}\left(\frac{n}{\ell}\right)^{k/2-t}$, with no lower bound on $d^{(t)}$ other than positivity.

    It is here we finally discuss $\eta_t$-boundedness. In the statement of \Cref{lem:mainrefutation}, we make no assumption on $\eta_t$, leaving us to take the trivial bound $\eta_t \leq \abs{\calH_U^{(t)}} \leq \tau_t$, which is sufficient to show $\sqrt{\frac{\eta_t}{\tau_t} \left(\frac{n}{\ell}\right)^{k/2-t}} \leq 1$ when $t \geq \frac{k}{2}$, since $\ell \leq n$, but does nothing for the case $t < \frac{k}{2}$. For $t\geq k/2$, $\frac{\eta_t}{\tau_t}\left(\frac{n}{\ell}\right)^{k/2-t} \leq 1$ and hence the diagonal term together with $S_t+L_t$ is at most $\varepsilon^2/2+\varepsilon^2/4+\varepsilon^4/4<\varepsilon^2$. It is here we make a slight tweak to our choice of $\algval^*$, opting instead to apply a preprocessing edge-deletion step to the Kikuchi graph in order to prune the local degree small enough to bound the resulting term.

    \begin{lemma}[Edge deletion algorithm]
        \label{lem:edgedeletion} 
        Let $A_{\bU_t}$ be the adjacency matrix of the Kikuchi graph for some Hamiltonian $k$-$\XOR$ instance $\calI^{(t)} = (\calH^{(t)}, \{(P_C, b_C)\}_{C \in \calH^{(t)}})$ with a $\calU^{(t)}$-bipartite decomposition $\calH^{(t)} = \{\calH^{(t)}_U\}_{U \in \calU^{(t)}}$ satisfying the output criteria of \Cref{lem:decompositionalg}. Suppose the weight of edges of type $(C, C')$ for any $C \neq C' \in \calH_U^{(t)}$ is $\Delta^{(t)}$. Then, there exists a subgraph with symmetric adjacency matrix $\hat{A}_{\bU_t}$ satisfying:
        \begin{enumerate}
            \item $\hat{A}_{\bU_t}$ has $O(1)^k \cdot \varepsilon^{-2}$-bounded local degree for some universal constant.
            \item The weight of edges of type $(C, C')$ drops to $(1-\gamma) \Delta^{(t)}$ for some $\gamma \in [0, \frac{1}{2}]$.
        \end{enumerate}
    \end{lemma}

    The above establishes the local degree condition we need to succeed when $t < \frac{k}{2}$, at the cost of deleting some small constant fraction of the edges. Since we maintain uniform weights, all the observations previously made, \Cref{obs:kikuchirelaxationodd} and \Cref{fact:avgdegreeodd} hold with $\hat{K}_{\bU_t} = \hat{A}_{\bU_t}^* \otimes \Id_{2^n}$ substituted in place of $K_{\bU_t}$ with at most a $1-\gamma$ degradation quantitatively. Recompute $\Gamma_t$ from the pruned unsigned graph, and use its average and local degrees in the preceding estimates. By increasing the universal constant in the density hypothesis, those estimates give, when $t < k/2$,
    \begin{equation*}
        \algval(\bU_t)\leq \frac{\varepsilon^2}{2}
        +\frac{\varepsilon^2}{4D^k}\sqrt{\frac{\eta_t}{\tau_t}\left(\frac{n}{\ell}\right)^{k/2-t}}
        +\frac{\varepsilon^4}{4D^{2k}}\frac{\eta_t}{\tau_t}\left(\frac{n}{\ell}\right)^{k/2-t}\mcom
    \end{equation*}
    where $D>1$ is a sufficiently large universal constant and $\frac{\eta_t}{\tau_t}\left(\frac{n}{\ell}\right)^{k/2-t}$ now uses the pruned local degree. Here $\algval(\bU_t)$ is defined analogously to $\algval^*(\bU_t)$ using the pruned, degree-regularized Kikuchi matrix. Substituting $\tau_t=\left(\frac{3n}{\ell}\right)^{k/2-t}\cdot 4k^2\varepsilon^{-2}$ and the local-degree bound from \Cref{lem:edgedeletion} gives $\frac{\eta_t}{\tau_t}\left(\frac{n}{\ell}\right)^{k/2-t} \leq D^{2k}$. Consequently,
    \begin{equation*}
        \lambdamax(\bU_t)\leq \algval(\bU_t)
        \leq \frac{\varepsilon^2}{2}+\frac{\varepsilon^2}{4}
        +\frac{\varepsilon^4}{4}<\varepsilon^2\mper
    \end{equation*}
    The final algorithm, which we state here, uses this edge-deleted $\algval(\bU_t)$ as our final spectral certificate.

    \begin{tcolorbox}[
    width=\textwidth,   
    colframe=black,  
    colback=white,   
    title= Bipartite Hamiltonian $k$-$\XOR$ Refutation Algorithm,
    colbacktitle=white, 
    coltitle=black,      
    fonttitle=\bfseries,
    center title,   
    enhanced,       
    frame hidden,           
    borderline={1pt}{0pt}{black},
    sharp corners,
    toptitle=2.5mm
]
\textbf{Input:} A Hamiltonian $k$-$\XOR$ instance $\mathbfcal{I}^{(t)} = (\calH^{(t)}, \{(P_C, \bb_C)\}_{C \in \calH^{(t)}})$ with a $\calU^{(t)}$-bipartite decomposition $\calH^{(t)} = \{\calH^{(t)}_U\}_{U \in \calU^{(t)}}$ a subset of $\calH$ describing an associated operator $\bU_t = \frac{k^2\abs{\calU^{(t)}}}{4\abs{\calH}^2} \sum_{U \in \calU^{(t)}} \sum_{C, C' \in \calH_U} \bb_C \bb_{C'} P_{\wt{C}} P_{\wt{C}'}$.\\

\textbf{Output:} $\algval(\bU_t) \in \R_{\geq 0}$ with guarantee $\algval(\bU_t) \geq \lambdamax(\bU_t)$.\\

\textbf{Algorithm:}
\begin{enumerate}
    \item Construct the adjacency operator $A_{\bU_t}$ for the $N \times N$ Kikuchi matrix of \Cref{def:kikuchimatrixodd} $K_{\bU_t}$ from the input and apply the edge deletion algorithm of \Cref{lem:edgedeletion} to produce the pruned matrix $\hat{A}_{\bU_t}$. Let $\Gamma_t$ be the degree matrix of the pruned unsigned graph plus $d^{(t)}\Id$.
    \item Compute and output $\frac{k^2\abs{\calU^{(t)}}}{4\abs{\calH}^2} \sum_{C \in \calH^{(t)}} \bb_C^2 + \frac{k^2\abs{\calU^{(t)}}}{\abs{\calH}^2} \sum_{U \in \calU^{(t)}} \binom{\abs{\calH_U^{(t)}}}{2} \cdot \norm{\Gamma_t^{-1/2}\hat{A}^*_{\bU_t} \Gamma_t^{-1/2}}{2}$.
\end{enumerate}

\end{tcolorbox}
\end{proof}

We finish our proof by proving the auxiliary lemmas, \Cref{lem:spectralnormodd} and \Cref{lem:edgedeletion}.

\begin{proof}[Proof of \Cref{lem:spectralnormodd}]
    We set out to prove $\norm{\Gamma_t^{-1/2}A^*_{\bU_t}\Gamma_t^{-1/2}}{2}$ is bounded, starting with the case the underlying instance is Rademacher signed. We use the trace power method $\norm{\Gamma_t^{-1/2} A^*_{\bU_t} \Gamma_t^{-1/2}}{2} \leq \tr((\Gamma_t^{-1}A^*_{\bU_t})^{2r})^{1/2r}$ for $r = \lceil \log N \rceil$ and bound the expectation, with respect to the interaction coefficients, as
\begin{align*}
    \E\left[\tr\left(\left(\Gamma_t^{-1} A^*_{\bU_t}\right)^{2r}\right)\right] 
    &= \E\left[\tr\left(\left(\Gamma_t^{-1} \sum_{U \in \calU^{(t)}} \sum_{C \neq C' \in \calH^{(t)}_U} \bb_C\bb_{C'} A^U_{C, C'}\right)^{2r}\right) \right]\\
    &= \E\left[\tr\left(\sum_{\substack{(C_1, C_1') ,\dots, (C_{2r}, C_{2r}') \\ \in \calS(\calU^{(t)})}} \prod_{i = 1}^{2r} \Gamma_t^{-1} \cdot \bb_{C_i} \bb_{C_i'} A^{U_i}_{C_i, C_i'} \right) \right]\\
    &= \sum_{\substack{(C_1, C_1') ,\dots, (C_{2r}, C_{2r}') \\ \in \calS(\calU^{(t)})}} \E\left[\tr\left(\prod_{i = 1}^{2r} \Gamma_t^{-1} \cdot \bb_{C_i}  \bb_{C_i'} A^{U_i}_{C_i, C_i'} \right) \right]\\
    &= \sum_{\substack{(C_1, C_1') ,\dots, (C_{2r}, C_{2r}') \\ \in \calS(\calU^{(t)})}} \E\left[\prod_{i = 1}^{2r}\bb_{C_i} \bb_{C_i'}\right] \cdot \tr\left(\prod_{i = 1}^{2r} \Gamma_t^{-1} A^{U_i}_{C_i, C_i'} \right)\mcom
\end{align*}
where $\calS(\calU^{(t)}) = \bigcup_{U \in \calU^{(t)}} (\calH_U^{(t)})^{\otimes 2} \setminus \{(C, C) \mid C \in \bigcup_{U \in \calU^{(t)}} \calH_U^{(t)}\}$.

We observe that if for some $C \in \calH_U^{(t)}$, $C$ appears an odd number of times in the sequence of pairs $(C_1, C_1'), \dots, (C_{2r}, C_{2r}')$, then $\E\left[\prod_{i=1}^{2r} \bb_{C_i} \bb_{C_i'} \right] = 0$, because $\bb_C$ is independent across $\calH$. This motivates the following definition.

\begin{definition}[Trivially closed sequence]
    \label{def:gentriviallyclosedwalks}
    Let $(C_1, C_1') ,\dots, (C_{2r}, C_{2r}') \in \calS(\calU^{(t)})$ as defined above. We say $(C_1, C_1') ,\dots, (C_{2r}, C'_{2r})$ is trivially closed with respect to $C \in \calH_U^{(t)}$ if $C$ appears an even number of times. We say that the sequence is trivially closed if it is trivially closed with respect to all $C \in \calH_U^{(t)}$.
\end{definition}
    With the above definition in hand, the above reduces to saying
\begin{flalign*}
  &  \E\left[\tr\left(\left(\Gamma_t^{-1} \boldsymbol{A}\right)^{2r}\right)\right] 
    = \sum_{\substack{(C_1, C_1') ,\dots, (C_{2r}, C_{2r}') \\ \in \calS(\calU^{(t)})\\ \text{trivially closed}}}  \tr\left(\prod_{i = 1}^{2r} \Gamma_t^{-1} A^{U_i}_{C_i, C_i'} \right)\mper
\end{flalign*}

The following lemma yields the desired bound.
\begin{lemma}
    \label{lem:genmaincountingbacktrackingwalks}
    Let $a_r\coloneqq 8\eta_t r/d^{(t)}$. Then
    \begin{equation*}
        \sum_{\substack{(C_1, C_1') ,\dots, (C_{2r}, C_{2r}') \\ \in \calS(\calU^{(t)})\\ \text{trivially closed}}}  \tr\left(\prod_{i = 1}^{2r} \Gamma_t^{-1} A^{U_i}_{C_i, C_i'} \right) \leq N \cdot 2^{2r} \cdot \max\{a_r^r,a_r^{2r}\} \mper
    \end{equation*}
\end{lemma}
Taking $r=\lceil\log N\rceil$ and applying Markov's inequality gives a norm bound $O(\max\{\sqrt{a_r},a_r\})$ with probability at least $1-N^{-1}$. Since $r=O(\ell\log n)$, this proves the Rademacher case of \Cref{lem:spectralnormodd} for a sufficiently large $C_0$.

\begin{proof}[Proof of \Cref{lem:genmaincountingbacktrackingwalks}]
    We bound the sum as follows. First, we observe that for a trivially closed sequence $(C_1, C'_1) ,\dots, (C_{2r}, C'_{2r})$, we have
    \begin{flalign*}
        \tr\left(\prod_{i = 1}^{2r} \Gamma_t^{-1} A^{U_i}_{C_i, C_i'} \right) \leq \sum_{Q_0, Q_1, \dots, Q_{2r-1} \in \calQ_\ell(n)} \prod_{i = 0}^{2r - 1} (\Gamma_t^{-1})_{Q_i, Q_i} \cdot \mathbf{1}\{Q_i\xrightarrow{\text{$P_{\wt{C}_{i+1}}, P_{\wt{C}'_{i+1}}$}} Q_{i+1}\} \mper
    \end{flalign*}
    with $Q_0 = Q_{2r}$. This would be an equality if not for the signs within $A_{C_i, C_i'}^{U_i}$, which are eliminated with a simple triangle inequality. Now, the sum that we wish to bound in \Cref{lem:genmaincountingbacktrackingwalks} simply counts the total weight of the trivially closed walks $Q_0, C_1, C'_1, Q_1, \dots, Q_{2r-1}, C_{2r}, C'_{2r}, Q_{2r}$ in the Kikuchi graph $A_{\bU_t}$, where the weight of a walk is simply $\prod_{i = 0}^{2r-1} (\Gamma_t^{-1})_{Q_i, Q_i}$.
    
    Let us now bound this weight by uniquely encoding a walk $Q_0, C_1, C'_1, Q_1, \dots, C_{2r}, C'_{2r}, Q_{2r}$ as follows.
    \begin{itemize}
        \item For $i = 1, \dots, 2r$, we let $z_i$ be $1$ if (1) $C_i = C_j$, (2) $C_i = C_j'$, (3) $C'_i = C_j$, or (4) $C'_i = C'_j$ for some $j < i$. In this case, we say that the edge is ``old''. Otherwise $z_i = 0$, and we say that the edge is ``new''.
        \item Now, we write down the start vertex $Q_0$. 
    	\item For $i = 1, \dots, 2r$, if $z_i$ is $1$ then we encode $Q_i$ by writing down the smallest $j \in [2r]$, the least $c \in [4]$ specifying which of the 4 cases above holds, and lastly the other parallel edge to take. For instance if we have $C_i = C_j$ we must specify $C'_i$ from $\calH_U^{(t)}$ for the other half of the edge.
    	\item For $i = 1, \dots, 2r$, if $z_i$ is $0$ then we encode $Q_i$ by writing down an integer in $1, \dots, \deg(Q_{i-1})$ that specifies the edge we take to move to $Q_i$ from $Q_{i-1}$ (we associate $[\deg(Q_{i-1})]$ to the edges adjacent to $Q_{i-1}$ with an arbitrary fixed map).
    \end{itemize}
    With the above encoding, we can now bound the total weight of all trivially closed walks as follows. First, let us consider the total weight of walks for some fixed choice of $z_1, \dots, z_{2r}$. We have $N$ choices for the start vertex $Q_0$. For each $i = 1, \dots, 2r$ where $z_i = 0$, we have $\deg(Q_{i-1})$ choices for $Q_i$, and we multiply by a weight of $(\Gamma_t^{-1})_{Q_{i-1}} \leq \frac{1}{\deg(Q_{i-1})}$. For each $i = 1, \dots, 2r$ where $z_i = 1$, we have at most $2r$ choices for the index $j < i$ and $4$ for the bit $c$. For the choice of partner, without loss of generality we say $C'_i$, we note that $Q_{i-1}, C_i$, and $C_i$'s position are fixed. The number of viable $C'_i$ is then just the local degree, and by the assumption of $\eta$-bounded local degree we can limit this to $\eta$ choices. Finally, we multiply by a weight of $(\Gamma_t^{-1})_{Q_{i-1}, Q_{i-1}} \leq \frac{1}{d^{(t)}}$. Hence, the total weight for a specific $z_1, \dots, z_{2r}$ is at most $N \left(\frac{8\eta_t r}{d^{(t)}}\right)^{\abs{z}}$, where $\abs{z}$ is the number of $z_i$ such that $z_i = 1$.
    
    Finally, we observe that any trivially closed walk must have $\abs{z} \geq r$ to satisfy the even multiplicity condition. Also $\abs{z}\leq 2r$, so $a_r^{\abs{z}}\leq\max\{a_r^r,a_r^{2r}\}$ for every $a_r>0$. Summing over all $z_1,\dots,z_{2r}$ gives $N2^{2r}\max\{a_r^r,a_r^{2r}\}$, which finishes the proof.
\end{proof}

Now we handle the Gaussian case. Since the odd Gaussian moments are $0$, the definition of trivially closed sequence from \Cref{def:gentriviallyclosedwalks} still holds, but we additionally want to count the number of coefficient repetitions, so we define the following.

    \begin{definition}[$\{m_1, \dots, m_q\}$-walks for Gaussian sequences]
        Given a trivially closed walk\\ $(C_1, C_1') ,\dots, (C_{2r}, C_{2r}') \in \calS(\calU^{(t)})$, we define the type of the walk as follows. For $(C_i, \cdot)$ in the walk, we say the edge has type $C_i$. Let $\{m_1, \dots, m_q\}$ be the multiplicities of the types of all $2r$ edges. We now repeat this process for $(\cdot, C_i')$, getting a new set of $2r$ types forming a set of multiplicities $\{m_1', \dots, m_{q'}'\}$. We say a walk is a top $\{m_1, \dots, m_q\}$-walk if $\prod_{i=1}^q m_i! \geq \prod_{i = 1}^{q'} m_i'!$. Otherwise we say it is bottom $\{m_1', \dots, m_{q'}'\}$. Note an $m_i$ may be equal to $1$, in which case we call the corresponding edge $(C_i, C_i')$ a singleton edge.
    \end{definition}

    The point of this definition is to claim the following
    \begin{align*}
      \E\left[\tr((\Gamma^{-1} \boldsymbol{A})^{2r})\right] &\leq \sum_{\substack{\{m_1, \dots, m_q\}\\ \sum_{i=1}^q m_i = 2r\\ \forall i \in [q], m_i \geq 1}} \sum_{\substack{(C_1, C_1') ,\dots, (C_{2r}, C_{2r}') \\ \in \calS(\calU^{(t)})\\ \{m_1, \dots, m_q\}\text{-walk}}} \E\left[\prod_{i = 1}^{2r}\bb_{C_i} \bb_{C_i'}\right] \cdot \abs{\tr\left(\prod_{i = 1}^{2r} \Gamma_t^{-1} A^{U_i}_{C_i, C_i'} \right)}\\
      &\leq 2^{2r+1}\sum_{\substack{\{m_1, \dots, m_q\}\\ \sum_{i=1}^q m_i = 2r\\ \forall i \in [q], m_i \geq 1}} \prod_{i = 1}^q m_i! \sum_{\substack{(C_1, C_1') ,\dots, (C_{2r}, C_{2r}') \\ \in \calS(\calU^{(t)})\\ \text{top }\{m_1, \dots, m_q\}\text{-walk}}} \abs{\tr\left(\prod_{i = 1}^{2r} \Gamma_t^{-1} A^{U_i}_{C_i, C_i'} \right)} \mper
    \end{align*}

    To introduce the term $2^{2r}\prod_{i=1}^q m_i!$, we argue the following. By the top $\{m_1, \dots, m_q\}$-walk assumption, we know the multiplicities of half of the $\bb_C$, and let $\{m_1', \dots, m_{q'}'\}$ be the multiplicities of the remaining bottom half. The natural prefactor should perhaps be $\prod_{i=1}^q (m_i-1)!! \prod_{i=1}^{q'} (m_i'-1)!!$, but this fails to account for the fact that some groups of coefficients may \textit{repeat} across the top and bottom half. Thinking adversarial, we can increase the value of the prefactor by merging groups of size $m_i$ and $m_j'$ for $i \in [q]$ and $j \in [q']$, thus trading $(m_i-1)!! (m_j'-1)!!$ for $(m_i+m_j'-1)!!$. We make two observations now: (1) merging groups is always advantageous as long as we assume in the worst case $m_i + m_j'$ is always even and (2) we can only ever merge pairs, since by assumption the types are always distinct within the top and bottom. The real prefactor is then $\prod_{i=1}^{\max(q, q')} (m_i + m_i' - 1)!!$ up to reordering the indices, which we bound with the following lemma.
    \begin{lemma}
        \label{lem:whatithinkisthefinalpieceofthepuzzle}
        Let $m_1, \dots, m_q$ and $m_1', \dots m_{q'}' \in \N$ be such that $\sum_{i=1}^q m_i = \sum_{i = 1}^{q'} m_i' = 2r$. Then,
        \begin{equation*}
            \prod_{i=1}^{\max(q, q')} (m_i + m_i' - 1)!! \leq 2^{2r} \cdot\max\left(\prod_{i=1}^q m_i!, \prod_{i=1}^{q'} m_i'!\right)\mper
        \end{equation*}
        Here if $m_i$ has $i > q$ then we interpret it to be $0$ and vice versa.
    \end{lemma}

    Applying the lemma immediately gives the inequality exhibited above, so we give a proof here.

    \begin{proof}[Proof of \Cref{lem:whatithinkisthefinalpieceofthepuzzle}]
        Without loss of generality, assume $q \leq q'$, then we write $\prod_{i=1}^{\max(q, q')} (m_i + m_i' - 1)!! = \prod_{i=1}^q (m_i + m_i' - 1)!! \cdot \prod_{i=q+1}^{q'} (m_i'-1)!!$. Fix $m_i, m_i'$ for $i \in [q]$ such that $m_i + m_i'$ is even. In this case we have
        \begin{equation*}
            (m_i + m_i' -1)!! \leq (m_i + m_i')!! = 2^{\frac{m_i+m_i'}{2}} \left(\frac{m_i+m_i'}{2}\right)! \leq 2^{\frac{m_i+m_i'}{2}}\sqrt{m_i!m_i'!}\mper
        \end{equation*}
        In the last line we appeal to the weak log-convexity of the factorial function. Similarly, if $m_i + m_i'$ is odd we have
        \begin{equation*}
            (m_i + m_i' -1)!! \leq 2^{\frac{m_i+m_i'-1}{2}} \left(\frac{m_i+m_i'-1}{2}\right)! \leq 2^{\frac{m_i+m_i'}{2}} \sqrt{m_i!m_i'!}\mper
        \end{equation*}
        Finally, we observe for any $i > q$, we have $(m_i'-1)!! \leq \sqrt{m_i'!}$ by comparing terms. Combining these bounds, we write
        \begin{align*}
            \prod_{i=1}^{\max(q, q')} (m_i + m_i' - 1)!! 
            &\leq \prod_{i=1}^{\max(q, q')} 2^{\frac{m_i+m_i'}{2}} \sqrt{m_i! m_i'!}\\
            &\leq 2^{\frac{\sum_{i=1}^q m_i + \sum_{i=1}^{q'} m_i'}{2}} \sqrt{\prod_{i=1}^q m_i! \prod_{i=1}^{q'} m_i'!}\\
            &\leq 2^{2r} \cdot\max\left(\prod_{i=1}^q m_i!, \prod_{i=1}^{q'} m_i'!\right)\mper
        \end{align*}
    \end{proof}
    
    We now turn our attention to bounding the sum in its entirety. As before, we can appeal to the $p(2r) \leq (e^{\pi \sqrt{\frac{2}{3}}})^{2r} \leq 16^{2r}$ bound on the number of integer partitions $p(2r)$ to focus on the maximum of the inner term. To finish, we prove the following.
    
    \begin{lemma}
    \label{lem:gaussianbacktrackingwalks}
    For any choice of $m_1, \dots, m_q$ with $\sum_{i=1}^q m_i = 2r$ let $a_r\coloneqq 8\eta_t r/d^{(t)}$. Then
    \begin{equation*}
        2^{2r+1}\prod_{i = 1}^q m_i! \sum_{\substack{(C_1, C_1') ,\dots, (C_{2r}, C_{2r}') \\ \in \calS(\calU^{(t)})\\ \text{top } \{m_1, \dots, m_q\}\text{-walk}}} \abs{\tr\left(\prod_{i = 1}^{2r} \Gamma_t^{-1} A^{U_i}_{C_i, C_i'} \right)} \leq N \cdot 2^{6r+2} \cdot \max\{a_r^r,a_r^{2r}\} \mper
    \end{equation*}
\end{lemma}
After summing over the at most $16^{2r}$ multiplicity profiles, the expected trace moment is at most
\begin{align*}
    N16^{2r}2^{6r+2}\max\{a_r^r,a_r^{2r}\}\mper
\end{align*}
Taking $r=\lceil\log N\rceil$ and applying Markov's inequality therefore gives a norm bound $O(\max\{\sqrt{a_r},a_r\})$ with probability at least $1-N^{-1}$. Since $r=O(\ell\log n)$, increasing $C_0$ proves the Gaussian case as well. \end{proof}

\begin{proof}[Proof of \Cref{lem:gaussianbacktrackingwalks}]
    As before, we have for any trivially closed sequence $(C_1, C'_1) ,\dots, (C_{2r}, C'_{2r})$,
    \begin{flalign*}
         \abs{\tr\left(\prod_{i = 1}^{2r} \Gamma_t^{-1} A^{U_i}_{C_i, C_i'} \right)} \leq \sum_{Q_0, Q_1, \dots, Q_{2r-1} \in \calQ_\ell(n)} \prod_{i = 0}^{2r - 1} (\Gamma_t^{-1})_{Q_i, Q_i} \cdot \mathbf{1}\{Q_i\xrightarrow{\text{$P_{\wt{C}_{i+1}}, P_{\wt{C}'_{i+1}}$}} Q_{i+1}\} \mper
    \end{flalign*}
    with $Q_0 = Q_{2r}$. The sum in \Cref{lem:gaussianbacktrackingwalks} simply counts the total weight of the trivially closed walks $Q_0, C_1, C'_1, Q_1, \dots, Q_{2r-1}, C_{2r}, C'_{2r}, Q_{2r}$ in the Kikuchi graph $A_{\bU_t}$, where the weight of a walk is simply $\prod_{i = 0}^{2r-1} (\Gamma_t^{-1})_{Q_i, Q_i}$ multiplied by the factorial prefactor $\prod_{i = 1}^q m_i!$.
    
    Let us bound this weight by uniquely encoding a walk $Q_0, C_1, C'_1, Q_1, \dots, C_{2r}, C'_{2r}, Q_{2r}$ as follows.
    \begin{itemize}
        \item First, we choose the template for the walk, which is the way in which the indices $1, \dots, 2r$ are related such that the multiplicities indeed fulfill a $\{m_1, \dots, m_q\}$-walk. Formally, we can let a template be a partition $T$ of the indices $1, \dots, 2r$ into indistinguishable buckets of sizes $\{m_1, \dots, m_q\}$.
        \item Let $q_0$ be the number of singleton edges $(C_i, C_i')$, the number of $m_j = 1$. Construct a string $z \in \{0,1\}^{q_0}$ as follows. We identify the indices for the singleton edges with $[q_0]$. For $i \in [q_0]$, we look at the corresponding edge $(C_i, C_i')$, and if $C_i$ appears as $C_j' = C_i$ for some earlier $j < i$, we let $z_i = 1$. Otherwise, we let $z_i = 0$. Intuitively, the entries with $z_i = 1$ are the ``old'' edges, repeating one of their entries earlier in the walk.
        \item Now we write down the start vertex $Q_0$.
        \item For $i = 1, \dots, 2r$, if $i$ is the first index in a bucket of size at least $2$, we encode the edge $(C_i, C_i')$ from the neighbors of $Q_{i-1}$ walked along, which can be done with a number in $[\deg(Q_{i-1})]$. If $i$ is not the first index in such a bucket, then we look at the first edge with the same type $(C_j, C_j')$ according to $T$. We choose an edge from $Q_{i-1}$ among those sharing at least $C_j$ or $C_j'$, since this is required to fulfill multiplicities. This can be encoded with a bit $b$ indicating $C_j$ or $C_j'$, and then a number in $[\eta_t]$, since there are at most $\eta_t$ neighbors with a fixed entry $C_j$ or $C_j'$ by $\eta_t$-boundedness.
        \item For $i = 1, \dots, 2r$ corresponding to a singleton edge $(C_i, C_i')$ according to $T$, if $z_i = 1$, we encode its pair with an index $j$ in $[2r]$, whether $C_i$ or $C_i'$ repeats with $b$, then encode its neighbor of $Q_{i-1}$ restricted to those matching $C_j$ or $C_j'$ with a number within $[\eta_t]$. If $z_i = 0$, we simply encode which neighbor of $Q_{i-1}$ is taken, a number from $[\deg(Q_{i-1})]$.
    \end{itemize}
    With the above encoding, we can now bound the total weight of closed $\{m_1, \dots, m_q\}$-walks as follows. First, let us consider the total weight of walks for some fixed choice of template $T$. We have $N$ choices for the start vertex $Q_0$. Looking first at entries which belong to a bucket of size at least $2$, for each $i = 1, \dots, 2r$ if $i$ is the first index in its bucket, we have $\deg(Q_{i-1})$ choices for $Q_i$, and we multiply by a weight of $\Gamma^{-1}_{Q_{i-1}, Q_{i-1}} \leq \frac{1}{\deg(Q_{i-1})}$, so the contribution to the product is $1$. Then, for each $i = 1, \dots, 2r$ where $i$ is not the first in its buckets, the number of choices is limited to be $2\eta_t$ by the assumption of $\eta_t$-bounded local degree, since after fixing either $C_j$ or $C_j'$ to appear in the edge, there are at most $\eta_t$ choices for its pair. We still multiply by a weight of $\Gamma^{-1}_{Q_{i-1}, Q_{i-1}} \leq \frac{1}{d^{(t)}}$, giving a factor $\left(\frac{2\eta_t}{d^{(t)}}\right)^{2r-q}$ since there are $2r-q_0$ non-singleton entries and $q-q_0$ of them are the first of their respective bucket. Second, we consider the singleton edges. There are $2^{q_0}$ ways to choose the string $z \in \{0, 1\}^{q_0}$. For each $i= 1, \dots, q_0$ where $z_i = 0$, we have $\deg(Q_{i-1})$ choices for $Q_i$, and we multiply by a weight of $\Gamma^{-1}_{Q_{i-1}, Q_{i-1}} \leq \frac{1}{\deg(Q_{i-1})}$, so the contribution to the product is again $1$. For each $i=1,\dots,q_0$ where $z_i=1$, the index in $[2r]$, the choice of a previous occurrence, and the restricted neighbor set give at most $8\eta_t r$ choices. After multiplying by the weight $\Gamma^{-1}_{Q_{i-1},Q_{i-1}}\leq 1/d^{(t)}$, the contribution from these steps is at most $a_r^{\abs{z}}$. There is a nice trick we can use to claim $\abs{z} \geq \frac{q_0}{2}$. Observe, every singleton $(C_i, \cdot)$ must have $C_i$ repeat in the walk to be backtracking, therefore it must match with a bottom entry either earlier or later in the walk, so there must be at least $\frac{q_0}{2}$ matches backwards or forwards. To make sure the majority is in the backward direction, we can simply encode the reversed walk instead, and just pay an extra factor $2$ to encode whether to read the walk in reverse or not. Since $q_0/2\leq\abs{z}\leq q_0$, the total weight for walks from a fixed template is at most $2^{q_0+1}N\left(\frac{2\eta_t}{d^{(t)}}\right)^{2r-q}\max\{a_r^{q_0/2},a_r^{q_0}\}$, without assuming $a_r\leq 1$. Letting $\binom{2r}{\{m_1, \dots, m_q\}}$ count the number of templates, our bound is now
    \begin{equation*}
        \prod_{i = 1}^q m_i! \binom{2r}{\{m_1, \dots, m_q\}} \cdot 2^{q_0+1} N \left(\frac{2\eta_t }{d^{(t)}}\right)^{2r-q} \max\{a_r^{q_0/2},a_r^{q_0}\}\mper
    \end{equation*}
    By supposing all $m_i = 1$ fulfill $i \in  \{q-q_0+1 \dots ,q\}$, we can simplify this to
    \begin{equation*}
        2^{2r}\prod_{i = 1}^{q-q_0} m_i! \binom{2r-q_0}{\{m_1, \dots, m_{q-q_0}\}} \cdot 2^{q_0+1} N \left(\frac{2\eta_t }{d^{(t)}}\right)^{2r-q} \max\{a_r^{q_0/2},a_r^{q_0}\}\mper
    \end{equation*}

    Here we just choose the subset of $[2r]$ that are singleton at a cost of $2^{2r}$, and then choose the rest from the remainder. The quantity $\prod_{i=1}^{q-q_0} m_i! \binom{2r-q_0}{\{m_1, \dots, m_{q-q_0}\}}$ can then be interpreted as choosing (1) a partition of $[2r-q_0]$ into $\{m_1, \dots, m_{q-q_0}\}$ and (2) an ordering within each bucket. We now argue we can upper bound this by $2^{2r-q_0}(2r)^{2r-q}$. Plugging this into the above yields a bound
    \begin{equation*}
        2^{4r+1} N \cdot\max\{a_r^{2r-q+q_0/2},a_r^{2r-q+q_0}\}\mper
    \end{equation*}
    We finish by noticing $2r-q+q_0/2 \geq r$, since $2r = \sum_{i=1}^q m_i \geq 2(q-q_0) + q_0$ which implies $r \geq q-q_0/2$. Also $q_0\leq q$, so
    \begin{equation*}
        r\leq 2r-q+q_0/2\leq 2r-q+q_0\leq 2r\mper
    \end{equation*}
    Both powers are therefore at most $\max\{a_r^r,a_r^{2r}\}$. Restoring the prefactor $2^{2r+1}$ in the lemma statement gives $N2^{6r+2}\cdot \max\{a_r^r,a_r^{2r}\}$, as required.

    To see the upper bound then, note we can encode such an object as follows. Scanning through $i = 1,\dots, 2r$ we construct a string $z \in \{0,1\}^{2r-q_0}$ by letting $z_i = 0$ if $i$ is the first element in its bucket (according to the chosen ordering), and $z_i = 1$ otherwise. Now for each $z_i = 1$, we have that it is the $j$th element for $j \geq 2$ in whatever bucket it falls in. We encode this by specifying the index of the preceding element, so the $i \in [2r-q_0]$ which is the $(j-1)$th element of its bucket. Clearly there are $2^{2r-q_0}$ choices for $z$, and given a $z$, we need to specify at most $(2r-q_0)^{\abs{z}}$ preceding elements for each $z_i = 1$, with $\abs{z} = 2r-q_0-(q-q_0) = 2r-q$ when the number of buckets is $q-q_0$. The final bound is then below $2^{2r-q_0}(2r)^{2r-q}$ as desired.
\end{proof}

\subsection{Edge deletion algorithm}

In this section, we design and analyze the following greedy edge deletion algorithm for \Cref{lem:edgedeletion}.

\begin{tcolorbox}[
    width=\textwidth,   
    colframe=black,  
    colback=white,   
    title=Edge Deletion Algorithm,
    colbacktitle=white, 
    coltitle=black,      
    fonttitle=\bfseries,
    center title,   
    enhanced,       
    frame hidden,           
    borderline={1pt}{0pt}{black},
    sharp corners,
    toptitle=2.5mm
]
\textbf{Input:} $A_{\bU_t}$ the level-$\ell$ adjacency matrix of the Kikuchi graph for some Hamiltonian $k$-$\XOR$ instance $\calI^{(t)} = (\calH^{(t)}, \{(P_C, b_C)\}_{C \in \calH^{(t)}})$ with a $\calU^{(t)}$-bipartite decomposition $\calH^{(t)} = \{\calH^{(t)}_U\}_{U \in \calU^{(t)}}$ satisfying the output criteria of \Cref{lem:decompositionalg}.\\

\textbf{Output:} A subgraph $\hat{A}_{\bU_t}$ satisfying the criteria of \Cref{lem:edgedeletion}.\\

\textbf{Algorithm:}
\begin{enumerate}
    \item While there is a vertex $Q \in \calQ_\ell(n)$ and $C \in \calH^{(t)}$ such that $d_{Q,C,0}$ or $d_{Q, C, 1} > \eta$, delete an arbitrary pair of edges $(Q,R)$ and $(R, Q)$ of type $(C, \cdot)$ or $(\cdot, C)$ respectively.
    \item Let $\rho$ be the largest fraction of edges deleted for any pair $(C, C') \in \bigcup_{U \in \calU^{(t)}} (\calH^{(t)}_U)^{\otimes 2}$. Delete edges arbitrarily until exactly a $\rho$ fraction has been deleted for all pairs.
    \item Output the remaining graph.
\end{enumerate}

\end{tcolorbox}

\begin{proof}[Proof of \Cref{lem:edgedeletion}]
    By the specification of the algorithm, any graph output clearly satisfies $\eta$-bounded local degree, so choosing $\eta = O(1)^k \cdot \varepsilon^{-2}$ for some constant to be chosen later is sufficient for the lemma. What remains is to show the largest fraction $\rho$ deleted is at most $\frac{1}{2}$, which is a result of the following technical lemma.
    \begin{lemma}
        \label{lem:edgedeletionanalysis}
        Let $A_{\bU_t}$ be the adjacency matrix of the Kikuchi graph for some Hamiltonian $k$-$\XOR$ instance $\calI^{(t)} = (\calH^{(t)}, \{(P_C, b_C)\}_{C \in \calH^{(t)}})$ with a $\calU^{(t)}$-bipartite decomposition $\calH^{(t)} = \{\calH^{(t)}_U\}_{U \in \calU^{(t)}}$ satisfying the output criteria of \Cref{lem:decompositionalg}. Fix $C \neq C' \in \calH_U^{(t)}$ for some $U \in \calU^{(t)}$. Let $E_{(C, C')}$ be the set of edges in the original graph associated with this pair. In the first step of the edge deletion algorithm above, we have the following guarantee for some universal constant
        \begin{equation*}
            \Pr_{(\boldsymbol{Q}, \boldsymbol{R}) \sim E_{(C, C')}}[(\boldsymbol{Q}, \boldsymbol{R}) \text{ deleted}] \leq \frac{D^k}{\eta} \sum_{s = 0}^{k-t} \tau_{s+t} \min\left(1, \left(\frac{\ell}{3n}\right)^{\lfloor \frac{k-t}{2}\rfloor - s}\right)\mper
        \end{equation*}
    \end{lemma}

    Note this bound directly translates to a bound for $\rho$, as the probability is just the fraction of deleted edges for some $(C, C')$. We then show setting $\eta = 16D^k k^3 \varepsilon^{-2}$ achieves $\rho \leq \frac{1}{2}$. In the case $s + t \geq \frac{k+1}{2}$, we can note that $\tau_{s+t} = 4k^2 \varepsilon^{-2}$ and bound
    \begin{equation*}
        \rho_{\mathrm{high}} \coloneqq \frac{D^k}{\eta} \sum_{\substack{s = 0\\ s +t \geq \frac{k+1}{2}}}^{k-t} \tau_{s+t} \min\left(1, \left(\frac{\ell}{3n}\right)^{\lfloor \frac{k-t}{2}\rfloor - s}\right) \leq \frac{1}{16k^3\varepsilon^{-2}} \cdot  (k-t+1)\tau_{s+t} \leq \frac{1}{16k^2\varepsilon^{-2}} \cdot  4k^2\varepsilon^{-2} \leq \frac{1}{4}\mper
    \end{equation*}
    
    In the case $s + t \leq \frac{k-1}{2}$ we have
    \begin{flalign*}
        \rho_{\mathrm{low}} &\coloneqq \frac{D^k}{\eta} \sum_{\substack{s = 0 \\ s + t \leq \frac{k-1}{2}}}^{k-t} \tau_{s+t} \min\left(1, \left(\frac{\ell}{3n}\right)^{\lfloor \frac{k-t}{2}\rfloor - s}\right) \leq \frac{D^k}{\eta} \sum_{\substack{s = 0 \\ s + t \leq \frac{k-1}{2}}}^{k-t} 4k^2\varepsilon^{-2}\left(\frac{3n}{\ell}\right)^{k/2-s-t} \left(\frac{\ell}{3n}\right)^{\lfloor \frac{k-t}{2}\rfloor - s}\\
        &\leq \frac{D^k}{\eta} \sum_{\substack{s = 0 \\ s + t \leq \frac{k-1}{2}}}^{k-t} 4k^2\varepsilon^{-2}\left(\frac{\ell}{3n}\right)^{\frac{t}{2} - \frac{\mathbf{1}\{k-t\text{ is odd}\}}{2}}\leq \frac{1}{4}\left(\frac{\ell}{3n}\right)^{\frac{t}{2} - \frac{\mathbf{1}\{k-t\text{ is odd}\}}{2}}\mper
    \end{flalign*}
    Since $t \geq 1$, we get $\frac{1}{4}$ as our desired bound. Adding the two together yields $\rho = \rho_{\mathrm{high}} + \rho_{\mathrm{low}} \leq \frac{1}{2}$.
\end{proof}

\begin{proof}[Proof of \Cref{lem:edgedeletionanalysis}]
    Fix $Q \in \calQ_\ell(n)$ a vertex, $C \in \calH^{(t)}_U$ for $U \in \calU^{(t)}$, and $b \in \{0,1\}$. We crudely assume the algorithm deletes \textit{all} edges with $d_{Q, C, b} > \eta$. We then compute:
    \begin{align*}
        \Pr_{(\boldsymbol{Q}, \boldsymbol{R}) \sim E_{(C, C')}}[(\boldsymbol{Q}, \boldsymbol{R}) \text{ deleted}]
        &\leq 2\Pr_{(\boldsymbol{Q}, \boldsymbol{R}) \sim E_{(C, C')}}[d_{\boldsymbol{Q}, C, 0} > \eta \cup d_{\boldsymbol{Q}, C', 1} > \eta]\\
        &\leq 2\Pr_{(\boldsymbol{Q}, \boldsymbol{R}) \sim E_{(C, C')}}[d_{\boldsymbol{Q}, C, 0} > \eta] + 2\Pr_{(\boldsymbol{Q}, \boldsymbol{R}) \sim E_{(C, C')}}[d_{\boldsymbol{Q}, C', 1} > \eta]\\
        &\leq \frac{2}{\eta}\E_{(\boldsymbol{Q}, \boldsymbol{R}) \sim E_{(C, C')}}[d_{\boldsymbol{Q}, C, 0}-1] + \frac{2}{\eta}\E_{(\boldsymbol{Q}, \boldsymbol{R}) \sim E_{(C, C')}}[d_{\boldsymbol{Q}, C', 1}-1]\\
        &\leq \frac{2}{\eta}\sum_{C'' \neq C, C' \in \calH^{(t)}_U} \Pr_{(\boldsymbol{Q}, \boldsymbol{R}) \sim E_{(C, C')}}[\exists R' \text{ s.t. } \boldsymbol{Q}  \xrightarrow{\text{$P_{\wt{C}}, P_{\wt{C}''}$}} R'] \\&+ \frac{2}{\eta}\sum_{C'' \neq C, C' \in \calH^{(t)}_U} \Pr_{(\boldsymbol{Q}, \boldsymbol{R}) \sim E_{(C, C')}}[\exists R' \text{ s.t. } \boldsymbol{Q}  \xrightarrow{\text{$P_{\wt{C}''}, P_{\wt{C}'}$}} R']\mper
    \end{align*}

We can think of this process like this: after fixing $(C, C')$ we sample an edge $(\boldsymbol{Q}, \boldsymbol{R})$. We are now interested in the probability that the edge we sampled is incident to an edge $(\boldsymbol{Q}, R')$ involving $C$. The main fact we establish is a bound on the probability such an edge exists given uniform sampling.
\begin{lemma}
    \label{fact:crossterms}
    Fix $C, C', C'' \in \calH_U^{(t)}$.
    \begin{equation*}
    \Pr_{(\boldsymbol{Q}, \boldsymbol{R}) \sim E_{(C, C')}}[\exists R' \text{ s.t. } \boldsymbol{Q} \xrightarrow{\text{$P_{\wt{C}}, P_{\wt{C}''}$}} R'] \leq \binom{k-t}{\lfloor \frac{k-t}{2} \rfloor}\min\left(1, \left(\frac{\ell}{3n}\right)^{\lfloor \frac{k-t}{2}\rfloor - \abs{P_{C'} \sqcap P_{C''}}+t}\right)\mcom
    \end{equation*}
    where we define $P_C \sqcap P_{C'} \subseteq [n]$ as $\supp(P_C) \cap \supp(P_{C'}) \setminus \supp(P_C P_{C'})$.
\end{lemma}

Using \Cref{fact:crossterms} we continue above by writing
\begin{align*}
    \Pr_{(\boldsymbol{Q}, \boldsymbol{R}) \sim E_{(C, C')}}[(\boldsymbol{Q}, \boldsymbol{R}) \text{ deleted}] 
    &\leq \frac{4}{\eta}\sum_{C'' \neq C, C' \in \calH^{(t)}_U} \binom{k-t}{\lfloor \frac{k-t}{2} \rfloor}\min\left(1, \left(\frac{\ell}{3n}\right)^{\lfloor \frac{k-t}{2}\rfloor - \abs{P_{C'} \sqcap P_{C''}}+t}\right)\\
    &\leq \frac{4}{\eta} \cdot \binom{k-t}{\lfloor \frac{k-t}{2} \rfloor} \sum_{s = 0}^{k-t} \sum_{\substack{C'' \neq C, C' \in \calH^{(t)}_U \\ J = P_{C'} \sqcap P_{C''} \\ \abs{J} = s+t}} \min\left(1, \left(\frac{\ell}{3n}\right)^{\lfloor \frac{k-t}{2}\rfloor - s}\right)\mper
\end{align*}

In the second line, we are essentially partitioning all $C''$ based on how the operator $P_{C''}$ intersects with $P_{C'}$. Note that since $C, C' \in \calH^{(t)}_U$, it is guaranteed $\abs{P_C \sqcap P_{C'}} \geq t$, since by definition they agree on $U$. We then make the observation that for a fixed $J$ with $\abs{J} = s+t$ for $s > 0$, the maximum number of $C'' \in \calH^{(t)}_U$ with $J = P_{C'} \sqcap P_{C''}$ is $\tau_{s+t}$. This follows directly from the $(\frac{\varepsilon}{2k}, \ell)$-regularity condition guaranteed in \Cref{item:decompositionalg5} of \Cref{lem:decompositionalg}. When $\abs{J} = t$, this bound follows more immediately from the fact that $\abs{\calH^{(t)}_U} \leq \tau_t$. This allows us to bound the number of terms in the inner sum above nicely as
\begin{align*}
    \Pr_{(\boldsymbol{Q}, \boldsymbol{R}) \sim E_{(C, C')}}[(\boldsymbol{Q}, \boldsymbol{R}) \text{ deleted}]
    &\leq \frac{4}{\eta} \cdot \binom{k-t}{\lfloor \frac{k-t}{2} \rfloor} \sum_{s = 0}^{k-t} \sum_{\abs{J} = s+t, J \subseteq C'} \tau_{s + t} \min\left(1, \left(\frac{\ell}{3n}\right)^{\lfloor \frac{k-t}{2}\rfloor - s}\right)\\
    &\leq \frac{4 \binom{k-t}{\lfloor \frac{k-t}{2} \rfloor} \binom{k}{\lfloor k/2\rfloor}}{\eta} \sum_{s = 0}^{k-t} \tau_{s+t} \min\left(1, \left(\frac{\ell}{3n}\right)^{\lfloor \frac{k-t}{2}\rfloor - s}\right)\mper
\end{align*}

Standard binomial estimates give that for some constant $D > 0$ the leading coefficient may be bounded by $\frac{D^k}{\eta}$, which is exactly the bound we sought. We finish by proving \Cref{fact:crossterms}.
 \end{proof}

\begin{proof}[Proof of \Cref{fact:crossterms}]

For arbitrary $(C, C')$ we let $\boldsymbol{Q} \in \calQ_\ell(n)$ be any vertex such that $\boldsymbol{Q} \xrightarrow{P_{\wt{C}}, P_{\wt{C}'}} \boldsymbol{R}$. We are interested in the probability that there exists $R'$ with $\boldsymbol{Q} \xrightarrow{P_{\wt{C}}, P_{\wt{C}''}} R'$ for a fixed $C''$.

Taking a look at the structure of $(\boldsymbol{Q}, \boldsymbol{R})$ in $E_{C, C'}$ the conditions of $\boldsymbol{Q} \xrightarrow{P_{\wt{C}}, P_{\wt{C}'}} \boldsymbol{R}$ guarantee we have $\boldsymbol{Q}^{(1)}$ matches $P_{\wt{C}}$ on either $\lfloor \frac{k-t}{2} \rfloor$ or $\lceil \frac{k-t}{2} \rceil$ of the tensored operators and $\boldsymbol{Q}^{(2)}$ matches $P_{\wt{C}'}$ similarly. Additionally, it guarantees that $\boldsymbol{Q}^{(1)}$ and $\boldsymbol{Q}^{(2)}$ are identity on all remaining operators intersecting supports with $P_{\wt{C}}$ and $P_{\wt{C}'}$ respectively. In total, this fixes $2(k-t)$ of the operators in $\boldsymbol{Q}$. Since $\abs{\boldsymbol{Q}^{(1)}}+\abs{\boldsymbol{Q}^{(2)}} = \ell$, there remains $\ell-k+t$ free non-identity operators outside of these. Note these may appear anywhere that is not already fixed, the only condition is they match in $\boldsymbol{Q}$ and $\boldsymbol{R}$ to cancel out. We can view the process of sampling an edge, conditioned on these $2(k-t)$ fixed operators, as simply sampling the free operators from $\binom{[2n-2(k-t)]}{\ell-k+t}$ and random operators from $\{X, Y, Z\}$ for each one.

In order for $\boldsymbol{Q} \xrightarrow{P_{\wt{C}}, P_{\wt{C}''}} R'$ to occur, we necessarily have that $\boldsymbol{Q}^{(2)}$ matches $P_{\wt{C}''}$ on at least $\lfloor \frac{k-t}{2} \rfloor$ operators. Any matches within $P_{\wt{C}'} \sqcap P_{\wt{C}''}$ we get for free since $\boldsymbol{Q}^{(2)}$ must already match, at least in the worst case. The $\lfloor \frac{k-t}{2} \rfloor - \abs{P_{\wt{C}'} \sqcap P_{\wt{C}''}}$ remaining must be randomly chosen as per our sampling. Note the number of such terms is $\lfloor \frac{k-t}{2} \rfloor - \abs{P_{C'} \sqcap P_{C''}} + t$ since $\abs{P_{C'} \sqcap P_{C''}} = \abs{P_{\wt{C}'} \sqcap P_{\wt{C}''}} +t$. Fix any choice of terms $T$, of which there are at most $\binom{k-t}{\lfloor \frac{k-t}{2} \rfloor}$ options. We union bound over all possible $T$. The desired condition is met given (1) $T$ is contained in the set of free operators and (2) all operators match those in $P_{C''}$. The probability of the latter is clearly $\frac{1}{3}^{\lfloor \frac{k-t}{2} \rfloor - \abs{P_{C'} \sqcap P_{C''}} + t}$. The former can be seen as $\binom{2n-2(k-t)-\gamma}{\ell-k+t-\gamma}/\binom{2n-2(k-t)}{\ell-k+t}$ letting $\gamma = \lfloor \frac{k-t}{2} \rfloor - \abs{P_{C'} \sqcap P_{C''}} + t$ and we can bound the total as $\left(\frac{\ell}{3n}\right)^{\lfloor \frac{k-t}{2} \rfloor - \abs{P_{C'} \sqcap P_{C''}} + t}$ given our binomial estimates \Cref{fact:binomest} or trivially $1$ if $\gamma \leq 0$.
\end{proof}

\section{Expansion Properties in our NLTS Construction}
\label{sec:expansion}

In this section, we prove the following expansion result for random hypergraphs as required in \Cref{sec:nlts}.

\begin{lemma}[Expansion of random re-encoding]
\label{lem:expansion}
Fix constants $k_0,d_0,L\geq1$. Let
$B\in\F_2^{r\times N}$ have nonzero rows of weight at most $k_0$ and
columns of weight at most $d_0$. Assume $r$ is divisible by $3$ and $M\leq r/3$.
Independently sample $\Delta+1$ uniformly random permutations of $[r]$,
and group the consecutive entries of each into triples. Let
$\mathbf A\in\F_2^{\Delta r/3\times r}$ consist of all triples in the
first $\Delta$ permutations, and let
$\mathbf R\in\F_2^{M\times r}$ consist of the first $M$ triples in the
last permutation.

For every sufficiently large fixed integer $\Delta > 0$ divisible by $4$, there is a
constant $\tau > 0$ such that, with high probability in $r$,
\begin{equation*}
    |\mathbf{A}v| \geq L\abs{v}\mcom \qquad \forall v\in\F_2^r\mcom
\end{equation*}
and
\begin{equation*}
    |\mathbf{J}^\top u| \geq \frac{\abs{u}}{2}\mcom \qquad \text{for }0<|u|\leq\tau r\mcom
       \qquad \mathbf{J}=\begin{bmatrix}\mathbf{A}\\\mathbf{R}\end{bmatrix}B\mper
\end{equation*}
\end{lemma}

\begin{proof}

For the first half of the result, we show the stronger inequality
\begin{align*}
    |\mathbf{A}v|\geq\frac{\Delta}{10}\abs{v}\mcom \qquad \forall v \in \F_2^r\mper
\end{align*}
Increasing $\Delta$ so that
$\Delta\geq10(L+1)$ proves the desired result. We identify each vector $v \in \F_2^r$ with its indicated set $S \subseteq [r]$ and let $\Gamma_{\mathbf A}(S)$ be the set of rows
intersecting $S$. There are $\Delta \abs{S}$ incidences between $S$ and these rows by construction. If $q$ rows meet $S$ exactly once, then
\begin{align*}
    q\geq2\abs{\Gamma_{\mathbf A}(S)}-\Delta \abs{S}\mper
\end{align*}
Consequently, since each singleton contributes weight to $\abs{A\mathbf{1}_S}$,
$|\Gamma_{\mathbf A}(S)|\geq3\Delta|S|/4$ immediately implies
$|\mathbf A\mathbf1_S|\geq\Delta|S|/2$. We first establish this neighborhood bound for small $S$. Fix $|S|=s$
and a set $U$ of $u$ row positions, with $u_\ell$ positions in partition
$\ell$. The probability that every row meeting $S$ belongs to $U$ is
\begin{align*}
    \Pr\sbra{\Gamma_{\mathbf A}(S)\subseteq U}
    &=\prod_{\ell=1}^{\Delta}
       \frac{\binom{3u_\ell}{s}}{\binom rs}
      \leq\prod_{\ell=1}^{\Delta}
         \left(\frac{3u_\ell}{r}\right)^s
      \leq\left(\frac{3u}{\Delta r}\right)^{\Delta s}\mper
\end{align*}
The last inequality is simply AM-GM. Fixing the size $u = 3\Delta s/4$ for $s\leq ar$ with a sufficiently small absolute
constant $a>0$, a union bound and standard binomial estimates give
\begin{align*}
    \Pr\sbra{\exists S,\ |S|=s \mid
           \abs{\Gamma_{\mathbf A}(S)} < 3\Delta s/4}
    &\leq\binom rs\binom{\Delta r/3}{3\Delta s/4}
                \left(\frac{9s}{4r}\right)^{\Delta s}
       \nonumber\\
    &\leq\left[
           4^\Delta\left(\frac sr\right)^{\Delta/4-1}
          \right]^s\mper
\end{align*}
Choosing $a > 0$ sufficiently small and taking $\Delta\geq8$ sufficiently large makes the bracketed equation at most $2^{-\Delta}$ for $s \leq ar$.
Since $\Delta/4-1\geq\Delta/8$, the bracket above is at most $2^{-\Delta}$ for $s\leq ar$.
For $s\leq\sqrt r$, it is at most $4^\Delta r^{-(\Delta/4-1)/2}$.
Summing the displayed failure bound first over $1\leq s\leq\sqrt r$ and then over $\sqrt r<s\leq ar$ therefore gives a probability bound polynomially small in $r$. Thus, with high probability, the neighborhood bound holds simultaneously for all $|S|\leq ar$.
For $|S|\geq ar$, we apply a concentration argument instead.

In one uniform partition, the expected number of triples meeting a fixed $s$-element set $S$ oddly is
\begin{align*}
    \frac r3\,
       \frac{s\binom{r-s}{2}+\binom s3}{\binom r3}
    =s\,
       \frac{(r-s)(r-s-1)+(s-1)(s-2)/3}{(r-1)(r-2)}
    \geq\frac s5\mcom
\end{align*}
for $r\geq8$. To see the inequality, the numerator in the last fraction is a quadratic in $s$ with minimum $r^2/4-r+2/3$, which is at least
$(r-1)(r-2)/5$ for $r\geq8$.
Thus, for $\mathbf F=|\mathbf A\mathbf1_S|$,
$\mathbb E\sbra{\mathbf F} \geq\Delta s/5$.

Our concentration argument is as follows. Expose the $\Delta$ permutations
one entry at a time, and let $\mathbf Z_j$ be the conditional expectation
of $\mathbf F$ after $j$ exposures. Two possible values of the next entry
can be coupled by transposing those two values in the unexposed
completion. This affects at most two triples, so the two conditional
expectations differ by at most two. Hence
$|\mathbf Z_j-\mathbf Z_{j-1}|\leq2$. We then invoke Azuma-Hoeffding.
\begin{lemma}[Azuma-Hoeffding inequality]
        Let $\{\Delta_i \coloneqq Z_{i} - Z_{i-1}\}_{i \in [m]}$ be a martingale difference sequence satisfying $\abs{\Delta_i} \leq 2$ almost surely for every $i \in [m]$. Then for all $t >0$ we have
        \begin{align*}
            \Pr\sbra{\abs{Z_m- Z_0} > t} \leq 2\exp\pbra{ - \frac{1}{8} \cdot \frac{t^2 }{ m}}\mper
        \end{align*}
    \end{lemma}

Apply the lemma to the martingale $\mathbf Z_j$, using at most $\Delta r$ exposures and $t=\Delta s/10$.
For each fixed $S$, this gives
\begin{equation*}
    \Pr\sbra{|\mathbf A\mathbf1_S|<\Delta s/10}
    \leq 2\exp\pbra{-\frac{\Delta s^2}{800r}}\mper
\end{equation*}
For $s\geq ar$, a union bound over all $2^r$ subsets is at most
\begin{equation*}
    2^{r+1}\exp\pbra{-\frac{\Delta a^2r}{800}}
    =2\exp\pbra{\pbra{\log 2-\frac{\Delta a^2}{800}}r}\mcom
\end{equation*}
which is exponentially small in $r$ provided the fixed constant $\Delta$ is sufficiently large that $\Delta a^2/800>\log 2$. Combining the small and large set bounds proves the desired inequality for all $v$ with high probability.

We now prove the second half of the lemma. Let $B_i$ be the support of row $i$ of $B$. Each $B_i$ is nonempty and intersects at most $k_0d_0$ row labels,
including its own label. This bound allows repeated or dependent rows.

Fix a set $T$ of $t$ row positions of $\mathbf J$. Before cancellation,
these rows contain $3t$ occurrences of row labels of $B$.
Form a graph on the occurrences, joining two whenever their corresponding
supports $B_i$ intersect. An isolated occurrence contributes at least one
coordinate to the mod $2$ sum of the selected equations. Coordinates
contributed by different isolated occurrences are distinct. Therefore,
if
\begin{align*}
    |\mathbf J^\top\mathbf1_T|<t/2\mcom
\end{align*}
then fewer than $t/2$ occurrences are isolated. More than $5t/2$
occurrences belong to components of size at least two. Taking a spanning
tree in each such component gives more than $5t/4$ edges, and hence a
forest with $e=\lceil5t/4\rceil$ edges.

Fix one candidate forest on the $3t$ occurrence positions. Root its
components and expose their row labels in parent-before-child order.
Within each partition the labels are sampled without replacement;
different partitions are independent. If $t\leq r/6$, every newly exposed
label has at least $r-3t\geq r/2$ possible values. At most $k_0d_0$ of them
intersect the already exposed parent's support. Thus the probability
that every edge of the fixed forest is present is at most
\begin{align*}
    \left(\frac{2k_0d_0}{r}\right)^e\mper
\end{align*}
This conditional bound holds after every history of exposed labels, so
multiplying it over the forest edges is valid despite sampling without
replacement.

There are at most $\binom{\binom{3t}{2}}e\leq(c_1t)^e$ candidate
forests, where $c_1$ is absolute. It follows that, for a constant
$c_2=c_2(k_0,d_0)$ and $t/r$ sufficiently small,
\begin{align*}
    \Pr\sbra{|\mathbf J^\top\mathbf1_T|<t/2}
       \leq\left(\frac{c_2t}{r}\right)^e
       \leq\left(\frac{c_2t}{r}\right)^{5t/4}\mper
\end{align*}
There are at most $(\Delta+1)r/3$ row positions in total. A union bound
over their $t$-subsets gives
\begin{equation*}
    \Pr\sbra{\exists T, \abs{T} = t \mid
        |\mathbf J^\top\mathbf1_T| < t/2}
    \leq\left[c_3\left(\frac tr\right)^{1/4}\right]^t\mcom
\end{equation*}
Choosing $\tau > 0$ small enough, this proves for all $t \leq \tau r$ the desired result via union bound.
\end{proof}
\section{Kikuchi Matrices and Non-commutative Sum-of-Squares}
\label{sec:kikuchi}

In this section, we show in a strong formal sense that our level-$\ell$ Kikuchi matrix for Hamiltonian $k$-$\XOR$ is captured by degree-$2\ell$ non-commutative Sum-of-Squares. More explicitly, we show for any $\ell \geq k/2$, the level-$\ell$ even-arity degree-regularized Kikuchi matrix $\wt{K}_{\bH_\calI}$ of \Cref{def:regularkikuchimatrix} for any Hamiltonian $k$-$\XOR$ instance $\calI = (\calH, \{(P_C, b_C)\}_{C \in \calH})$ satisfies:
\begin{equation*}
    \ncSoS_{2\ell}(\bH_\calI) \leq \frac{1}{2}+\norm{\wt{K}_{\bH_\calI}}{2}\mper
\end{equation*}

To show this, we rewrite the result \Cref{obs:kikuchirelaxation} in the format of a degree-$2\ell$ non-commutative Sum-of-Squares proof (and with degree-regularization included). We show this proof in the simple even case for clarity; the odd case can be adapted without changing the proof structure.

Recall the notation $\psi^{\odot \ell} \in ((\C^{2})^{\otimes n})^{\mathcal{P}_\ell(n)}$ to refer to the vector
\[
\psi^{\odot \ell} \coloneqq \bigoplus_{P \in \mathcal{P}_\ell(n)} P \ket{\psi}\mper
\]
Instead of defining this lifted vector, suppose we instead defined $\calQ \in \C^{\abs{\calP_\ell(n)}2^n \times 2^n}$ via
\begin{align*}
    \calQ \coloneqq \bigoplus_{P \in \mathcal{P}_\ell(n)} P \mper
\end{align*}
The proof of \Cref{obs:kikuchirelaxation} then verbatim yields
\begin{align*}
    \calQ^\dagger (A^*_{\bH_\calI} \otimes \Id_{2^n}) \calQ = \Delta \sum_{C \in \calH} b_C P_C\mper
\end{align*}
Now observe that $\lVert\wt{K}_{\bH_\calI}\rVert_2\Id - \Gamma^{-1/2}A^*_{\bH_\calI}\Gamma^{-1/2} \succeq 0$ by definition, so we can write $\lVert\wt{K}_{\bH_\calI}\rVert_2\Gamma - A^*_{\bH_\calI} = L^\dagger L$ for some matrix $L \in \C^{\calP_\ell(n) \times \calP_{\ell}(n)}$. Now define for any $R \in \calP_\ell(n)$
\begin{align*}
    M_R \coloneqq \sum_{Q \in \calP_\ell(n)} L(R, Q) Q\mcom
\end{align*}
so each $M_R$ is degree-$\ell$. From here we just observe
\begin{align*}
    \sum_{R \in \calP_\ell(n)} M_R^\dagger M_R &= \calQ^\dagger (L^\dagger L\otimes \Id_{2^n}) \calQ\\
    &=  \calQ^\dagger [(\lVert\wt{K}_{\bH_\calI}\rVert_2\Gamma - A^*_{\bH_\calI})\otimes \Id_{2^n}] \calQ\\
    &=  \lVert\wt{K}_{\bH_\calI}\rVert_2 \tr(\Gamma) \Id - \Delta \sum_{C \in \calH} b_C P_C\mper
\end{align*}
Plugging in $\tr(\Gamma) = 2\abs{\calH}\Delta$ and rearranging from here yields
\begin{align*}
    \pbra{\frac{1}{2}+\norm{\wt{K}_{\calH_\calI}}{2}} \Id - \bH_\calI = \frac{1}{2\abs{\calH}\Delta}\sum_{R \in \calP_\ell(n)} M_R^\dagger M_R\mper
\end{align*}
Since each $M_R$ is a degree-$\ell$ operator, this yields the desired operator inequality in degree-$2\ell$.

\end{document}